%% file: main.tex
\documentclass{article}

\usepackage[DIV=14]{typearea}
\input{\macrospath/ben-macros-preamble}
\input{\macrospath/ben-macros-paper}

\input{\macrospath/ben-macros-diagrams}

\input{\macrospath/ben-macros-nets}

\input{\macrospath/thm-environments}
\usepackage{microtype}
\usepackage{versions}
\excludeversion{SHORT}
\includeversion{LONG}

\begin{document}

\title{The Abstract Machinery of Interaction}

\author{Beniamino Accattoli\and Ugo Dal Lago\and Gabriele Vanoni}

\date{}

\maketitle

\begin{abstract}
  \input{00-Abstract}
\end{abstract}

\input{01-Introduction}

\input{02-Gentle}

\input{03-IAM}


\input{04-Properties_of_the_IALM}


\input{05-Correctness}

\input{06-Micro-Step_Refinement}

\input{07-Exhaustible}


\input{08-Improvements_Abstractly}
\input{09-Improvements_Concretely}
\input{10-Concrete_Correctness}


\input{11-Proof-nets_Comparison}

\input{12-Conclusion}

\bibliography{main}

\end{document}

%% file: latex-macros/ben-macros-preamble.tex
\usepackage{ifthen}

\newboolean{talk}
\setboolean{talk}{false}
\newboolean{paper}
\setboolean{paper}{false}

\newboolean{IEEEstyle}
\setboolean{IEEEstyle}{false}
\newboolean{lipicsstyle}
\setboolean{lipicsstyle}{false}
\newboolean{eptcsstyle}
\setboolean{eptcsstyle}{false}
\newboolean{sigplanstyle}
\setboolean{sigplanstyle}{false}
\newboolean{PACMPL}
\setboolean{PACMPL}{false}
\newboolean{acmartstyle}
\setboolean{acmartstyle}{false}

\newboolean{needstheorems}
\setboolean{needstheorems}{false}
\newboolean{withimages}
\setboolean{withimages}{false}
\newboolean{withproofs}
\setboolean{withproofs}{true}

\newboolean{french}
\setboolean{french}{false}

%% file: latex-macros/ben-macros-paper.tex
\ifthenelse{\boolean{talk}}{}{\usepackage[utf8]{inputenc}}
\ifthenelse{\boolean{PACMPL}}{}{
	\ifthenelse{\boolean{french}}{\usepackage[french]{babel}}{
	  \ifthenelse{\boolean{lipicsstyle}}{}{\usepackage[english]{babel}}}
	  }
\ifthenelse{\boolean{PACMPL}}{}{
	\usepackage{amssymb} }
\usepackage{amsmath}
\usepackage{graphicx}
\usepackage{bussproofs}
\usepackage{fancybox}
\usepackage{cmll}
\usepackage{stmaryrd}
\usepackage{multirow}
\ifthenelse{\boolean{IEEEstyle}}{
\usepackage[bookmarks={false}]{hyperref}}{
\usepackage{hyperref}}
\usepackage{proof}
\usepackage{hhline}
\usepackage{xspace}
\usepackage{booktabs}
\usepackage{subcaption}
\usepackage{wrapfig}
\usepackage{array}
\usepackage{arydshln}

\ifthenelse{\boolean{IEEEstyle}}{
\setboolean{needstheorems}{true}}{}

\ifthenelse{\boolean{eptcsstyle}}{
\setboolean{needstheorems}{true}}{}

\ifthenelse{\boolean{sigplanstyle}}{
\setboolean{needstheorems}{true}}{}

\ifthenelse{\boolean{needstheorems}}{
\input{\macrospath/thm-environments}}{}

\newcommand{\myproof}[1]{
\ifthenelse{\boolean{withproofs}}{#1}{}}

\newcommand{\withproofs}[1]{
\ifthenelse{\boolean{withproofs}}{#1}{}}

\newcommand{\withoutproofs}[1]{
\ifthenelse{\boolean{withproofs}}{}{#1}}

\newcommand{\tm}{t}
\newcommand{\tmtwo}{u}
\newcommand{\tmthree}{r}
\newcommand{\tmfour}{w}

\newcommand{\var}{x}
\newcommand{\vartwo}{y}
\newcommand{\varthree}{z}

\newcommand{\rootRew}[1]{\mapsto_{#1}}
\newcommand{\Rew}[1]{\rightarrow_{#1}}

\renewcommand{\to}{\Rew{}}

\newcommand{\toh}{\Rew{h}}

\newcommand{\symfont}[1]{\mathsf{#1}}

\newcommand{\lssym}{{\symfont{ls}}}
\newcommand{\gcsym}{{\symfont{gc}}}
\newcommand{\db}{{\symfont{dB}}}

\newcommand{\varsym}{{\symfont{var}}}

\newcommand{\ctxholep}[1]{[#1]}
\newcommand{\ctxhole}{\ctxholep{\cdot}}

\newcommand{\ctx}{C}
\newcommand{\ctxtwo}{D}

\newcommand{\ctxp}[1]{\ctx\ctxholep{#1}}

\newcommand{\hctx}{H}
\newcommand{\hctxtwo}{K}
\newcommand{\hctxthree}{G}
\newcommand{\hctxp}[1]{\hctx\ctxholep{#1}}

\newcommand{\nbvctxtwo}[1]{\nbvctxtwo{#1}}

\newcommand{\sctx}{S}
\newcommand{\sctxtwo}{\sctx'}

\newcommand{\sctxp}[1]{\ctxholep{#1}\sctx}

\newcommand{\defeq}{:=}

\newcommand{\grameq}{::=}

\newcommand{\isub}[2]{\{#1/#2\}}
\newcommand{\esub}[2]{[#1/#2]}

\newcommand{\togc}{\Rew{\gcsym}}

\newcommand{\tohl}{\multimap}
\newcommand{\tohldb}{\multimap_\db}
\newcommand{\tohlls}{\multimap_\lssym}
\newcommand{\tohlgc}{\multimap_\gcsym}
\newcommand{\tohllst}[1]{\multimap_{\lssym,#1}}
\newcommand{\tohlgcv}[1]{\multimap_{\gcsym,#1}}

\newcommand{\llbrace}{\{ \kern -0.27em \vert}
\newcommand{\rrbrace}{\vert \kern -0.27em \}}

\renewcommand{\l}{\lambda}
\newcommand{\ie}{\textit{i.e.}\xspace}
\newcommand{\eg}{\textit{e.g.}\xspace}
\newcommand{\ih}{\textit{i.h.}\xspace}
\newcommand{\fv}[1]{\symfont{fv}(#1)}

\newcommand{\red}[1]{{\color{red} {#1}}}
\newcommand{\blue}[1]{{\color{blue} {#1}}}

\newcommand{\ignore}[1]{}

\newcommand{\myinput}[1]{\ifthenelse{\boolean{withimages}}{\input{#1}}{}}

\newcommand{\levy}{{L{\'e}vy}\xspace}

\newcommand{\ax}{\mathsf{ax}}
\newcommand{\cut}{\mathsf{cut}}

\newcommand{\tens}{\otimes}
\newcommand{\lolli}{\multimap}

\newcommand{\der}{{\mathsf{d}}}
\newcommand{\contr}{{\mathsf{c}}}

\newcommand{\set}[1]{\{#1\}}

\newcommand{\nat}{\mathbb{N}}

\newcommand{\size}[1]{|#1|}

\ifthenelse{\boolean{PACMPL}}{\renewcommand{\state}{s}}{\newcommand{\state}{s}}
\newcommand{\statetwo}{{s'}}
\newcommand{\statethree}{s''}
\newcommand{\statefour}{q}

\newcommand{\exec}{\rho}

\newenvironment{varitemize}
{
\begin{list}{\labelitemi}
{\setlength{\itemsep}{0pt}
 \setlength{\topsep}{0pt}
 \setlength{\parsep}{0pt}
 \setlength{\partopsep}{0pt}
 \setlength{\leftmargin}{15pt}
 \setlength{\rightmargin}{0pt}
 \setlength{\itemindent}{0pt}
 \setlength{\labelsep}{5pt}
 \setlength{\labelwidth}{10pt}
}}
{
 \end{list} 
}

\newcounter{numberone}

\newenvironment{varenumerate}
{
\begin{list}{\arabic{numberone}.}
{
  \usecounter{numberone}
  \setlength{\itemsep}{0pt}
  \setlength{\topsep}{0pt}
  \setlength{\parsep}{0pt}
  \setlength{\partopsep}{0pt}
  \setlength{\leftmargin}{15pt}
  \setlength{\rightmargin}{0pt}
  \setlength{\itemindent}{0pt}
  \setlength{\labelsep}{5pt}
  \setlength{\labelwidth}{15pt}
}}
{
\end{list} 
}

\newcounter{numbertwo}

\newcommand{\hnf}[1]{\textsf{hnf}(#1)}

\newcommand{\lhnf}[1]{\textsf{lhnf}(#1)}

\newcommand{\letin}[3]{{\sf let}\ #1=#2\ {\sf in}\ #3}

\newcommand{\trpos}{logged position\xspace}
\newcommand{\trposs}{logged positions\xspace}

\newcommand{\TrPoss}{Logged Positions\xspace}
\newcommand{\tr}{logged\xspace}

\newcommand{\Log}{Log\xspace}
\newcommand{\Logs}{Logs\xspace}

\newcommand{\pntransl}[1]{#1^{\dagger}}
\newcommand{\mtype}{o}
\newcommand{\med}[2]{
($(#1)!.5!(#2)$)
}

\renewcommand{\ctxholep}[1]{\langle #1\rangle}
\newcommand{\ctxtwop}[1]{\ctxtwo\ctxholep{#1}}

\newcommand{\reflemma}[1]{Lemma~\ref{l:#1}}
\newcommand{\reflemmap}[2]{Lemma\ref{l:#1}.\ref{p:#1-#2}}

\newcommand{\refprop}[1]{Prop.~\ref{prop:#1}}
\newcommand{\refsect}[1]{Sect.~\ref{sect:#1}}

\newcommand{\refeq}[1]{(\ref{eq:#1})}

\newcommand{\refcoro}[1]{Corollary~\ref{coro:#1}}
\newcommand{\refdef}[1]{Def.~\ref{def:#1}}

\newcommand{\ls}{\lssym}

\newcommand{\lctx}[1]{\ctx_{#1}}
\newcommand{\lctxp}[2]{\lctx{#1}\ctxholep{#2}}

\newcommand{\octx}[1]{O_{#1}}

\renewcommand{\esub}[2]{[#1{\shortleftarrow}#2]}
\renewcommand{\isub}[2]{\{#1{\shortleftarrow}#2\}}

\newcommand{\expset}{\mathcal{L}}

\newcommand{\resm}{\psym}
\renewcommand{\resm}{\bullet}
\newcommand{\psym}{\mathsf{p}}
\newcommand{\qsym}{\mathsf{q}}
\newcommand{\rsym}{\mathsf{r}}
\newcommand{\lsym}{\mathsf{l}}

\newcommand{\lpos}{p}
\renewcommand{\lpos}{l}

\newcommand{\sizelpos}[1]{\size{#1}_{\exps}}
\newcommand{\exps}{\lpos}
\newcommand{\expstwo}{\exps'}

\newcommand{\polset}{\mathcal{D}}
\newcommand{\upp}{\blue{\uparrow}}
\newcommand{\downp}{\red{\downarrow}}

\newcommand{\tlog}{L}

\newcommand{\ste}{\tlog}
\newcommand{\stetwo}{\ste'}

\newcommand{\tape}{T}
\newcommand{\tapetwo}{\tape'}
\newcommand{\tapethree}{\tape''}
\newcommand{\stme}{\tape}
\newcommand{\stmetwo}{\stme'}
\newcommand{\stmethree}{\stme''}
\newcommand{\pol}{d}
\newcommand{\poltwo}{\pol'}

\newcommand{\iamev}{\rightarrow}
\newcommand{\iamevp}[1]{\iamev_{#1}}

\newcommand{\iamdvartwo}{\iamevp{\textsf{var}2}}
\newcommand{\iamdes}{\iamevp{\textsf{es}}}
\newcommand{\iamues}{\iamevp{\textsf{es}2}}
\newcommand{\iamuestwo}{\iamevp{\textsf{var}3}}
\newcommand{\relm}{{\blacktriangleright_{\db}}}
\newcommand{\rele}{\blacktriangleright_{\ls}}
\newcommand{\relgc}{{\blacktriangleright_{\gcsym}}}
\newcommand{\relf}{{\blacktriangleright}}

\newcommand{\nopolstate}[5]{(#1,#2,#4,#3,#5)}

\newcommand{\dstate}[4]{(\underline{\red{#1}},#2,#4,#3)}

\newcommand{\ustate}[4]{(#1,\underline{\blue{#2}},#4,#3)}

\newcommand{\dstatetab}[4]{\underline{\red{#1}} & #2 & #4 & #3 }

\newcommand{\ustatetab}[4]{#1 & \underline{\blue{#2}} & #4 & #3 }

\newcommand{\ndstatetab}[5]{\underline{\red{#1}} & #2 & #4 & #3 & #5}
\newcommand{\nustatetab}[5]{#1 & \underline{\blue{#2}} & #4 & #3 & #5}

\newcommand{\cons}{{\cdot}}

\newcommand{\IAM}{\mathrm{\l\mbox{-}IAM}}

\newcommand{\tomachhole}[1]{\rightarrow_{#1}}

\newcommand{\btsym}{\mathsf{bt}}
\newcommand{\tomachdotone}{\tomachhole{\resm 1}}
\newcommand{\tomachdottwo}{\tomachhole{\resm 2}}
\newcommand{\tomachvar}{\tomachhole{\varsym}}
\newcommand{\tomachbttwo}{\tomachhole{\btsym 2}}
\newcommand{\iamdap}{\tomachdotone}
\newcommand{\iamdlamone}{\tomachdottwo}
\newcommand{\iamdvar}{\tomachvar}
\newcommand{\iamdlamtwo}{\tomachbttwo}

\newcommand{\argsym}{\mathsf{arg}}
\newcommand{\tomachdotthree}{\tomachhole{\resm 3}}
\newcommand{\tomachdotfour}{\tomachhole{\resm 4}}
\newcommand{\tomacharg}{\tomachhole{\argsym}}
\newcommand{\tomachbtone}{\tomachhole{\btsym 1}}
\newcommand{\iamuapltwo}{\tomachdotthree}
\newcommand{\iamulam}{\tomachdotfour}
\newcommand{\iamuaplone}{\tomacharg}
\newcommand{\iamuapr}{\tomachbtone}

\newcommand{\tomachbttwodec}{\tomachhole{\btsym 2, \lpos}}
\newcommand{\tomachbtonedec}{\tomachhole{\btsym 1,\lpos}}

\newcommand{\toiam}{\rightarrow_{\IAM}}

\newcommand{\stempty}{\epsilon}

\newcommand{\sizee}[1]{\sizelpos{#1}}
\newcommand{\la}[1]{\lambda #1.}
\newcommand{\expsn}{\ste_n}

\newcommand{\expsind}[1]{\ste_{#1}}

\newcommand{\outs}[1]{\mathsf{out}(#1)}
\newcommand{\outsi}[2]{\mathsf{out}_{#1}(#2)}

\newcommand{\sem}[2]{\llbracket#1\rrbracket_{#2}}

\newcommand{\exstates}{\mathcal{E}}

\newcommand{\midd}{\; \; \mbox{\Large{$\mid$}}\;\;}

\newcommand{\rel}{\mathcal{R}}



%% file: latex-macros/thm-environments.tex
\usepackage{amsthm}

    \newtheorem{theorem}{Theorem}[section]
    \newtheorem{lemma}[theorem]{Lemma}
    \newtheorem{corollary}[theorem]{Corollary}
    \newtheorem{proposition}[theorem]{Proposition}
    \newtheorem{definition}[theorem]{Definition}

    \newtheorem{example}[theorem]{Example}

%% file: latex-macros/ben-macros-diagrams.tex
\usepackage{tikz}
\usetikzlibrary{matrix,arrows}
\usetikzlibrary{decorations.shapes}
\usetikzlibrary{decorations.text}
\usetikzlibrary{decorations.pathmorphing}
\usetikzlibrary{decorations.markings}
\usetikzlibrary{arrows.meta}


%% file: latex-macros/ben-macros-nets.tex
\usepackage{tikz}
\usetikzlibrary{calc}
\usetikzlibrary{backgrounds}
\usetikzlibrary{decorations.shapes}
\usetikzlibrary{decorations.text}
\usetikzlibrary{decorations.pathmorphing}
\usetikzlibrary{decorations.markings}
\usetikzlibrary{matrix}
\usetikzlibrary{decorations.pathreplacing}
\usetikzlibrary{arrows}
\usetikzlibrary{positioning}
\usepackage[nofancy]{tikz-inet}

\setboolean{withimages}{true}

\tikzset{%
  ocenter/.style={baseline={([yshift=-.5ex, xshift=-.5ex]current bounding box)}},  
  nospace/.style={inner sep= 0pt},
  etic/.style={inner sep= 0.5pt, fill=white, anchor= center},
  letic/.style={inner sep= 0.7pt, fill=white, anchor= center,circle,draw=black},
  port/.style={inner sep= 0.3pt, anchor= center,circle,draw=black,fill=black, minimum size = 0.3pt},
  eport/.style={inner sep= 0.8pt, anchor= center,circle,draw=cyan,fill=white, minimum size = 0.8pt},
  mport/.style={inner sep= 0.8pt, anchor= center, circle,draw=white,fill=brown, minimum size = 0.8pt, solid, 
line width=0.10ex},
  eportn/.style={inner sep= 0.8pt, anchor= center,circle,draw=blue,fill=white, minimum size = 0.8pt},
  mportn/.style={inner sep= 0.8pt, anchor= center, circle,draw=red,fill=red, minimum size = 0.8pt, solid, 
line width=0.10ex},
  nopol/.style={->, shorten <=0.5pt, shorten >=0.5pt, draw=gray, line width=0.18ex},
  nopolrev/.style={<-, shorten <=1pt, shorten >=1pt, draw=gray, line width=0.18ex},
  nopolgen/.style={preaction={decorate},decoration={markings,mark=at position .5 with {\draw [shorten >=0pt, shorten <=0pt,draw=gray,-](0pt,3pt) -- (0pt,-3pt);}},->, shorten <=0.5pt, shorten >=0.5pt, draw=gray, line width=0.18ex},
  nopolrevgen/.style={preaction={decorate},decoration={markings,mark=at position .5 with {\draw [shorten >=0pt, shorten <=0pt,draw=gray,-](0pt,3pt) -- (0pt,-3pt);}},<-, shorten <=1pt, shorten >=1pt, draw=gray, line width=0.18ex},
  outEdge/.style={<-, shorten >=0.5pt, shorten <=0.5pt, draw=blue, line width=0.18ex},
  inpEdge/.style={->, densely dotted, shorten >=0.5pt, shorten <=0.5pt, draw=red, line width=0.18ex, overlay},
  inpEdgeGen/.style={preaction={decorate},decoration={markings,mark=at position .5 with {\draw [shorten >=0pt, shorten <=0pt,draw=red,-](0pt,3pt) -- (0pt,-3pt);}},->, densely dotted, shorten >=0.5pt, shorten <=0.5pt, draw=red, line width=0.18ex, overlay},
    eprinc/.style={postaction={decorate},decoration={markings,mark=at position .4 with {\node [inner sep= 0.5pt, anchor= center,circle,draw=cyan, fill=cyan, minimum size = 0.8pt, solid, line width=0.10ex]{};}}},
    mprinc/.style={postaction={decorate},decoration={markings,mark=at position .4 with {\node [inner sep= 0.5pt, anchor= center,circle,draw=brown, fill=brown, minimum size = 0.8pt, solid, line width=0.10ex]{};}}},
    mprincpar/.style={postaction={decorate},decoration={markings,mark=at position .6 with {\node [inner sep= 0.5pt, anchor= center,circle,draw=brown, fill=brown, minimum size = 0.8pt, solid, line width=0.10ex]{};}}},
    pure/.style={->, shorten <=1pt, shorten >=1pt, draw=brown, line width=0.18ex},
    pureRev/.style={<-, shorten <=1pt, shorten >=1pt, draw=brown, line width=0.18ex},
    exppure/.style={->, shorten <=1pt, shorten >=1pt, draw=cyan, line width=0.18ex, densely dotted},
    exppureRev/.style={<-, shorten <=1pt, shorten >=1pt, draw=cyan, line width=0.18ex, densely dotted},    genexppure/.style={preaction={decorate},decoration={markings,mark=at position .5 with {\draw [shorten >=0pt, shorten <=0pt,draw=cyan,-](0pt,3pt) -- (0pt,-3pt);}},->, shorten <=1pt, shorten >=1pt, draw=cyan, line width=0.18ex, densely dotted},
    enaryedge/.style={preaction={decorate},decoration={markings,mark=at position .5 with {\draw [shorten >=0pt, shorten <=0pt,draw=cyan,-](0pt,3pt) -- (0pt,-3pt);}}},
    eprincn/.style={postaction={decorate},decoration={markings,mark=at position .4 with {\node [inner sep= 0.7pt, anchor= center,circle,draw=blue, fill=blue, minimum size = 0.8pt, solid, line width=0.10ex]{};}}},
    mprincn/.style={postaction={decorate},decoration={markings,mark=at position .6 with {\node [inner sep= 0.7pt, anchor= center,circle,draw=red, fill=white, minimum size = 0.8pt, solid, line width=0.10ex]{};}}},
    mprincnpar/.style={postaction={decorate},decoration={markings,mark=at position .4 with {\node [inner sep= 0.7pt, anchor= center,circle,draw=red, fill=white, minimum size = 0.8pt, solid, line width=0.10ex]{};}}},
    puren/.style={->, shorten <=1pt, shorten >=1pt, draw=red, line width=0.18ex},
    purenRev/.style={<-, shorten <=1pt, shorten >=1pt, draw=red, line width=0.18ex},
    exppuren/.style={->, shorten <=1pt, shorten >=1pt, draw=blue, line width=0.18ex, densely dotted},
    exppurenRev/.style={<-, shorten <=1pt, shorten >=1pt, draw=blue, line width=0.18ex, densely dotted},    genexppuren/.style={preaction={decorate},decoration={markings,mark=at position .5 with {\draw [shorten >=0pt, shorten <=0pt,draw=blue,-](0pt,3pt) -- (0pt,-3pt);}},->, shorten <=1pt, shorten >=1pt, draw=blue, line width=0.18ex, densely dotted},
    enaryedge/.style={preaction={decorate},decoration={markings,mark=at position .5 with {\draw [shorten >=0pt, shorten <=0pt,draw=blue,-](0pt,3pt) -- (0pt,-3pt);}}},
  jboxline/.style={draw= gray,rounded corners, line width=0.20ex, overlay},
  exboxline/.style={draw= gray,line width=0.20ex, overlay},
  noboxline/.style={draw= white,rounded corners, line width=0ex},
  net/.style={draw=gray,inner sep=2pt,thick,ellipse, anchor=center, font=\scriptsize},
  inductiveTr/.style={draw=black!50, minimum size=0.9cm},
  inductiveTrSmall/.style={draw=black!50, minimum size=0.6cm},
  every label/.style={label distance = 1pt, font=\scriptsize, inner sep= 1pt},  
  every node/.style={font=\scriptsize }
}

\newcommand{\stalt}{22pt}
\newcommand{\stlar}{15pt}
\newcommand{\hstalt}{\stalt/2}
\newcommand{\hstlar}{\stlar/2}

\input{\macrospath/graphical-macros/generic-links}

\input{\macrospath/graphical-macros/links-MELL}

\newcommand{\gdotsname}[4]{
\node at ($(#1)!.5!(#2)$) [#3](#4){\tiny $\ldots$};
}

%% file: latex-macros/graphical-macros/generic-links.tex
\newcommand{\lUnarySymbol}[5]{
\node at ($(#2.center) ! .5 ! (#1.center)$) [#5] (#3){ #4};
}

\newcommand{\lUnary}[5]{
\lUnarySymbol{#1}{#2}{#3}{#4}{etic}

\draw[#5] (#2) to (#3);
\draw[#5] (#3) to (#1);
}

\newcommand{\lBinSymbol}[6]{
\node at ($(#2.center) ! .5 ! (#3.center) ! .5 ! (#1.center)$) [#6] (#4){#5};
}

\newcommand{\lBinEdgesAbove}[5]{
\draw[#4, in=150, out=-90] (#1) to (#3);
\draw[#5, in=30, out=-90] (#2) to (#3);
}

\newcommand{\lBinaryWithEdgeAboveTypes}[8]{
\lBinSymbol{#1}{#2}{#3}{#4}{#5}{etic}

\lBinEdgesAbove{#2}{#3}{#4}{#7}{#8}

\draw[#6] (#4) to (#1);
}

\newcommand{\lBinary}[6]{
\lBinaryWithEdgeAboveTypes{#1}{#2}{#3}{#4}{#5}{#6}{#6}{#6}
}

\newcommand{\sepbox}{2pt}

\newcommand{\boxnodes}[4]{
\node at (#1.center)[left = #2,nospace](#1so){};
\node at (#1.center)[right = #3,nospace](#1se){};
\node at (#1se.center)[above=#4,nospace](#1ne){};
\node at (#1ne-|#1so)[nospace](#1no){};}

\newcommand{\boxline}[2]{
\draw[#2](#1.center) -- (#1se.center) -- (#1ne.center) -- (#1no.center) -- (#1so.center)--(#1.center);}

\newcommand{\abox}[5]{
\boxnodes{#1}{#3}{#4}{#5}
\boxline{#1}{#2line}}



%% file: latex-macros/graphical-macros/links-MELL.tex
\newcommand{\lder}[3]{
\lUnary{#1}{#2}{#3}{$\der$}{nopol}
}

\newcommand{\lcontr}[4]{
\lBinary{#1}{#2}{#3}{#4}{$\contr$}{nopol}}

%% file: 00-Abstract.tex
This paper revisits the Interaction Abstract Machine (IAM), a machine 
based on Girard's Geometry of Interaction, introduced by Mackie and Danos 
\& Regnier. It is an unusual machine, not relying on environments, presented 
on linear logic proof nets, and whose soundness proof is convoluted and passes 
through various other formalisms. Here we provide a new direct proof of its 
correctness, based on a variant 
of Sands's improvements, a natural notion of bisimulation. Moreover, our proof is 
carried out on a new presentation of the IAM, defined as 
	a machine acting directly on $\lambda$-terms, rather than on
	linear logic proof nets.


%% file: 01-Introduction.tex
\section{Introduction}

The advantage, and at the same time the drawback, of the $\l$-calculus is its 
distance from low-level, implementative details. It comes with just one rule, 
$\beta$-reduction, and with no indications about how to implement it on 
low-level machines. It is an 
advantage when \emph{reasoning} about programs expressed as $\l$-terms. It is a 
drawback, instead, when one wants to \emph{implement} the $\l$-calculus, or to do 
complexity analyses, because $\beta$-steps are far from being atomic 
operations. 
In particular, terms can grow exponentially with the number of $\beta$-steps, a 
degeneracy known as \emph{size explosion}, which is why $\beta$-reduction 
cannot be reasonably implemented, at least if one sticks to an explicit 
representation of $\lambda$-terms.

\paragraph{Environment Machines} Implementations solve this issue by evaluating the 
$\l$-calculus up to \emph{sharing of sub-terms}, where sharing is realized 
through a data structure called \emph{environment}, collecting the sharing 
annotations generated by the machine during the execution, one for each 
encountered $\beta$-redex. For common \emph{weak} evaluation strategies (\ie 
that do not inspect in scope of $\lambda$-abstractions) such as 
call-by-name/value/need, the number of $\beta$-steps is a reasonable time cost 
model 
\cite{DBLP:conf/fpca/BlellochG95,DBLP:conf/birthday/SandsGM02,DBLP:journals/tcs/LagoM08}.
 Environment machines---whose most famous examples are Landin's 
SECD~\cite{landin_correspondence_1965}, Felleisen and Friedman's 
CEK~\cite{felleisen_control_1986} or Krivine's
KAM~\cite{krivine_call-by-name_2007}---can be extended to open terms and optimized in such a way that they 
run within a linear overhead with respect to the number of $\beta$-steps 
\cite{DBLP:conf/lics/AccattoliC15,DBLP:conf/fsen/AccattoliG17}. Said 
differently, they \emph{respect} the time cost model (see~\cite{DBLP:journals/entcs/Accattoli18} for an overview). For space, the situation is different. 
Only very recently the problem has been 
tackled~\cite{cbv_reasonable} and some preliminary and limited results have 
appeared. Then, environment machines store information for every 
$\beta$-step, therefore using space linear in time, which is \emph{the worst 
possible use of space}\footnote{On sequential models space cannot exceed time, 
as one needs a unit of time to use a unit of space.}.

\paragraph{Beyond Environments} In practice, frameworks based on the $\l$-calculus 
are invariably implemented using environments. Nonetheless, the lack of a fixed 
execution schema for the $\l$-calculus leaves open, in theory, the possibility 
of alternative implementation schemes. The theory of linear logic indeed 
provides a completely different style of abstract machines, rooted in Girard's 
Geometry of Interaction~\cite{girard_geometry_1989} (shortened to GoI in the following). These 
GoI machines were pioneered by Mackie and Danos \& Regnier in the nineties 
\cite{mackie_geometry_1995,danos_reversible_1999}. The basic idea is that the 
machine does not use environments, while it keeps track of information that 
allows \emph{retrieving} previous $\beta$-redexes, by using a data structure 
called \emph{token}, saving information about the history of the computation. 
The key point is that the token does not store  information about every single 
$\beta$-redex, thus disentangling space-consumption from time-consumption. In 
other words, GoI machines are good candidates for space-efficient 
implementation schemes, as first shown by Sch\"opp and coauthors 
\cite{bllspace,dal_lago_computation_2016}. The price to pay is that the machine 
wastes a 
lot of time to retrieve $\beta$-redexes, so that time is sacrificed for space. 
The same, however, happens with space-sensitive Turing machines.

\paragraph{The Interaction Abstract Machine} The original GoI machine is the 
\emph{Interaction Abstract Machine} (IAM). It was developed at the same time by Mackie and Danos \& Regnier, and its first appearance is in a paper by Mackie in 1995 \cite{mackie_geometry_1995}, dealing with implementations. Danos and Regnier study it in two papers, one in 1996 together with Herbelin
\cite{DBLP:conf/lics/DanosHR96}, where it is dealt with quickly, and
its implementation theorem (or correctness\footnote{The result that an abstract machine implements a strategy is sometimes called \emph{correctness} of the machine. We prefer to avoid such a terminology, because it suggests the existence of a dual \emph{completeness} result, that is never given because already contained in the statement of correctness. We then simply talk of an \emph{implementation theorem}.}) is proved via game semantics, and one
 by themselves \cite{danos_reversible_1999}, published only in 1999 but reporting work dating back of a few years, dedicated to the IAM and to
 an optimization based on a fine analysis of IAM
runs. These papers differ on many details but they all formulate the
IAM on linear logic proof-nets as a reversible, bideterministic
automaton.


In \cite{DBLP:conf/lics/DanosHR96}, Danos, Herbelin, and Regnier prove
that the IAM implements \emph{linear head evaluation}
$\rightarrow_{\mathtt{LHE}}$ (shortened to LHE), a refinement of head
evaluation, arising from the linear logic decomposition of the
$\l$-calculus.
Their proof of the implementation theorem for the IAM---the only one in the literature---is
indirect and rooted in game semantics, as it follows from a sequence
of results relating the IAM to AJM games, AJM games to HO games, HO
games to another abstract machine, the PAM, and finally the PAM to
LHE.  Moreover, the proof is technical, the main ingredients unclear,
and it is not as neat as for environment machines.


\paragraph{New Proof of the Implementation Theorem} The main contribution of this paper is an 
alternative proof of the implementation theorem for the IAM, which is 
independent of game semantics and of other abstract machines. Our proof is 
direct and based on a simple tool, namely a variation over 
Sands' 
\emph{improvements} \cite{DBLP:journals/toplas/Sands96}, a natural notion of 
bisimulation. 

The implementation theorem of GoI machines amounts to showing that their result 
is an adequate  and sound semantics for LHE, that is, it is invariant by LHE 
(soundness) and it exists if and only if LHE terminates (adequacy). The key 
point for soundness is that---in contrast to the study of environment 
machines---steps of the GoI machine are not mapped to LHE steps, because the 
GoI computes differently. What is shown is that if $\tm 
\rightarrow_{\mathtt{LHE}} \tmtwo$ then the run of the machine on $\tm$ is 
'akin' to the run on $\tmtwo$, and they produce the same result---see 
\refsect{soundness-intro} for more details.

In our proof, 'akin' is naturally interpreted as being \emph{bisimilar}. An \emph{improvement} is a bisimulation asking that the run on $\tmtwo$ is no longer than the run on $\tm$. Building on such a quantitative refinement, we prove adequacy.

The proof of our implementation theorem is arguably conceptually simpler than Danos, Herbelin, and Regnier's. Of course, their deep connection with game semantics is an important contribution that is not present here. We believe, however, that having independent and simpler proof techniques is also valuable.



\paragraph{The Lambda Interaction Abstract Machine} The second contribution of 
the 
paper is a formulation of the IAM as a machine acting directly
on $\lambda$-terms rather than on linear logic proof nets. Our proof might also 
have been carried out on proof nets, but we prefer switching to $\l$-terms for 
two reasons. First, manipulating terms rather than proof nets is easier and 
less error-prone for the technical development. Second, we aim at minimizing 
the background required for understanding the IAM, and so doing we remove 
any explicit reference to linear logic and graphical syntaxes.

The starting point of our \emph{Lambda} Interaction 
Abstract Machine ($\IAM$) is seeing a position in the code $\tm$ 
(what is usually the position of the token on the proof net representation of 
$\tm$) as a pair $(\tmtwo, \ctx)$ of a sub-term $\tmtwo$ and a context $\ctx$ 
such that $\ctxp\tmtwo = \tm$. These positions are simply a readable 
presentation of pointers\footnote{For the acquainted reader, they play a role 
akin to the initial labels in \levy's labeled $\l$-calculus, itself having deep connections 
with the IAM~\cite{DBLP:conf/lics/AspertiDLR94}.}. 

The main novelty of the new presentation is that some of the exponential transitions on proof nets are packed together in macro transitions. The shape of our transitions makes a sort of backtracking mechanism more evident. \emph{Careful}: that the IAM rests on backtracking is the key point of Danos and Regnier in \cite{danos_reversible_1999}, and therefore it is not a novelty in itself. What is new is that such a mechanism is already visible at the level of transitions, while on proof nets it requires a sophisticated analysis of runs.

It may be argued that linear logic provides a useful conceptual 
framework for the GoI. While this is undeniable, we are trying to show 
that such a framework is however not needed, and that an alternative 
presentation provides other useful intuitions---the two presentations give 
different insights, and thus complement each other. The easy correspondence 
between the two is stated in \refsect{pn}. 



\paragraph{More About the $\IAM$} The original papers on the IAM \cite{mackie_geometry_1995,DBLP:conf/lics/DanosHR96,danos_reversible_1999} differ on many points. Here we follow \cite{DBLP:conf/lics/DanosHR96}, modelling the $\IAM$ on the call-by-name translation of the $\l$-calculus in 
linear logic and considering only the path/runs starting on the distinguished conclusion corresponding to the output of 
the net/term. This is natural for terms, and also along the lines of 
how AJM games interpret terms. Similarly to AJM games, then, our GoI 
semantics is sound also for open terms with respect to erasing steps.
 
An original point of our work is the identification of a 
 new invariant of the $\IAM$---probably of independent interest---based on what 
 we call \emph{exhaustible states}. Informally, a state of the
 $\IAM$ is exhaustible if its token can be emptied
 in a certain way, somehow mimicking the computation which leads to the
 state itself. The invariant is an essential ingredient of the proof of soundness.

\paragraph{This Paper in Perspective} This paper is one of the last chapters of 
a long-time endeavor by the authors
directed at understanding complexity measures and implementation schemas for 
the $\l$-calculus. We provide a new proof technique for GoI implementation theorems not relying on game semantics, together with an new presentation of 
the original machine by Mackie and Danos \& Regnier not relying on linear logic. The aim is to set the ground for 
a formal, robust, and systematic study of GoI machines and their complexity, 
while at the same time shrinking to the minimum the required background. A 
further motivation behind our work is the desire to make the study of GoI 
machines easier to formalize in proof assistants, as proof nets are 
particularly cumbersome in that respect.

\paragraph{Related Work on GoI} This is certainly \emph{not} the first paper on the 
GoI and the $\lambda$-calculus.  Indeed, the literature on the topic
and its applications is huge, and goes from Girard's original
papers~\cite{girard_geometry_1989}, to Abramsky \emph{et al}'s
reformulation using the
INT-construction~\cite{abramsky_geometry_2002}, Danos and Regnier's
using path algebras~\cite{danos_local_1993}, Ghica's applications to
circuit synthesis~\cite{ghica_geometry_2007}, together with extensions
by Hoshino, Muroya, and Hasuo to languages with various kinds of
effects \cite{hoshino_memoryful_2014}, and Laurent's extension to the
additive connectives of linear
logic~\cite{DBLP:conf/tlca/Laurent01}. In all these cases, the GoI
interpretation, even when given on $\lambda$-terms, goes
\emph{through} linear logic (or symmetric monoidal categories) in an
essential way.  The only notable exceptions are perhaps the recent
contributions by Sch\"opp on the relations between GoI, CPS, and
defunctionalization~\cite{Schopp14,Schopp15} in which, indeed, some
deep relations are shown to exist between GoI and classic tools in the
theory of $\lambda$-calculus. Even there, however, GoI is seen as
obtained through the 
INT-construction~\cite{joyal_street_verity_1996,abramsky_geometry_2002},
although applied to a syntactic category of terms.

The GoI has also been studied in relationship with implementation of
functional languages, by Gonthier, Abadi and Levy as a proof
methodology in the study of optimal
implementations~\cite{gonthier_geometry_1992}, and by Mackie with his
GoI machine for \textsf{PCF}~\cite{mackie_geometry_1995} and G\"odel
System $\mathsf T$ \cite{DBLP:conf/wollic/Mackie17}.  Recently, the
space-efficiency studied by Dal Lago and
Sch\"opp~\cite{dal_lago_computation_2016} has been exploited by Mazza
in \cite{DBLP:conf/csl/Mazza15} and, together with Terui, in
\cite{DBLP:conf/icalp/MazzaT15}. Dal Lago and coauthors have also
introduced variants of the IAM acting on proof nets for a number of
extensions of the $\l$-calculus
\cite{DBLP:conf/csl/LagoFHY14,DBLP:conf/lics/LagoFVY15,DBLP:conf/popl/LagoFVY17,DBLP:conf/lics/LagoTY17}. Curien
and Herbelin study abstract machines related to game semantics and the
IAM in \cite{CurienHFlops,DBLP:journals/corr/abs-0706-2544}.  Muroya
and Ghica have recently studied the GoI in combination with rewriting
and abstract machines in \cite{DBLP:conf/csl/MuroyaG17}. The already
cited works by Sch\"opp~\cite{Schopp14,Schopp15} highlight how GoI can
be seen as an optimized form of CPS transformation, followed by
defunctionalization.


\paragraph{Related Work on Environment Machines} 
The \emph{time} efficiency of environment machines has been recently
closely scrutinized. Before 2014, the topic 
had 
been mostly neglected---the only two counterexamples being Blelloch and Greiner in 
1995 \cite{DBLP:conf/fpca/BlellochG95} and Sands, Gustavsson, and Moran in 
2002 \cite{DBLP:conf/birthday/SandsGM02}. Since 2014---motivated by advances by 
Accattoli and Dal Lago on time cost models for the $\l$-calculus 
\cite{accattoli_leftmost-outermost_2016}---Accattoli and co-authors have 
explored time analyses of environment machines from different angles 
\cite{DBLP:conf/icfp/AccattoliBM14,DBLP:conf/fsen/AccattoliG17,DBLP:conf/ppdp/AccattoliB17,DBLP:conf/ppdp/AccattoliCGC19}.


%% file: 02-Gentle.tex
\section{A Gentle Introduction to the Geometry of 
	Interaction}\label{sec:example}
This section is an informal introduction to Girard's
Geometry of Interaction as implemented by the $\IAM$,
the abstract machine we are introducing in this paper.  Many details are left out, and shall be covered in the next sections.
\medskip

\paragraph{Preliminaries.} 
The $\IAM$ implements \emph{head evaluation}, the simple reduction defined as:
\[
\la{\var_1} \ldots \la{\var_k} (\la\vartwo \tm) \tmtwo \tmthree_1 \ldots 
\tmthree_h \ \ \toh \ \ \la{\var_1} \ldots \la{\var_k} \tm \isub\vartwo \tmtwo 
\tmthree_1 \ldots \tmthree_h.
\]
The meaning of ``implement'' is explained a bit here, and more extensively in 
\refsect{soundness-intro}. Moreover, the $\IAM$ rather implements a linear 
variant of $\toh$, but for now the difference does not matter. 

An essential point is that the initial code $\tm$ of the machine never changes. 
The $\IAM$ only moves over it, in a local way, with no rewriting of the code 
and without ever substituting terms for variables. The current position in the 
code $\tm$ is represented as a pair $(\tmtwo,\ctx)$ where $\ctx$ is a context 
(that is, a term with a hole)
and $\ctx\ctxholep\tmtwo=\tm$. 

Beyond the current position, the state of the machine 
also contains the \emph{token}, which is 
given by two stacks, called \emph{log} and \emph{tape} respectively. The log 
is dedicated exclusively to the internal functioning of the machine. The tape, 
additionally, has an input/output role. Their functioning shall be explained 
soon.

Environment machines are either weak (that is, 
\emph{never} enter abstractions) or strong (they enter into \emph{all} 
abstractions). In contrast, the $\IAM$ is \emph{incrementally strong}, \ie, it has a finer mechanism for entering into 
\emph{some} abstractions. The number of head abstractions that a run of 
the $\IAM$ can cross, called the \emph{depth} of the run, is specified at the 
beginning by the content of the tape, coded in unary: depth $n$ is represented 
with $n$ occurrences of the distinguished symbol $\resm$. Note the difference 
with environment machines: once the code $\tm$ is fixed, such machines have 
only one initial state, while the $\IAM$ has a \emph{family} of initial states, 
one for each depth.

A tricky point is that, given $\tm$, the $\IAM$ does not compute the whole head 
normal form $\hnf\tm$ of $\tm$, but only the head variable of $\hnf\tm$. This 
is very much in accordance with the idea of head evaluation, in which the 
arguments of the head variable are never touched. More about this shall be 
explained in \refsect{soundness-intro}.

Before giving an example run, we need one last concept. Beyond the current position and the token, a $\IAM$ 
state has a \emph{direction}, $\downp$ or $\upp$. When the direction is downwards ($\downp$), the machine looks 
for the head variable of the subterm. When it is upwards ($\upp$), the 
$\IAM$ looks for the argument the found head variable would be substituted for 
under head evaluation (explanations below).

\medskip

\paragraph{An Example of $\IAM$ run.} Suppose one wants to evaluate the term 
$\tm\defeq((\la z\la x x)w)(\la y y)$,
 whose head normal form is $\la y y$. We know that the head variable $\vartwo$ 
 in $\hnf\tm$ is under one abstraction. 
 Then, to find it, we have to run the $\IAM$ at depth 1, that is, starting from
 $(\red\tm,\ctxhole,\epsilon,\resm,\downp)$, that is, on position 
 $(\tm,\ctxhole)$, with empty log, with $\resm$ on the tape, and direction 
 $\downp$. We expect $\vartwo$ as the result 
 of the run.
 

Let's then consider the first four transitions of such a computation, that perform a visit of the
leftmost branch of $\tm$, called the \emph{spine}, until a variable is found.
\[{\footnotesize
\begin{array}{c|c|c|c|c}
\mathsf{Sub}\mbox{-}\mathsf{term} & \mathsf{Context} & \mathsf{\Log} & 
\mathsf{Tape} & \mathsf{Dir}
\\
		\cline{1-5}
\ndstatetab{((\la z\la x x)w)(\la y y)} {\ctxhole} {\resm} {\epsilon} 
\downp\\
\ndstatetab{(\la z\la x x)w} {\ctxhole(\la y y)} {\resm\cdot\resm} {\epsilon} 
\downp\\
\ndstatetab{\la z\la x x} {(\ctxhole w)(\la y y)} {\resm\cdot\resm\cdot\resm} 
{\epsilon} 
\downp\\
\ndstatetab{\la x x} {((\la z\ctxhole) w)(\la y y)} {\resm\cdot\resm} 
{\epsilon} \downp\\
\ndstatetab{x} {((\la z\la x\ctxhole) w)(\la y y)} {\resm} {\epsilon} 
\downp
\end{array}}
\]
Note the \emph{pushing and popping of $\resm$}:
  one of the tasks of the tape is to account for the
  abstractions and applications encountered along the spine:
  the symbol $\resm$ is pushed on applications, and
  pulled on abstractions (when the direction is $\downp$), so that the crossing of a $\beta$-redex 
  leaves the tape unchanged. We shall say that the \emph{$\IAM$ searches up to $\beta$-redexes}.
  Note also that, contrary to environment machines,
  arguments of the encountered $\beta$ redexes are \emph{not} saved,
  this way saving space, and disentangling space from time.

Once in the state $\dstate{x} {((\la z\la x\ctxhole) w)(\la y y)} {\resm} 
{\epsilon} $, the $\IAM$ switches to phase $\upp$, and starts to 
check whether $\var$ would be substituted during head evaluation. In
the KAM, it is enough to look up the environment, while in the
$\IAM$, this needs to be reconstructed, because encountered $\beta$-redexes were not recorded. This is done by the
next four steps, where again the search is up to $\beta$-redexes.
\[{\footnotesize 
\begin{array}{c|c|c|c|c}
\mathsf{Sub}\mbox{-}\mathsf{term} & \mathsf{Context} & \mathsf{\Log} & 
\mathsf{Tape} & \mathsf{Dir}
\\
		\cline{1-5}

  \ndstatetab{x} {((\la z\la x\ctxhole) w)(\la y y)} {\resm} {\epsilon} 
  \downp \\
  \nustatetab{\la x x} {((\la z\ctxhole) w)(\la y y)} 
  {(x,\la x\ctxhole,\epsilon)\cdot\resm} {\epsilon}{\upp} \\
  \nustatetab{\la z\la x x} {(\ctxhole w)(\la y y)} 
  {\resm\cdot(x,\la x\ctxhole,\epsilon)\cdot\resm} 
  {\epsilon}{\upp}\\
  \nustatetab{(\la z\la x x)w} {\ctxhole(\la y y)} 
  {(x,\la x\ctxhole,\epsilon)\cdot\resm} {\epsilon} 
  \upp\\
  \ndstatetab{\la y y} {((\la z\la x x)w)\ctxhole} {\resm} 
  {(x,\la x\ctxhole,\epsilon)}{\downp} \\
\end{array}}
\]
Some further crucial aspects of the $\IAM$ show up here.
\begin{itemize}
\item \emph{Phases}:
  the $\IAM$ starts looking for the term that may be substituted for $x$, 
  from a natural place, 
  namely the $\lambda$-abstraction
  binding $\var$. One needs to keep track of which
  of the (possibly many) occurrences of the bound variable one is coming from.
  This is done by simply pushing on the tape the position of the found occurrence of $\var$ (w.r.t. its binder), and by switching the machine in \emph{upward} mode $\upp$. 
  \item \emph{Locality}: transitions are \emph{local} in the sense that they 
  move between contiguous positions. Note that also the transition from the 
  variable occurrence to the binder is local if $\l$-terms are represented by 
  implementing occurrences as pointers to their binders, as in the proof net 
  representation of $\l$-terms, see \refsect{pn} for a precise comparison.

  \item \emph{Log}:
  the upward journey is guided by the context (note the \blue{blue} color). In the example, a term that would be substituted is found, namely $\la\vartwo\vartwo$. Observe 
  that the \emph{log}
  gets touched for the first time. Roughly, it saves the information that the 
  sub-term $\la\vartwo\vartwo$ currently under evaluation is meant to replace 
  the occurrence of position $(x,\la x\ctxhole)$, even if such replacement 
  never happens. The log keeps enough information as to 
  potentially backtrack to 
the position in its entry, called \emph{\trpos}, as it shall be explained in the next section.

  \item \emph{Succeed or iterate}: in general, if the machine finds no term to 
  substitute on $\var$, then the \trpos shall not be removed from the tape, providing the result of the run---that is, the head variable. 
 If instead 
  a term $\tmtwo$ to substitute is found, as in the example, then the process 
  starts over, switching to $\downp$ phase and looking for the head variable 
  of $\tmtwo$.
\end{itemize}
Once the argument $\la y y$ is found, the $\IAM$ now looks for its head 
variable $y$. Please note that this is possible because of the $\resm$ on the tape. Otherwise, \ie if the initial state were 
$\dstate\tm\ctxhole\epsilon\epsilon$, the $\IAM$ would be stuck in this 
\emph{final} 
state, signaling that $\tm\toh\la y\tmtwo$ for a term $\tmtwo$. Indeed, each 
$\resm$ in the initial state allows for the inspection of one head lambda 
of $\hnf\tm$.
\[{\footnotesize
\begin{array}{c|c|c|c|c}
\mathsf{Sub}\mbox{-}\mathsf{term} & \mathsf{Context} & \mathsf{\Log} & 
\mathsf{Tape} & \mathsf{Dir}
\\
\cline{1-5}

\ndstatetab{\la y y} {((\la z\la x x)w)\ctxhole} {\resm} 
{(x,\la\var\ctxhole,\epsilon)}{\downp} \\
\ndstatetab{y} {((\la z\la x x)w)(\la y\ctxhole)} {\epsilon} 
{(x,\la\var\ctxhole,\epsilon)}{\downp} \\
\nustatetab{\la y y} {((\la z\la x x)w)\ctxhole} {(y,\la\vartwo\ctxhole,\epsilon)} 
{(x,\la\var\ctxhole,\epsilon)}{\upp} \\
{\cdots}&{\cdots}&{\cdots}&{\cdots}&{\cdots}\\
\nustatetab{\tm} {\ctxhole} {(y,\la\vartwo\ctxhole,\epsilon)} 
{\epsilon} \upp
\end{array}}
\]
The head variable $y$ is found in two steps. After that, the machine switches 
to $\upp$ phase and runs again though the same path, thus arriving again at the 
root of the term $\tm$. The $\IAM$ then stops and gives its output: $y$, the 
head variable of $\hnf\tm$, is on the tape.

%% file: 03-IAM.tex
\section{The Lambda Interaction Abstract Machine}\label{sec:iam}
In this section we introduce the data structures used by the $\IAM$ and its transition rules. 

\paragraph{Terms and Levelled Contexts.} Let $\mathcal{V}$ be a countable set 
of variables. 
Terms of the \emph{$\lambda$-calculus} are defined as follows.
\[\begin{array}{rrcl}
   \textsc{$\l$-terms}\quad\quad & \tm,\tmtwo,\tmthree & \grameq & x\in\mathcal{V}\midd \lambda x.\tm\midd 
\tm\tmtwo.
  \end{array}
\]
\emph{Free} and \emph{bound variables} are defined as 
usual: $\la\var\tm$ binds $\var$ in $\tm$.
Terms are considered modulo $\alpha$-equivalence, and $\tm\isub\var\tmtwo$ denotes capture-avoiding (meta-level) substitution of 
all the free occurrences of $\var$ for $\tmtwo$ in $\tm$. 

The study of the $\IAM$ requires \emph{contexts}, that are terms with a single 
occurrence of a special constant $\ctxhole$, called \emph{the hole}, that is a 
place-holder for a removed sub-term. In fact, we need a notion of 
context more informative than the usual one, introduced 
next.
\[
\begin{array}{cc}
\textsc{Leveled contexts}\quad\quad & 
\begin{array}{rcl}
\ctx_0		& \grameq &	\ctxhole \midd \la\var\ctx_0 \midd \ctx_0\tm;\\
\ctx_{n+1}	& \grameq &	\la\var\ctx_{n+1}\midd \ctx_{n+1}\tm \midd \tm\ctx_{n}.
\end{array}
\end{array}
\]
The index $n$ in $\ctx_n$ counts the number of arguments into which the 
hole $\ctxhole$ is contained in $\ctx_n$. Such an index has a natural 
interpretation in linear logic terms. According to the standard (call-by-name) 
translation of the $\lambda$-calculus into linear logic proof nets, in a 
context $\ctx_n$ the hole lies inside
exactly $n$ $!$-boxes. Contexts of level 0 are also called \emph{head contexts} 
and are 
denoted by $\hctx,\hctxtwo,\hctxthree$. The level of a context shall be omitted 
when not relevant to the discussion---note that any ordinary 
context can be written \emph{in a unique way} as a leveled context, so that the 
omission is anyway harmless. 

The \emph{plugging} 
$\ctx_n\ctxholep\tm$ of a term 
$\tm$ in $\ctx_n$ is defined by replacing the hole $\ctxhole$ with $\tm$, 
potentially capturing free variables of $\tm$. Plugging 
$\ctx_n\ctxholep{\ctx_m}$ of a context for a context is defined similarly. 
A \emph{position} (of 
level $n$) in a term $\tmtwo$ is a pair $(\tm,\ctx_{n})$ such that 
$\ctx_n\ctxholep{\tm}=\tmtwo$.

%

\paragraph{Logs and \TrPoss.} The $\IAM$ relies on two mutually recursive 
notions, namely \emph{\trposs} and \emph{logs}: a \trpos is a position $(\tm, 
\ctx_n)$ together with a log\footnote{In computer science logs are traces that 
can only grow, while here they also shrink. The terminology suggests a tracing 
mechanism---\emph{trace} is avoided because related to categorical formulations 
of the GoI.} $\ste_n$, that is a list of \trposs, having length $n$.
\[
\begin{array}{c@{\hspace{2cm}} c@{\hspace{.8cm}} c}
\textsc{\TrPoss}
&\multicolumn{2}{c}{\textsc{\Logs}}
\\
\exps  \grameq (\tm, \ctx_n,\ste_{n}) 
&
\ste_0  \grameq  \epsilon
&
\ste_{n+1}  \grameq  \exps\cdot\ste_n
\end{array}
\]
The set of \trposs is $\expset$, and we use $\cdot$ also to concatenate logs, 
writing, \eg, $\ste_n\cdot\ste$, using $\ste$ for a log of unspecified length. 
Intuitively, logs contain some minimal information for backtracking to the 
associated position.\medskip

\paragraph{Tape, Token, Direction, State.} 
The \emph{tape} $\stme$ is a finite sequence of elements of two kinds, namely \trposs, and occurrences of the special symbol $\resm$, needed to cross abstractions 
and applications. A \emph{token} is a log plus a tape. A machine \emph{state} 
is given by a position and a token, together with a mode of operation called 
\emph{direction}.
\begin{definition}[$\IAM$ State] 
	A state $\state$ of the $\IAM$ is a quintuple $(\tm,\ctx,\ste,\stme,\pol)$ where:
	\begin{enumerate}
		\item $\tm$ is a $\lambda$-term: the \emph{code term};
		\item $\ctx$ is a context: the \emph{code context};
		\item $\ste$ is an element of $\expset^*$: the \emph{log};
		\item $\stme$ is an element of $(\set{\resm}\cup\expset)^*$: the \emph{tape};
		\item $\pol$ is an element in $\polset=\{\upp,\downp\}$: the \emph{direction}.
	\end{enumerate}
\end{definition}
Directions shall be represented mostly via colors and underlining: the code 
term in \red{red} and underlined, to represent 
$\downp$, and the code context in \blue{blue} and underlined, to represent 
$\upp$. This way, 
the fifth component is often omitted.

\paragraph{Initial  States.} The $\IAM$ starts on \emph{initial states} of the 
form $\state_{\tm,k}\defeq (\tm,\ctxhole,\epsilon,\resm^k,\downp)$, where 
$\tm$ is a term, $k\geq 0$ is the \emph{depth} of the state, and 
$\epsilon$ is the empty log. Intuitively, the machine
evaluates the term $\tm$ being allowed to inspect up to $k$ 
$\lambda$-abstractions of the head normal form of $\tm$. 
Note that there are many initial states for a given term $\tm$, one for each tape $\resm^k$.
\begin{figure*}[t]
\centering
${\small \begin{array}{l@{\hspace{.6cm}} 
l@{\hspace{.6cm}}l@{\hspace{.6cm}}lll@{\hspace{.6cm}} 
l@{\hspace{.6cm}}l@{\hspace{.6cm}}l}
	\mathsf{Sub}\mbox{-}\mathsf{term} & \mathsf{Context} & \mathsf{Log} & \mathsf{Tape}
	&&	\mathsf{Sub}\mbox{-}\mathsf{term} & \mathsf{Context} & \mathsf{Log} & \mathsf{Tape}
	\\
		\hhline{=========}
		\\[-8pt]
    \dstatetab{ \tmtwo\tm }{ \ctx }{ \stme }{ \ste } &
    \iamdap &
    \dstatetab{ \tmtwo }{ \ctxp{\ctxhole\tm} }{ \resm\cdot\stme }{ \ste }
	\\[3pt]		
    
    \dstatetab{ \la\var\tm }{ \ctx }{ \resm\cdot\stme }{ \ste } &
    \iamdlamone &
    \dstatetab{ \tm }{ \ctxp{\la\var\ctxhole} }{ \stme }{ \ste }
	\\[3pt]  
	
	\dstatetab{ \var }{ \ctxp{\la\var\ctxtwo_n} }{ \stme }{ \expsn\cdot\ste } &
    \iamdvar &
    \ustatetab{ \la\var\ctxtwo_n\ctxholep\var}{ \ctx }{ (\var,\la\var\ctxtwo_n,\expsn)\cdot\stme }{ \ste }
	\\[3pt]
	
	\dstatetab{ \la\var\ctxtwo_n\ctxholep{\var} }{ \ctx }{ 
	(\var,\la\var\ctxtwo_n,\ste_n)\cdot\stme }{ \ste } &
    \iamdlamtwo &
    \ustatetab{ \var }{ \ctxp{\la\var\ctxtwo_n} }{ \stme }{ \ste_n\cdot\ste }
	 \\[3pt]
		\cline{1-9}
		\\[-8pt]
    \ustatetab{ \tmtwo }{ \ctxp{\ctxhole\tm} }{ \resm\cdot\stme }{ \ste } &
    \iamuapltwo &
    \ustatetab{ \tmtwo\tm }{ \ctx }{ \stme }{ \ste }
	\\[3pt]

    \ustatetab{ \tm }{ \ctxp{\la\var\ctxhole} }{ \stme }{ \ste } &
    \iamulam &
    \ustatetab{ \la\var\tm }{ \ctx }{ \resm\cdot\stme }{ \ste }
	\\[3pt]

    \ustatetab{ \tmtwo }{ \ctxp{\ctxhole\tm} }{ \exps\cdot\stme }{ \ste } &
    \iamuaplone &
    \dstatetab{ \tm }{ \ctxp{\tmtwo\ctxhole} }{ \stme }{ \exps\cdot\ste }
	 \\[3pt]
    
    \ustatetab{ \tm }{ \ctxp{\tmtwo\ctxhole} }{ \stme }{ \exps\cdot\ste } &
    \iamuapr &
    \dstatetab{ \tmtwo }{ \ctxp{\ctxhole\tm} }{ \exps\cdot\stme }{ \ste }
	\\
		\cline{1-9}
	\end{array}}$
	\vspace{-8pt}
	\caption{$\IAM$ transitions.}
	\label{tab:iam}
\end{figure*}

\paragraph{Transitions.} The transitions of the $\IAM$ are in 
Fig.~\ref{tab:iam}. Their union is noted $\toiam$. A \emph{run} is a potentially empty sequence of transitions. A state $s$ is \emph{reachable} if 
$s_{\tm,k} \toiam^*s$ for an initial state $s_{\tm,k}$ and it is \emph{final} if there 
exists no $\statetwo$ such that $\state \toiam \statetwo$. The shape of final 
states is characterized in Sect.~\ref{sec:prop}. 

The idea is that $\downp$-states 
$\dstate\tm\ctx\stme\ste$ are queries about the head variable of (the head normal form of) $\tm$ and 
$\upp$-states $\ustate\tm\ctx\stme\ste$ are queries about the argument of an abstraction.
Next, we explain how the transitions realize three entangled mechanisms of the machine.

\paragraph{Mechanism 1: Search Up to $\beta$-Redexes} Note that $\iamdap$ skips the 
 argument and adds a $\resm$ on the tape. The idea is that $\resm$ keeps track 
 that an argument has been encountered---its identity is however forgotten. 
 Then $\iamdlamone$ does the dual job: it skips an abstraction when the tape 
 carries a $\resm$, that is, the trace of a previously encountered 
 argument. This mechanism thus realizes search \emph{up to $\beta$-redexes}, 
 that is, without recording them and leaving the tape unchanged. Note that 
 $\iamuapltwo$ and $\iamulam$ realize the same during the $\upp$ phase.

\paragraph{Mechanism 2: Finding Variables and Arguments} When the head variable $\var$ of the active subterm is found, transition $\iamdvar$  switches from direction $\downp$ to $\upp$, and the machine starts looking for potential substitutions for $\var$. The $\IAM$ then moves to the position of the binder $\lambda\var$ of $\var$, and starts exploring the context $\ctx$, looking for the first argument up to $\beta$-redexes. The relative position of $\var$ w.r.t. its binder is recorded in a new \trpos that is added to the tape. Since the machine moves out of a context of level $n$, namely $\ctxtwo_n$, the \trpos contains the first $n$ \trposs of the log. Roughly, this is an encoding of the run that led from the level of $\la\var\ctxtwo_n\ctxholep\var$ to the occurrence of $\var$ at hand, in case the machine would later need to backtrack. 
 
 When the argument $\tm$ for the abstraction binding the variable $\var$ in 
 $\exps$ is found, transition $\iamuaplone$ switches direction from $\upp$ to $\downp$, making the machine looking for the head 
 variable of $\tm$. Note that moving to $\tm$, the level 
 increases, and that the \trpos $\exps$ is moved from the tape to the 
log. The idea is that $\exps$ is now a completed argument query, 
 and it becomes part of the history of how the machine got to the current 
 position, to be potentially used for backtracking.

\paragraph{Mechanism 3: Backtracking} It is started by transition $\iamuapr$. 
 The idea is that the search for an argument of the $\upp$-phase has to 
 temporarily stop, because there are no arguments left at the current level. 
 The search of the argument then has to be done among the arguments of the 
 variable occurrence that triggered the search, encoded in $\exps$. Then the 
 machine enters into backtracking mode, which is denoted by a $\downp$-phase 
 with a \trpos on the tape, to reach the position in $\exps$. 
 Backtracking is over when $\iamdlamtwo$ is fired. 
 
 The $\downp$-phase and the \trpos
 on the tape mean that the $\IAM$ is backtracking. In fact, in this configuration the machine is 
 not looking for the head variable of the current subterm $\la\var\tm$, it is 
 rather going back to the variable position in the tape, to find its 
 argument. This is realized by moving to the position in the tape and 
 changing direction. Moreover, the log $\ste_n$
 encapsulated in the \trpos is put back on the global log. An invariant 
 shall guarantee that the \trpos on the tape always contains a position 
 relative to the active abstraction.\begin{SHORT} We provide an example of a 
 $\IAM$ run that exhibits backtracking in the Appendix (\refsect{backtracking}).\end{SHORT}

\begin{LONG}\begin{example}
	We provide an example of a $\IAM$ run that exhibits backtracking. Let us consider the $\lambda$-term $\tm:=(\la\var\var\var)(\la\vartwo\vartwo)$. We evaluate $\tm$ according to weak head reduction, thus starting from the state $\dstate\tm\ctxhole\epsilon\epsilon$. The first steps of the computation are needed to reach the head variable, namely $\var$.
	\[{\footnotesize
	\begin{array}{c|c|c|c|c}
	\mathsf{Sub}\mbox{-}\mathsf{term} & \mathsf{Context} & \mathsf{\Log} & 
	\mathsf{Tape} & \mathsf{Dir}
	\\
	\cline{1-5}
	
	\ndstatetab{(\la x xx)(\la y y)} {\ctxhole} {\epsilon} {\epsilon} \downp\\
	\ndstatetab{\la x xx} {\ctxhole(\la y y)} {\resm} {\epsilon} 
	\downp\\
	\ndstatetab{xx} {(\la x \ctxhole)(\la y y)} {\epsilon} {\epsilon} \downp\\
	\ndstatetab{x} {(\la x \ctxhole x)(\la y y)} {\resm} {\epsilon} 
	\downp
	\end{array}}
	\]
	Once the head variable $\var$ has been found, the machine switches to upward mode $\upp$ in order to find its argument $\la\vartwo\vartwo$.
	\[{\footnotesize
	\begin{array}{c|c|c|c|c}
	\mathsf{Sub}\mbox{-}\mathsf{term} & \mathsf{Context} & \mathsf{\Log} & 
	\mathsf{Tape} & \mathsf{Dir}
	\\
	\cline{1-5}
	\ndstatetab{x} {(\la x \ctxhole x)(\la y y)} {\resm} 
	{\epsilon}{\downp} \\
	\nustatetab{\la x xx} {\ctxhole(\la y y)} 
	{(x,\la x \ctxhole x,\epsilon)\cdot\resm} {\epsilon}{\upp} \\
	\ndstatetab{\la y y} {(\la x xx)\ctxhole} {\resm} 
	{(x,\la x \ctxhole x,\epsilon)}{\downp} \\
	\end{array}}
	\]
	Intuitively, the first occurrence of $\var$ has been substituted for $\la\vartwo\vartwo$, thus forming a new virtual $\beta$-redex $(\la\vartwo\vartwo)x$. Indeed, a $\resm$ is on top of the tape, thus allowing the $\IAM$ to inspect $\la\vartwo\vartwo$, reaching its head variable $\vartwo$.
	\[{\footnotesize
	\begin{array}{c|c|c|c|c}
	\mathsf{Sub}\mbox{-}\mathsf{term} & \mathsf{Context} & \mathsf{\Log} & 
	\mathsf{Tape} & \mathsf{Dir}
	\\
	\cline{1-5}
	\ndstatetab{\la y y} {(\la x xx)\ctxhole} {\resm} 
	{(x,\la x \ctxhole x,\epsilon)}\downp \\
	\ndstatetab{y} {(\la x xx)(\la y \ctxhole)} {\epsilon} 
	{(x,\la x \ctxhole x,\epsilon)}\downp \\
	\nustatetab{\la y y} {(\la x xx)\ctxhole} 
	{(y,\la\vartwo\ctxhole,\epsilon)} {(x,\la x \ctxhole x,\epsilon)}\upp \\
	\end{array}}
	\]
	Once the head variable $y$ has been found, the machine, in upward mode $\upp$, starts looking for the argument of $y$ from its binder $\la\vartwo\vartwo$. However, $\la\vartwo\vartwo$ was not the left side of an application forming a $\beta$-redex. Indeed, it was virtually substituted for the first occurrence of $\var$, in the log, thus creating the virtual redex $(\la\vartwo\vartwo)x$. Its argument is thus the second occurrence of $\var$. The $\IAM$ is able to retrieve it, walking again the path towards the variable $\la\vartwo\vartwo$ has been virtually substituted for, namely the first occurrence of $x$, saved in the log. This is what we call backtracking.
	\[{\footnotesize
	\begin{array}{c|c|c|c|c}
	\mathsf{Sub}\mbox{-}\mathsf{term} & \mathsf{Context} & \mathsf{\Log} & 
	\mathsf{Tape} & \mathsf{Dir}
	\\
	\cline{1-5}
	\nustatetab{\la y y} {(\la x xx)\ctxhole} 
	{(y,\la\vartwo\ctxhole,\epsilon)} {(x,\la x \ctxhole x,\epsilon)}\upp \\
	\ndstatetab{\la x xx} {\ctxhole(\la y y)} 
	{(x,\la x \ctxhole x,\epsilon)\cdot(y,\la\vartwo\ctxhole,\epsilon)} {\epsilon}\downp \\
	\nustatetab{x} {(\la x \ctxhole x)(\la y y)} 
	{(y,\la\vartwo\ctxhole,\epsilon)} {\epsilon}\upp \\
	\ndstatetab{x} {(\la x x\ctxhole)(\la y y)} {\epsilon} 
	{(y,\la\vartwo\ctxhole,\epsilon)}\downp \\
	\end{array}}
	\]
	Notice that we are able to backtrack because we saved the occurrence of the substituted variable in the token, otherwise the machine would not be able to know  which occurrence of $x$ is the right one. Of course, when the first occurrence of $x$ is reached the $\IAM$, now again in upward mode $\upp$, finds immediately its argument, that is the second occurrence of $x$. At this point the machine looks for the argument of this last occurrence of $x$, finding, of course, again $\la\vartwo\vartwo$.
	\[{\footnotesize
	\begin{array}{c|c|c|c|c}
	\mathsf{Sub}\mbox{-}\mathsf{term} & \mathsf{Context} & \mathsf{\Log} & 
	\mathsf{Tape} & \mathsf{Dir}
	\\
	\cline{1-5}
	\ndstatetab{x} {(\la x x\ctxhole)(\la y y)} {\epsilon} 
	{(y,\la\vartwo\ctxhole,\epsilon)}\downp \\
	\nustatetab{\la x xx} {\ctxhole(\la y y)} 
	{(x,\la x x\ctxhole,(y,\la\vartwo\ctxhole,\epsilon))} {\epsilon}\upp \\
	\ndstatetab{\la y y} {(\la x xx)\ctxhole} {\epsilon} 
	{(x,\la x x\ctxhole,(y,\la\vartwo\ctxhole,\epsilon))}\downp\\
	\end{array}}
	\]
	The computation then stops, signalling that $\tm$ has weak head normal form. Please notice that the position on the log has now a nested structure. Indeed it carries information about the virtual substitutions already performed.
\end{example}\end{LONG}



%% file: 04-Properties_of_the_IALM.tex
\section{Properties of the $\lambda$-IAM}\label{sec:prop}
Here we first discuss a few invariants of the data structures of the machine, and then we analyze final states and the semantic interpretation defined by the $\IAM$.\medskip

\paragraph{The Code Invariant.} An inspection of the 
rules shows that, along a computation, the 
machine travels on a $\lambda$-term without altering it.
\begin{proposition}[Code Invariant]
	If $(\tm,\ctx,\ste,\stme,\pol)\toiam(\tmtwo,\ctxtwo,\stetwo,\stmetwo,\pol')$, then $\ctxp{\tm}=\ctxtwop{\tmtwo}$.
\end{proposition}

\paragraph{The Balance Invariant.} Given a state 
$(\tm,\ctx,\ste,\stme,\pol)$, the log and the tape, \ie the token, 
verify two easy invariants connecting them to the position $(\tm,\ctx)$ and the 
direction $\pol$. The log $\ste$, together with the position 
$(\tm,\ctx)$,  forms a \trpos, \ie the length of $\ste$ is exactly the level of 
the code context 
$\ctx$. Then, the length of $\ste$ is exactly the number of (linear 
logic) \emph{boxes} in which the code term is contained. This fact guarantees 
that the $\IAM$ never gets stuck because the log is not long enough for 
transitions $\iamdvar$ and $\iamuapr$ to apply.

About the tape, note that every time the machine switches from a 
$\downp$-state to an $\upp$-state (or vice versa), a \trpos is 
pushed (or popped) from the tape $\stme$. Thus, for reachable states, the number of \trposs in $\stme$ gives the direction of the state. These intuitions are formalized by the balance invariant below. Given a direction $\pol$ we use
$\pol^n$ for the direction obtained by switching $\pol$ exactly $n$
times (i.e., $\downp^0=\downp$, $\upp^0=\upp$, $\downp^{n+1} =
\upp^{n}$ and $\upp^{n+1}=\downp^{n}$).
%
\begin{lemma}[Balance Invariant]\label{lemma:invarianttwo}
  Let $\state = \nopolstate{\tm}{\ctx_n}{\stme}{\ste}{\pol}$ be a reachable state and $\sizee\stme$ the number of \trposs in $\stme$. Then 
  \begin{enumerate}
  	\item \emph{Position and log}: $(\tm,\ctx_n, \ste)$ is a \trpos, and 
	\item \emph{Tape and direction}: $\pol=\downp^{\sizee\stme}$.
  \end{enumerate}
\end{lemma}
\begin{LONG}\begin{proof}
  By induction on the execution $\state_0 \toiam^k \state$ from the initial 
  state
  $\state_0$. If $k=0$,
  $s=i=\dstate\tm\ctxhole{\resm^k}\epsilon$. Clearly
  $\ctxhole$ is a level $0$ context, and $|\ste|=0$. Moreover,
  $|\stme|_\textsf{e}=0$ and $\downp^0=\downp$. Now, let us
  consider a IAM run of length $k>0$ and let $\{s_h\}_{0\leq h\leq k}$
  be the sequence of states of this run. By induction hypothesis
  $s_{k-1}=(\tm,\ctx_{n},\stme,\ste,\pol)$ is a \trpos i.e $|\ste|=n$
  and $\downp^{|\stme|_\textsf{e}}=\pol$. We can show, by cases, that the Lemma 
  holds for $s_k$.
  \begin{itemize}
  	\item $\pol=\downp$.
  	\begin{itemize}
  		\item
  		$\tm=\tmtwo\tmthree$. Then
  		$s_k=\dstate\tmtwo{\ctx\ctxholep{\ctxhole\tmthree}}{\resm\cdot\stme}\ste$.
  		$\ctx\ctxholep{\ctxhole\tmthree}$
  		is a context of level $n=|\ste|$ and both $|\stme|_\textsf{e}$
  		and $\pol$ are unchanged.
  		\item
  		$\tm=\la x\tmtwo$ and $\stme=\resm\cdot\stme'$. Then
  		$s_k=\dstate\tmtwo{\ctx\ctxholep{\la
  				x\ctxhole}}{\stme'}\ste$. $\ctx\ctxholep{\la
  			x\ctxhole}$ is a context of level $n=|\ste|$ and both
  		$|\stme|_\textsf{e}$ and $\pol$ are unchanged.
  		\item
  		$\tm=\la x\ctxtwo_m\ctxholep\var$ and $\stme=(x,\la
  		x\ctxtwo_m,\stetwo)\cdot\stme'$. Then $s_k=\ustate x{\ctx\ctxholep{\la
  				x\ctxtwo_m}}{\stme'}{\stetwo\cdot\ste}$. $\ctx\ctxholep{\la
  			x\ctxtwo_m}$ is a context of level $n+m=|\ste|+|\stetwo|$ and
  		since $\downp^{|\stme|_\textsf{e}}=\downp$, then
  		$\downp^{|\stme|_\textsf{e}'}=\downp^{|\stme|_\textsf{e}-1}=\upp$.
  		\item $\tm=x$, $\ctx=\ctx_m'\ctxholep{\lambda x.\ctxtwo_l}$ and 
  		$\ste=\ste_l\cdot\ste'$. Then $s_k=\ustate{\lambda 
  			x.\ctxtwo_l\ctxholep{x}}{\ctx_m'}{(x,\lambda 
  			x.\ctxtwo_l,\ste_l)\cdot\stme}{\ste'}$. Since $m+l=|\ste|$, then 
  		$|\ste'|=m$ and since $\downp^{|\stme|_\textsf{e}}=\downp$, then 
  		$\downp^{|\stme|_\textsf{e}+1}=\upp$.
  		\item
  		$\tm=x$, $\ctx=\ctx_m'\ctxholep{\ctxtwo_l[x\leftarrow\tmtwo]}$
  		and $\ste=\ste_l\cdot\ste'$\footnote{Notice that proofs are already 
  		carried out in the more general framework of the linear substitution 
  		calculus, to be introduced in Section~\ref{sect:micro-refinement}.}. 
  		Then
  		$s_k=\ustate\tmtwo{\ctx_m'\ctxholep{\ctxtwo_l\ctxholep{x} 
  				[x\leftarrow\ctxhole]}}\stme{(x,\ctxtwo_l[x\leftarrow\tmtwo],\ste_l)\cdot\ste'}$.
  		Since $m+l=|\ste|$, then
  		$|\ste'|=m$. Thus
  		$|(x,\ctxtwo_l[x\leftarrow\tmtwo],\ste_l)\cdot\ste'|=m+1$ which the 
  		level of
  		$\ctx_m'\ctxholep{\ctxtwo_l\ctxholep{x}[x\leftarrow\ctxhole]}$. Both
  		$|\stme|_\textsf{e}$ and $\pol$ are unchanged.
  		\item
  		$\tm=\tmtwo[x\leftarrow\tmthree]$. Then
  		$s_k=\dstate\tmtwo{\ctx\ctxholep{\ctxhole[x\leftarrow\tmthree]}}\stme\ste$.
  		$\ctx\ctxholep{\ctxhole[x\leftarrow\tmthree]}$
  		is context of level $n=|\ste|$. Both $|\stme|_\textsf{e}$ and
  		$\pol$ are unchanged.
  	\end{itemize}
  	\item $\pol=\upp$. The proof is equivalent to the one above.
  \end{itemize}
\end{proof}\end{LONG}
Note that, because of the invariant, the tape $\stme$ of a reachable 
$\upp$-state always contains at 
least one \trpos, which is why it can be seen as the answer to a query about 
the head variable.
\medskip

\paragraph{The Exhaustible State Invariant.} The study of the $\IAM$ requires 
to prove that some 
bad configurations never arise. On states such as 
$\dstate{ \la\var\ctxtwop\var }{ \ctx }{ \exps\cons\stme }{ \ste  }$, transition $\iamdlamtwo$ requires 
the \trpos $\exps$ to have  shape $(\var, \la\var\ctxtwo, \ste')$, 
that is, to contain a position isolating an occurrence of $\var$ in 
$\la\var\ctxtwop\var$, otherwise the machine is stuck.
 The \emph{exhaustible state 
invariant} guarantees that the machine never 
gets stuck for this reason. The invariant being technical, it is 
developed in the Section~\ref{sec:exh}. 
Here we only mention its main consequence.

\begin{proposition}[\TrPoss Never Block the $\IAM$]
\label{prop:exhaust} 
Let $\state$ be a reachable state.
  If $\state = \dstate{ \la\var\ctxtwop\var }{ \ctx }{ \exps\cons\stme }{ 
		\ste
	}$ then $\exps =
	(\var,\la\var\ctxtwo,\stetwo)$.
\end{proposition}

\paragraph{Reversibility.} The proof of \refprop{exhaust} relies on a key 
property of the $\IAM$, that is, bi-determinism, or
reversibility: for each state $s$ 
there is at most one state $s'$ such that $s'\toiam s$. 
The property follows by simply inspecting the rules. Moreover, a run can be 
reverted by 
simply switching the direction. 
\begin{proposition}[Reversibility]
	If \\$\nopolstate{\tm}{\ctx}{\stme}{\ste}{\pol}\toiam
	\nopolstate{\tmtwo}{\ctxtwo}{\stmetwo}{\stetwo}{\poltwo}$, then
	$\nopolstate{\tmtwo}{\ctxtwo}{\stmetwo}{\stetwo}{\poltwo^1}\toiam
	\nopolstate{\tm}{\ctx}{\stme}{\ste}{\pol^1}$.
\end{proposition}

\paragraph{Final States.} A run of initial state 
$\state_{\tm,k}=\dstate{\tm}{\ctxhole}{\resm^k}{\epsilon}$  may either 
never stop or end in one of three possible final states. To explain them, let 
$\la{\var_0}\ldots \la{\var_i} (\vartwo \tmtwo_1\ldots \tmtwo_j)$ be the 
head normal form $\hnf\tm$ of $\tm$. The exhaustible state invariant and \refprop{exhaust} together guarantee 
that the final states of the $\IAM$ can only have one of these three 
shapes:
\begin{varitemize}
	\item \emph{Failure} $\dstate{\la\var\tm}{\ctx}{\epsilon}{\ste}$: this is 
	the machine's way of saying that $i>k$, that is, $\hnf\tm$ has more head 
	abstraction than 
	those that the depth $k$ of the initial state $\state_{\tm,k}$ allows to explore.
	\item \emph{Open success} $\dstate{\vartwo}{\ctx}{\resm^j}{\ste}$: the 
	machine found the head variable, and it is 
	the \emph{free} variable $\vartwo$, which has $j$ arguments. Note that if 
	$y$ is instead bound by a 
	$\lambda$-abstraction, then the machine is not stuck, as the machine would 
	do a $\iamdvar$ transition (as guaranteed by the balance invariant).
	\item \emph{Bound success} $\ustate{\tm}{\ctxhole}{\resm^m\cdot 
	\exps\cdot\resm^j}{\ste}$: the head variable has been found and 
	it is $\vartwo = \var_{m}$, to which $j$ arguments are applied. When the 
	machine $\downp$-travels on the head 
	variable $\vartwo$, and it is abstracted, the \trpos $\exps$ containing 
	$\var_m$ is put on the tape and the direction switches---the answer has 
	been found. The sequence $\resm^m$ on top of tape in the final state 
	comes from the $\upp$ backtracking along the spine of $\hnf\tm$ for the 
	equivalent of $m$ abstractions, each one adding one $\resm$. At this point 
	the $\IAM$ stops. Thus the abstraction binding $\vartwo$ is $\l\var_{m}$.
\end{varitemize}

\paragraph{The Semantics.} The characterization of final states induces a semantic interpretation of terms, that we are going to 
show to be sound and adequate with respect to (linear) head evaluation.


\begin{definition}[$\IAM$ Semantics]
\label{def:semantics}
	We define the $\IAM$ semantics of $\lambda$-terms by way of a family of 
	functions 
	$\sem{\cdot}{k}:\Lambda\rightarrow(\mathbb{N}\times\mathbb{N}) 
	\cup(\mathcal{V}\times\mathbb{N})\cup\{\Downarrow,\bot\}$,
	 where $k\in\mathbb{N}$, defined as follows.
	\[
	\sem{\tm}{k}=\begin{cases}
	\langle h,j \rangle & \text{ if } 
	(\red{\tm},\ctxhole,\epsilon,\resm^k)\toiam^* 
	(\tm,\blue\ctxhole,\epsilon,\resm^h\cdot \exps\cdot\resm^j),\\
	\langle\var,h\rangle & \text{ if } 
	(\red{\tm},\ctxhole,\epsilon,\resm^k)\toiam^* 
	\dstate{\var}{\ctx}{\resm^h}{\ste},\\
	\Downarrow & \text{ if } (\red{\tm},\ctxhole,\epsilon,\resm^k)\toiam^*  
	(\red{\la\var\tmtwo},\ctx,\ste,\epsilon),\\
	\bot & \text{otherwise.}
	\end{cases}
	\]
\end{definition}

\subsection{Further Properties}
The following properties of the $\IAM$ are required for the proofs but are not essential for a first understanding of its functioning, so we suggest to skip them at a first reading.

\paragraph{Lifting} The $\IAM$ verifies a sort of context-freeness with respect 
to the tape $\stme$. Intuitively, the $\IAM$ consumes the next entry of the 
initial input only when the question asked by the previous one(s) has been 
fully answered. Precisely, \emph{lifting} the tape preserves the shape of the 
run and of the final state (up to lifting).

\begin{lemma}[Lifting]
\label{l:pumping}
	If $\nopolstate{\tm}{\ctx}{\tape}{\ste}{\pol}\toiam^n
	\nopolstate{\tmtwo}{\ctxtwo}{\tapetwo}{\stetwo}{\poltwo}$, then
	$\nopolstate{\tm}{\ctx}{\tape\cons\tapethree}{\ste}{\pol}\toiam^n
	\nopolstate{\tmtwo}{\ctxtwo}{\tapetwo\cons\tapethree}{\stetwo}{\poltwo}$.
\end{lemma}

\begin{LONG}\begin{proof}
	We proceed by induction on $n$. Thus we have that if 
	$\nopolstate{\tm}{\ctx}{\tape}{\ste}{\pol}\toiam^{n-1}
	\nopolstate{\tmtwo}{\ctxtwo}{\tapetwo}{\stetwo}{\poltwo}$, then 
	$\nopolstate{\tm}{\ctx}{\tape\cons\tapethree}{\ste}{\pol}\toiam^{n-1}
	\nopolstate{\tmtwo}{\ctxtwo}{\tapetwo\cdot\tapethree}{\stetwo}{\poltwo}$. 
	The proof now proceeds analyzing all possible transitions from 
	$\nopolstate{\tmtwo}{\ctxtwo}{\stme}{\tape}{\poltwo}$ and 
	$\nopolstate{\tmtwo}{\ctxtwo}{\tape\cons\tapethree}{\stetwo}{\poltwo}$. The 
	key point 
	is that every transition of the $\IAM$ consumes at most $1$ element of the 
	tape. This is why the pushed stack $\tapethree$ never gets touched.
\end{proof}\end{LONG}

\paragraph{Monotonicity of Runs} The previous lemma states that lifting the 
input from $\resm^k$ to 
$\resm^{k+1}$ cannot decrease the length of the $\IAM$ run. Next, we show that 
if the run of input $\resm^k$ is successful then the run of input 
$\resm^{k+1}$ is also successful, in the same way, and it has the same length. 
As a consequence, the length may increase only if the run on $\resm^k$ fails.

We write $|\tm|_k$ for the length of the $\IAM$ run of initial state 
$\state_{\tm,k} \defeq \dstate{\tm}{\ctxhole}{\resm^k}{\epsilon}$, that is for 
the length of the maximum sequence of transitions $\state_{\tm,k}$, if the 
$\IAM$ terminates, and $|\tm|_k= \infty$ if the machine diverges. The next 
lemma compares run lengths, for which we consider that $i<\infty$ for every 
$i\in \nat$ and $\infty \not<\infty$. We also write $\state_{\tm,k}^n$ for the 
state such that $\state_{\tm,k} \toiam^n \state_{\tm,k}^n$, if it exists.

\begin{lemma}[Monotonicity of runs]
\label{l:runs-monotonicity}
 The length of runs cannot decrease if the input increases, that is, $|\tm|_k\leq |\tm|_{k+1}$. Moreover, if $|\tm|_k =n\in\nat$ and the final state $\state_{\tm,k}^n$ is bound (resp. open) successful then $|\tm|_k = |\tm|_{h}$ for every $h>k$ and the final state $\state_{\tm,h}^n$ is bound (resp. open) successful. 
\end{lemma}

\begin{LONG}\begin{proof}
 Let $\state_{\tm,k} \toiam^n (\tmtwo,\ctx,\ste,\stme,\pol) = \state_{\tm,k}^n$.
 By the pumping lemma (\reflemma{pumping}), if $\state_{\tm,k+1} \toiam^n 
 (\tmtwo,\ctx,\ste,\stme\cdot\resm,\pol) =\state_{\tm,k+1}^n$. If $|\tm|_k= 
 \infty$ then $\state_{\tm,k} \toiam^n \state_{\tm,k}^n$ for every $n\in\nat$ 
 and so $\state_{\tm,k+1} \toiam^n \state_{\tm,k+1}^n$, that is, $|\tm|_{k+1}= 
 \infty=|\tm|_k$. 
 
 If $|\tm|_k= n\in \nat$ then $\state_{\tm,k}^n$ is final. Two cases. If 
 $\state_{\tm,k}^n$ is an approximating final state 
 $\dstate{\la\var\tm}{\ctx}{\epsilon}{\ste}$ then 
 $\state_{\tm,k+1}^n=\dstate{\la\var\tm}{\ctx}{\resm}{\ste}$ which can make a 
 transition, that is, $|\tm|_k< |\tm|_{k+1}$. If instead $\state_{\tm,k}^n$ is 
 a bound successful final state $\ustate{\tm}{\ctxhole}{\resm^m\cdot 
 \exps\cdot\resm^n}{\ste}$ then $\state_{\tm,k+1}^n = 
 \ustate{\tm}{\ctxhole}{\resm^m\cdot \exps\cdot\resm^{n+1}}{\ste}$ which is 
 also a successful final state, and $|\tm|_k = |\tm|_{k+1}$. Similarly for an 
 open successful final state. A straightforward induction then shows that the 
 same holds for every other $h>k$.
\end{proof}\end{LONG}

%% file: 05-Correctness.tex
\section{Soundness and Adequacy, Explained}
\label{sect:soundness-intro}

Proving the implementation theorem of the $\IAM$ amounts to showing that the 
interpretation $\sem{\tm}{}$ of \refdef{semantics} is a sound and adequate 
semantics for $\lambda$-terms with respect to head evaluation. \emph{Soundness} 
is the invariance of $\sem{\tm}{}$ by head evaluation. \emph{Adequacy} is the 
fact that $\sem{\tm}{}$ reflects the observable behavior of 
$\tm$, that is, termination in the case of weak evaluation. 
In the rest of this section, we compare this notion with the fundamentally different notion of implementation for environment machines. 

\paragraph{Soundness of Environment Machines.} An environment abstract machine $M$ 
executes a term $\tm$ according to a strategy $\to$ if from \emph{the} initial 
state $\state_\tm$ of code $\tm$ it computes a representation of the normal 
form $\mathsf{nf}_{\rightarrow}(\tm)$. In particular, the machine somehow maintains the representation of how the strategy $\to$ modifies the term $\tm$ they both 
evaluate. Soundness is a weak bisimulation between the transitions $\state 
\Rew{M} \statetwo$ of the machine and the steps $\tm\to\tmtwo$ of the strategy. 
In particular, a run $\exec_\tm$ of the machine on $\tm$ passes through some 
states representing $\tmtwo$, and the final states $\state_f$ of the machine 
decode to $\to$-normal forms.
\begin{figure*}[t!]
	\begin{center}$
		\begin{array}{c@{\hspace{.8cm}} cc}
		\begin{array}{rcl}
		\multicolumn{3}{c}{\textsc{Rules at top level}}
		\\
		\sctxp{\lambda x.\tm}\tmtwo \quad&\mapsto_\db& \quad 
		\sctxp{\tm\esub\var\tmtwo}\\
		\hctxp \var\esub\var\tm\quad&\mapsto_\ls& \quad \hctxp\tm\esub\var\tm\\
		\tm\esub\var\tmtwo\quad&\mapsto_\gcsym& \quad \tm\quad\text{if 
		$\var\notin\fv\tm$}
		\end{array}
		&
		\begin{array}{c}
		\textsc{Contextual closure}\\ 
		\infer[\tt{a} \in \set{\db,\ls,\gcsym}]{\hctxp\tm \lolli_{\tt{a}} 
			\hctxp\tmtwo}{\tm\rootRew{\tt{a}}\tmtwo}
		\end{array}
		&
		\begin{array}{c}
		\textsc{Notation}\\
		\tohl\,\defeq\, \tohldb\cup\tohlls\cup\tohlgc
		\end{array}
		\end{array}
		$
		\vspace{-8pt}
		\caption{Rewriting rules for linear head evaluation $\tohl$.}
		\label{fig:lhe}
	\end{center}
\end{figure*}

\paragraph{Soundness of the $\IAM$.} The $\IAM$, and more generally GoI machines, 
do implement strategies, but in a different way. The $\IAM$ has many initial 
states, therefore many runs, for a given code 
$\tm$, one for each possible depth $k\in\nat$. Moreover, the machine does not 
trace how the strategy modifies the term. If $\tm\toh\tmtwo$, a run of code 
$\tm$ never passes through a representation of $\tmtwo$, as soundness denotes 
something else. The idea is that, on a fixed input, the run of code $\tm$ is 
bisimilar to the run of code $\tmtwo$. Notably, the latter is 
shorter---rewriting the code is a way of improving the associated $\IAM$.
Notice the difference with environment machines: there the bisimulation is 
between \emph{steps} on terms and transitions on states. For the $\IAM$, it is 
between \emph{transitions} on states (of code $\tm$) and transitions on states 
(of code $\tmtwo$).

\paragraph{On Not Computing Results.} Another difference is that the $\IAM$ does not 
compute a code representation of the result $\hnf\tm$. It recovers the micro 
information $\sem\tm k$ about it, by exploring only the immutable code $\tm$. 
This is in accordance with other models: space-sensitive Turing machines do not 
compute the whole output but only single bits of it.
To compute the spine of $\hnf\tm$, one needs to compute $\sem\tm k$ for 
various values of $k$, one for each abstraction of the spine, starting each 
time with a different input $\stme =\resm^k$, and then once more for the head 
variable (adding a $\resm$), if the machine ever terminates. On a head 
normal form $\tm$, the runs of the $\IAM$ become an immediate interactive 
reading of the spine of $\tm$. Inputs represent questions about the head of the 
normal form, the answer is encoded in the tape at the end of the run, when the run succeeds.

\paragraph{Adequacy.} Soundness is not enough. A trivial semantics where every object is mapped on 
the same element, for example, is sound but not informative. Adequacy guarantees 
that the interpretation $\sem{\tm}{}$ reflects some 
observable aspects of $\tm$ and 
vice versa. For a head strategy in an untyped calculus, one usually observes 
termination, and, if it holds, the identity of the head variable. And this is 
exactly what $\sem{\tm}{}$ reflects, or is it adequate for.

%% file: 06-Micro-Step_Refinement.tex
\section{Micro-Step Refinement}
\label{sect:micro-refinement}
The proof of soundness of the $\IAM$ cannot be directly carried out with 
respect to head evaluation: this is specified using meta-level substitutions, 
here noted $\tm\isub\var\tmtwo$, which is a macro operation, potentially making 
\emph{many} copies of $\tmtwo$ and modifying $\tm$ in \emph{many} places, while 
the $\IAM$ does a minimalistic evaluation that in general does not even pass 
through most of those many places. It is very hard---if possible at all---to 
define explicitly a bisimulation of $\IAM$ runs (as required for soundness) 
that relates states whose code is modified by meta-level substitution.

We then switch to \emph{linear} head evaluation (shortened to LHE), a refinement of head evaluation in which substitution is performed in \emph{micro-steps}, replacing only the head variable occurrence, and keeping the substitution suspended for all the other occurrences. This is also the approach followed by Danos, Herbelin, and Regnier~\cite{DBLP:conf/lics/DanosHR96}. 

We depart from their approach, however, in the way we formally define LHE. We 
adopt a formulation where the suspension of the substitution is formalized via 
a sharing constructor $\tm\esub\var\tmtwo$, which is nothing else but a compact 
notation for $\letin \var\tmtwo\tm$, and the rewriting is modified accordingly. 
They instead avoid sharing, by encoding $\tm\esub\var\tmtwo$ as 
$(\la\var\tm)\tmtwo$, which is more compact but conflates different concepts 
and makes the technical development less clean.

An important point is that head evaluation and its linear variant are 
observationally equivalent, that is, one terminates on $\tm$ if and only if the 
other terminates on $\tm$, and they produce the same head 
variable\begin{SHORT}, for more details see the Appendix (\refsect{lhe-explanation})\end{SHORT}.

\paragraph{The Adopted Presentation.} Linear head evaluation was introduced by
Mascari \& Pedicini and Danos \& Regnier~\cite{DBLP:journals/tcs/MascariP94,Danos04headlinear} as a strategy on proof 
nets. It is to proof nets for the $\l$-calculus what head evaluation is to the $\l$-calculus. The presentation adopted here, noted $\tohl$,
was introduced by Accattoli~\cite{DBLP:conf/rta/Accattoli12}, formulated as a strategy in a $\l$-calculus with explicit sharing, the \emph{linear substitution
calculus}\footnote{The LSC is a subtle
reformulation of Milner's calculus with explicit
substitutions~\cite{DBLP:journals/entcs/Milner07,KesnerOConchuir},
 inspired by Accattoli and Kesner structural
$\l$-calculus ~\cite{DBLP:conf/csl/AccattoliK10}.} (shortened to LSC). The LSC presentation of $\tohl$ is isomorphic to the one on proof nets \cite{DBLP:conf/ictac/Accattoli18}, while the one used by Danos and Regnier---although closely related to proof nets---is not. It is isomorphic only up to Regnier's $\sigma$-equivalence~\cite{regnier_sigma}.

\paragraph{LSC Terms and Levelled contexts.} Let $\mathcal{V}$ be a countable set of variables. 
	Terms of the \emph{linear substitution calculus} (LSC) are defined by the 
	following grammar.
	\[\begin{array}{rrcl}
   \textsc{LSC terms}\quad\quad & \tm,\tmtwo,\tmthree & \grameq & x\in\mathcal{V}\midd \lambda x.\tm\midd 
	\tm\tmtwo\midd\tm\esub\var\tmtwo.

  \end{array}
\]
The construct $\tm\esub\var\tmtwo$ is called an \emph{explicit 
substitution} or ES, not to be confused with the meta-level substitution 
$\tm\isub\var\tmtwo$. 
As is standard, $\tm\esub\var\tmtwo$ binds 
$\var$ in $\tm$, but 
not in $\tmtwo$---terms are still considered up to $\alpha$-conversion. Levelled contexts naturally extend to the LSC.
\[
\begin{array}{rcl}
\multicolumn{3}{c}{\textsc{Leveled Contexts}}\\
\ctx_0		& \grameq &	\ctxhole \midd \la\var\ctx_0 \midd \ctx_0\tm \midd \ctx_0\esub\var\tm;\\
\ctx_{n+1}	& \grameq &	\la\var\ctx_{n+1}\midd \ctx_{n+1}\tm \midd \ctx_{n+1}\esub\var\tm \midd \tm\ctx_{n} \midd \tm\esub\var{\ctx_n}.
\end{array}
\]

\begin{figure*}[t]
\centering
${\footnotesize \begin{array}{l@{\hspace{.45cm}} 
l@{\hspace{.45cm}}l@{\hspace{.45cm}}lll@{\hspace{.45cm}} 
l@{\hspace{.45cm}}l@{\hspace{.45cm}}l}
	\mathsf{Sub}\mbox{-}\mathsf{term} & \mathsf{Context} & \mathsf{Log} & \mathsf{Tape}
	&&	\mathsf{Sub}\mbox{-}\mathsf{term} & \mathsf{Context} & \mathsf{Log} & \mathsf{Tape}
	\\
		\hhline{=========}
		\\[-8pt]
    \dstatetab{ \tm\esub\var\tmtwo }{ \ctx }{ \stme }{ \ste } &
    \iamdes &
    \dstatetab{ \tm }{ \ctxp{\ctxhole\esub\var\tmtwo} }{ \stme }{ \ste }
	\\[3pt]
	
	\dstatetab{ \var }{ \ctxp{\ctxtwo_n\esub\var\tmtwo} }{ \stme }{ 
	\expsn\cdot\ste } &
	\iamdvartwo &
	\dstatetab{ \tmtwo }{ \ctxp{\ctxtwo_n\ctxholep\var\esub\var\ctxhole} }{ \stme }{ (\var,\ctxtwo_n\esub\var\tmtwo,\expsn)\cdot\ste }
	\\[3pt]
		\cline{1-9}
		\\[-8pt]
	\ustatetab{ \tm }{ \ctxp{\ctxhole\esub\var\tmtwo} }{ \stme }{ \ste } &
	\iamues &
	\ustatetab{ \tm\esub\var\tmtwo }{ \ctx }{ \stme }{ \ste }
 \\[3pt]
 
	\ustatetab{ \tmtwo }{ \ctxp{\ctxtwo_n\ctxholep{\var}\esub\var\ctxhole} }{ 	\stme }{ (\var, \ctxtwo_n\esub\var\tmtwo,\ste_n)\cdot\ste } &
	\iamuestwo &
	\ustatetab{ \var }{ \ctxp{\ctxtwo_n\esub\var\tmtwo} }{ \stme }{ 	\ste_n\cdot\ste }
	\\
	\cline{1-9}
	\end{array}}$		
		\vspace{-8pt}
		
		\caption{Transitions for LSC-terms.}
		\label{tab:iames}
	\end{figure*}
	
\paragraph{Contexts and Plugging.} The LSC makes a crucial use of contexts to define 
its operational semantics. First of all, we need \emph{substitution 
contexts}, that simply packs together ES:
\[
	\begin{array}{r@{\hspace{.5cm}} rcl}
	\textsc{Substitution contexts}\quad\quad & \sctx 		& \grameq &	\ctxhole \midd \sctx\esub\var\tm.
	\end{array}
	\]
	When plugging is used for substitution contexts, we write it in a 
	post-fixed 
	manner, that is $\sctxp\tm$, to stress that the ES actually appears on the 
	right of $\tm$.
	
\paragraph{Linear Head Evaluation.} The LSC comes with a notion of reduction 
that 
resembles the decomposed, micro-step process of 
cut-elimination in linear logic proof-nets. Essentially, the meta-level 
substitution $\tm\isub\var\tmtwo$ is decomposed into a sequence of many 
replacements from $\tm\esub\var\tmtwo$ of one occurrence of $\var$ in $\tm$ 
with $\tmtwo$ at the time. Linear head evaluation, moreover, is the reduction 
that only replaces the head variable occurrence $\vartwo$, if it is bound by an 
ES $\esub\vartwo\tmthree$ and leaves the other occurrences of $\vartwo$, if 
any, bound by $\esub\vartwo\tmthree$. 

The rewriting rules are first defined at top level and then closed by head 
contexts, in Figure~\ref{fig:lhe}. A feature of the LSC is that contexts are 
also used to define the linear substitution rule at top level $\mapsto_\ls$.
In plugging $\tm$ in $\hctx$, rule $\tohlls$ may perform on-the-fly renaming of bound variables in $\hctx$, to avoid capture of free variables of $\tm$.
Often, the literature does not include rule $\togc$, responsible for erasing 
steps, in the definition of $\tohl$. The reason is that $\togc$ is strongly 
normalizing and it can be postponed.

\begin{LONG}Note that our definition of $\tohl$ allows more than one $\tohl$ redex at a time in a term. It is not a problem, as $\tohl$ has the diamond property---this is standard.

\paragraph{Relationship with Head Evaluation, and Normal Forms.} Linear head 
evaluation is studied at length in the literature, in particular its 
relationship with head evaluation is well known. On a given term $\tm$, linear 
head evaluation $\tohl$ terminates on the linear head normal form $\lhnf\tm$ if 
and only if head evaluation $\toh$ terminates on the head normal form 
$\hnf\tm$. Moreover, $\hnf\tm$ is obtained from $\lhnf\tm$ by simply 
\emph{unfolding} ES, that is turning them into meta-level substitutions. A 
linear head normal form has the same shape $\la{\var_1}\ldots \la{\var_k} 
(\vartwo \tm_1\ldots \tm_h)$ of a head normal form but for the fact that each 
spine sub-term may be surrounded by a substitution context $\sctx$, that is, 
they have the cumbersome shape (where $\sctx_i$ surrounds $\la{\var_i}\ldots 
\la{\var_k} (\vartwo \tm_1\ldots \tm_h)$ and $\sctxtwo_j$ surrounds $\vartwo 
\tm_1\ldots \tm_j$):
\begin{equation}
\label{eq:lhnf}
\ctxholep{\la{\var_1} \ctxholep{\la{\var_2}\ldots \ctxholep{\la{\var_k} 
(\ctxholep{\ctxholep{\ctxholep\vartwo\tm_1}\sctxtwo_1\ldots 
\tm_h}\sctxtwo_h)}\sctx_k\ldots}\sctx_2} \sctx_1
\end{equation}
where none of the ES in $\sctx_i$ and $\sctxtwo_j$ binds $\vartwo$ (otherwise there would be a $\tohlls$ redex). Unfolding 
the ES of a linear head normal form produces a head normal form having the same 
\emph{spine structure}, that is, with the same abstractions, the same 
\emph{head 
variable} and the same number of arguments---concretely, unfolding the term in 
\refeq{lhnf} one obtains the head normal form $\la{\var_1}\ldots \la{\var_k} 
(\vartwo \tmtwo_1\ldots \tmtwo_h)$ for some $\tmtwo_1, \ldots, \tmtwo_h$. Therefore, in the paper we shall refer to a \emph{$\tohl$-normal term up to substitution} $\la{\var_1}\ldots \la{\var_k} 
(\vartwo \tmtwo_1\ldots \tmtwo_h)$ meaning that we harmlessly ignore the substitution contexts around the spine sub-terms. Please note that we do 
not have any restriction on closed terms, and thus the number of 
$\lambda$-abstractions in the spine of $\lhnf\tm$ and $\hnf\tm$ could also be 0.

%
\end{LONG}

\begin{example}
	We provide here an example of LHR sequence.
	Consider the following 3 steps:
	\[
	\begin{array}{rcl}
	(\la\var\var\var)(\la\vartwo\vartwo)&\tohldb&
	(\var\var)\esub\var{\la\vartwo\vartwo}\tohlls((\la\vartwo\vartwo)\var)\esub\var{\la\vartwo\vartwo}\\
	&\tohldb&
	\vartwo\esub\vartwo\var\esub\var{\la\vartwo\vartwo}
	\end{array}{}
	\]
	that turn a $\beta$/multiplicative redex into a ES, substitute on the head variable occurrence, and continue with another multiplicative step. Two micro substitution steps on the head, followed by two steps of garbage collection complete the evaluation:
	\[
	\begin{array}{rcl}
	\vartwo\esub\vartwo\var\esub\var{\la\vartwo\vartwo}&\tohlls&
	\var\esub\vartwo\var\esub\var{\la\vartwo\vartwo}\\
	&\tohlls&
	(\la\vartwo\vartwo)\esub\vartwo\var\esub\var{\la\vartwo\vartwo}
		\tohlgc^2\la\varthree\varthree
		\end{array}
		\]
	\end{example}

\paragraph{Additional $\IAM$ transitions.} The $\IAM$ presented in the previous 
sections is easily adapted to the LSC, by simply considering (\tr) positions with respect to the extended syntax, and adding the 4
transitions for ES in Fig.~\ref{tab:iames}.

 Transitions $\iamdes$ and $\iamues$ simply skips ES during search---now search 
 is \emph{up to $\beta$-redexes and ES}. Transition $\iamdvartwo$ shortcuts the 
 search of the term $\tmtwo$ to substitute for $\var$, given that $\tmtwo$ is 
 already available in $\esub\var\tmtwo$. Therefore, the machine stays in the 
 $\downp$ phase and moves to evaluate $\tmtwo$. Note that the \trpos for $\var$ 
 is directly added to the log and not to the tape. This is because we have 
 avoided the search of the argument. We have reached it directly: 
 note that when a $\upp$-search ends with the $\iamuaplone$ transition, the 
 \trpos indeed goes from the tape to the log. Transition $\iamuestwo$ is dual 
 to 
 $\iamdvartwo$, and it is used to keep looking for arguments when the current 
 subterm $\tmtwo$ has none left.
 
 All results and 
 considerations of Section~\ref{sec:iam} and \ref{sec:prop} still hold in this more general 
 setting, \emph{mutatis mutandis}\begin{SHORT}, and can be found in 
 \refsect{app-properties-iam} of the Appendix\end{SHORT}.


%% file: 07-Exhaustible.tex
\section{The Exhaustible State Invariant}\label{sec:exh}
The previous sections introduced all the ingredients for the formal study of the $\IAM$. From now on, we turn to development of the proofs of soundness and adequacy.  The first step, taken here, is to formalize the exhaustible state invariant mentioned in Sect. \ref{sec:prop}.

The intuition behind the invariant is that whenever a
\trpos $\exps$ occurs in a reachable state, it is there \emph{for a
  reason}, because no \trpos occur in initial states,
and transitions only add \trposs to which the machine is supposed to come back. 
In particular, one can somehow \emph{revert} the
process which is responsible for having placed $\exps$ in the state, and
\emph{exhaust} $\exps$.

\paragraph{Why It Is Needed.} The exhaustible state invariant is meant to 
show that some undesirable configurations never arise, to characterize the final states of the 
$\IAM$. On states such as 
$\dstate{ \la\var\ctxtwop\var }{ \ctx }{ \exps\cons\stme }{ \ste  }$ the $\IAM$ requires 
the \trpos $\exps$ to have the shape $(\var, \la\var\ctxtwo, \ste')$, 
that is, to be associated to a position isolating an occurrence of $\var$ in 
$\la\var\ctxtwop\var$, otherwise the machine is stuck. Similarly, on 
states such as 
$\ustate{\tm}{\ctxp{\ctxtwop{x}\esub\var\ctxhole}}{\stme}{\exps\cdot\ste}$ the 
position of $\exps$ is expected to isolate an occurrence of $\var$ in 
$\ctxtwop\var$, or the machine is stuck. Luckily, the machine is never 
stuck 
for these reasons, and exhaustible states are the technical tool to prove it. 
  
  One could redefine the transitions of the $\IAM$ asking---for these states---to jump to whatever variable position is in the \trpos $\exps$. Then the $\IAM$ would not get stuck, and the invariant would not be needed for characterizing final states, but we would then need it for soundness---there is no easy way out. 
  
  \paragraph{First Reading?} Then we suggest to skip this section, as the 
  invariant is involved. It is nonetheless a key technical ingredient and one 
  of the contributions of the paper. The key result used in the rest of the 
  paper is \refcoro{exhaust}.
  
\paragraph{Preliminaries.} Exhaustible states rest on some \emph{tests} for 
their \trposs. More specifically, each \trpos $\exps$ in a 
state $\state$ has an associated test state $s_\exps$, supposed to test the 
exhaustibility of $\exps$ in $\state$. Actually, there shall be \emph{two} 
classes of test states, one accounting for the \trposs in the tape of 
$s$, called
\emph{tape tests}, and one for the
 those in the log of $s$, called \emph{log tests}.

\paragraph{Tape Tests.} Tape tests are easy to define. They focus on one of 
the \trposs in the tape, discarding everything that follows it on the 
tape.
\begin{definition}[Tape tests]
	Let $\state = 
	\nopolstate{\tm}{\ctx_n}{\stmetwo\cons\exps\cons\stmethree}{\ste_n}{\pol}$ 
	be a state. Then the
	\emph{tape test of $\state$ of focus $\exps$} is the state
	$\state_p=\nopolstate{\tm}{\ctx_n}{\stmetwo 
	\cons\exps}{\ste_n}{\upp^{\sizee{\stmetwo \cons\exps}} }$. 
\end{definition}
Note that the direction of tape tests is reversed with respect to that stated 
by the balance invariant, and so, in general, they are not reachable states. 
Such a counter-intuitive fact is needed for the invariant to go through. The 
same shall be true for log tests, introduced next.

\paragraph{Log Tests.} The definition of log tests is more involved. 
The idea is analogous: they focus on a given \trpos in the log. Their 
definition however requires more than simply stripping down the log, as the new 
log and the position still have to form a \trpos---said differently, the 
\emph{position and the log} part of the balance invariant has to be preserved. 
Roughly, when focussing on the $m$-th \trpos $\exps_m$ in the log of a state $ 
\nopolstate{\tm}{\ctx_n}{\stme}{\exps_n\cdots \exps_2\cdot\exps_1}{\pol}$ we 
remove the prefix $\exps_n\cdots \exps_{m+1}$ (if any), and move the current 
position up by $n-m$ levels. Moreover, the tape is emptied and 
the direction is set to $\upp$. Let us define the position change.

Let $(\tmtwo,\lctx{n+1})$ be a position. Then, for 
every
decomposition of $n$ into two natural numbers $m,k$ with
$m+k=n$, we can find contexts $\lctx{m}$ and $\lctx{k}$, and a
  term $\tmthree$ satisfying exactly one of the two following conditions 
  (levels can be incremented in two ways).
\begin{varitemize}
\item
  Case $\tm = \lctxp{m}{\tmthree \lctxp{k}\tmtwo}$.
  Then, the \emph{$m+1$-outer context} of the position $(\tmtwo,\lctx{n+1})$
  is the context $\octx{m+1} \defeq\lctxp{m}{r\ctxhole}$ 
  of level $m+1$ and the \emph{$m+1$-outer position} is
  $(\lctxp{k}\tmtwo,\octx{m+1})$.
\item
  Case $\tm = \lctxp{m}{\tmthree[x\leftarrow\lctxp{k}\tmtwo]}$.
  Then, the \emph{$m+1$-outer context} of the position $(\tmtwo,\lctx{n+1})$
  is the context $\octx{m+1} \defeq\lctxp{m}{r[x\leftarrow\ctxhole]}$ 
  of level $m+1$ and the \emph{$m+1$-outer position} is
  $(\lctxp{k}\tmtwo,\octx{m+1})$.
\end{varitemize}
Note that the $m$-outer context and the $m$-outer position (of a given position) have level $m$.
It is easy to realize that any position having level $n$ has \emph{unique}
$m$-outer context and $m$-outer position, for every $1\leq m\leq n+1$, and that, moreover, outer positions are hereditary, in the following sense: the $i$-outer position of the $m$-outer position of $(\tmtwo,\lctx{n+1})$ is exactly the $i$-outer position of $(\tmtwo,\lctx{n+1})$.
\begin{definition}[Log tests]
  Let $\state =
  \nopolstate{\tm}{\ctx_n}{\stme}{\exps_n\cdots \exps_2\cdot\exps_1}{\pol}$ be a
   state with $1\leq m\leq n$, and $(\tmtwo,\octx{m})$ be the $m$-outer 
  position of
  $(\tm,\lctx{n})$. The \emph{$m$-log test of $\state$ of focus $\exps_m$} is the
  state $\outsi{m}{\state}\defeq\nopolstate{\tmtwo}{\octx{m}}{\stempty}{\exps_m\cdots \exps_2\cdot \exps_1}{\upp}$.
\end{definition}

\begin{LONG}
By definition, log tests for $\state$ do not depend on the direction of
$\state$, nor on the underlying tape, and they are stable by head
translations of the position $(\tm,\ctx_n)$ of $\state$, in the sense that if 
$\tm
= \hctxp\tmthree$ then $\state= \nopolstate{\tm}{\ctx_n}{\stme}{\ste}{\pol}$
and its head translation
$\nopolstate{\tmthree}{\ctx_n\ctxholep\hctx}{\stmetwo}{\ste}{\pol}$
induce the same log tests (because the two positions have the same
outer positions and the two states have the same logs).

\begin{lemma}[Invariance properties of log tests]
  \label{l:outer-inv} 
  Let $\state = \nopolstate{\tm}{\ctx_n}{\stme}{\expsn}{\pol}$ be a
   state. Then:
  \begin{varenumerate}
  \item \label{p:outer-inv-one} \emph{Direction}: the dual
    $\nopolstate{\tm}{\ctx_n}{\stme}{\expsn}{\pol^1}$ of $\state$
    induces the same log tests;
  \item \label{p:outer-inv-two} \emph{Tape}: the state
    $\nopolstate{\tm}{\ctx_n}{\stmetwo}{\expsn}{\pol}$ obtained from
    $\state$ replacing $\stme$ with an arbitrary tape $\stmetwo$
    induces the same log tests;
  \item \label{p:outer-inv-three} \emph{Head translation}: if $\tm =
    \hctxp\tmthree$ then the head translation
    $\nopolstate{\tmthree}{\ctx_n\ctxholep\hctx}{\stmetwo}{\expsn}{\pol}$
    of $\state$ induces the same log tests.
  \item \label{p:outer-inv-four} \emph{Inclusion}: if $\ctx_n = 
  \ctx_m\ctxholep{\ctx_i}$ and $\expsn = \ste_i\cdot \ste_m$
  then the log tests of 
  $\nopolstate{\ctx_i\ctxholep\tm}{\ctx_m}{\stmetwo}{\ste_m}{\pol}$
  are log tests of $\state$.
  \end{varenumerate}
\end{lemma}
\begin{proof}
    The first three points are immediate consequences of the definition
    of log test. We prove the fourth point. Let $\statetwo =
    \nopolstate{\ctx_j\ctxholep\tm}{\ctx_i}{\stme}{\expsind
      i}{\pol}$. By induction on $j$. If $j = 0$ then $i = n$ and
    $\state = \statetwo$, therefore the statement is simply says that
    the log test of $\state$ is $\outs\state$, that is obviously
    true. Let $j>0$. By \ih, the log test $\statethree$ of
    $\nopolstate{\ctx_{j-1}\ctxholep\tm}{\ctx_{i+1}}{\stme}{\exps\cons
      \expsind i}{\pol}$ is $\outsi i \state$. Let us spell out
    $\statethree$. If $\ctx_{i+1} = \ctx_i\ctxholep{\tmtwo\ctx_0}$ then
    $\statethree = \nopolstate{\tmtwo}{\ctx_i\ctxholep{\ctxhole
        \ctx_{j-1}\ctxholep\tm}}{\exps}{\expsind i}{\upp}$. Note that
    $\ctx_j = \ctx_i\ctxholep{\tmtwo\ctxhole}$. Since $\statethree =
    \outsi i \state$, we have $ \outsi {i-1} \state =
    \outs\statethree$. Now, since log tests are stable by head
    translation (Point 3), we have that $\outsi {i-1} \state$ is also
    the log test of the translation of $\statethree$ with respect to
    $\ctx_{j-1}\ctxholep\tm$, that is, of the state $\nopolstate{\tmtwo
      \ctx_{j-1}\ctxholep\tm}{\ctx_i}{\exps}{\expsind i}{\upp} =
    \nopolstate{\ctx_j\ctxholep\tm}{\ctx_i}{\exps}{\expsind i}{\upp}$.
\end{proof}\end{LONG}

Exhausting a log position $\lpos$ means backtracking to it. We decorate the 
backtracking transition $\tomachbtone$ and $\tomachbttwo$ as $\tomachbtonedec$ 
and $\tomachbttwodec$ to specify the 
involved logged position $\lpos$. Finally, we need a notion of state extending 
the context of a \trpos.
\begin{definition}[State surrounding a position]
 Let $\exps=(\tm,\ctxtwo,\stetwo)$ be a \trpos. A  state $\state$ surrounds 
 $\exps$ if $\state = 
 \ustate{\tm}{\ctx_n\ctxholep{\ctxtwo}}{\stempty}{\stetwo\cdot\ste_n}$ for some 
 $\ctx_n$ and $\ste_n$.
\end{definition}

\paragraph{The Exhaustibility Invariant.} After having introduced all the 
necessary preliminaries, we can now  state the the property that we are next 
showing to be invariant.
\begin{definition}[Exhaustible States]
   $\exstates$ is the smallest
   set of states $\state$ such that if $\state_\lpos$ is a tape or a log 
   test of $\state$ of focus $\lpos$, then 
   $\state_\lpos\toiam^*\tomachbttwodec\statethree\in\exstates$, where
   $\statethree$ surrounds $\lpos$. States in $\exstates$ are called 
   \emph{exhaustible}.
\end{definition}
Informally, exhaustible states are those for which every \trpos can be 
successfully tested,  that is, the $\IAM$ can backtrack to (an exhaustible 
state surrounding) it, if properly initialized. Roughly, a state is exhaustible 
if the backtracking information encoded in its \trposs is coherent. The set
$\exstates$ being the \emph{smallest} set of such states implies
that checking that a state is exhaustible can be finitely certified, \ie there must be a finitary proof.
\begin{proposition}[Exhaustible invariant]
\label{prop:good-invariant}
   Let $\state$ be a $\IAM$ reachable state. Then $\state$ is exhaustible.
\end{proposition}
The proof of \refprop{good-invariant} is long, but logically
quite simple, being structured around a simple induction on the length of the 
run from the initial state to $\state$, and can be found in the Appendix.
\begin{LONG}\input{proofs/good-invariant-proof-old}\end{LONG}

\paragraph{Exhaustible and Final States.} We are now ready to prove that the 
$\IAM$ never gets stuck for a mismatch of \trposs.
\begin{corollary}[\TrPoss Never Block the $\IAM$]
\label{coro:exhaust} 
Let $\state$ be a reachable state.
  \begin{varenumerate}
   \item \label{cor:lambda} 
   If $\state = \dstate{ \la\var\ctxtwop\var }{ \ctx }{ \exps\cons\stme }{ 
   \ste
  }$ then $\state$ is not final.
  \item \label{cor:es} 
  If $\state=\ustate{\tmtwo}{\ctx\ctxholep{\ctxtwo\ctxholep{x}[x\leftarrow\ctxhole]}}{\stme}{\exps\cdot\ste}$
  then $\state$ is not final.
  \end{varenumerate}
\end{corollary}
\begin{proof}
For the point 1, by the exhaustible invariant (\refprop{good-invariant}), $\state$ is exhaustible. By successful testing, its tape test
  $\dstate{\la\var\ctxtwop\var}{\ctx}{\exps}{\ste}$ evolves to a state $\statetwo\neq\state$ surrounding $\exps$. Point 2 is analogous, just consider the log test
  $\statetwo=\ustate{\tmtwo}{\ctx\ctxholep{\ctxtwo\ctxholep{x}[x\leftarrow\ctxhole]}}{\stempty}{\exps\cdot\ste}$.
\end{proof}

%% file: proofs/good-invariant-proof-old.tex
\begin{proof}
  Let $\state = \dstate{\tm}{\ctxhole}{\resm^h}{\stempty} \toiam^k \statetwo$. By
  induction on $k$. For $k=0$ there is nothing to prove because the state has no tape nor log tests. Then suppose
  $\state\toiam^{k-1}\statethree\toiam\statetwo$. By \ih, $\statethree
  = \nopolstate{\tmtwo}{\ctx}{\stme}{\ste}{\pol}$ is exhaustible, and with
  this hypothesis we need to conclude that $\statetwo$ is exhaustible, too.
  There are many cases to take into account, depending on the transition used to
  move from $\statethree$ to $\statetwo$. We recall that we use $\sizee{\stme}$ for the number of \trpos in $\stme$, called \emph{position length of $\stme$} in the proof.
  
  First,
  suppose that $\pol = \downp$. Cases of $\statethree\toiam\statetwo$:
  \begin{enumerate}
  \item
    \emph{Application}, \ie\ $\tmtwo = \tmthree\tmfour$ and
	\begin{center}$
    	   \dstate{ \tmthree\tmfour }{ \ctx }{ \stme }{ \ste } 
	   \iamdap 
	   \dstate{ \tmthree }{ \ctxp{\ctxhole\tmfour} }{ \resm\cons \stme }{ \ste } = \statetwo.
	 $\end{center}
	 We have to show that the obtained state $\statetwo$ is
         exhaustible. For log tests, it follows from
         \reflemmap{outer-inv}{three} and the \ih:
         $\statethree$ is a head translation of $\statetwo$, and the
         lemma states that they have the same log tests, which are
         exhaustible because $\statethree$ is exhaustible by \ih
		
	 For tape tests, consider a decomposition $\stme =
         \stmetwo\cons\exps\cons\stmethree$. Two cases, depending on
         the parity of $\sizee{\stmetwo}$:
	 \begin{enumerate}
	 \item
           \emph{$\sizee{\stmetwo}$ is odd}. Then the position
           length of the tape $\resm\cons \stmetwo\cons\exps$ is
           even (occurrences of $\resm$ are ignored) and so the
           direction of the corresponding tape test $\statetwo_\exps$ is
           $\upp$. Note that $\statetwo_\exps$ reduces to a tape test
           $\statethree_\exps$ for $\statethree$ having the same focus $\exps$ of $\statetwo_\exps$:
	   	\begin{center}$\statetwo_\exps = \ustate{\tmthree}{
             \ctxp{\ctxhole\tmfour} }{ \resm\cons \stmetwo\cons\exps}{
             \ste } \iamuapltwo \ustate{\tmthree \tmfour}{ \ctx }{
             \stmetwo\cons\exps}{ \ste } = \statethree_\exps$\end{center}
              By \ih,
           $\statethree$ is exhaustible, and so $\statethree_\exps$
           evolves to an exhaustible state surrounding $\exps$, call
           it $\statefour_\exps$.
           Then 
           $\statetwo_\exps$ evolves to $\statefour_\exps$ and the test is successful.
	 \item
           \emph{$\sizee{\stmetwo}$ is even}. Then $\sizee{\resm\cons
             \stmetwo\cons\exps}$ is odd, and the direction of the corresponding
           tape test $\statetwo_\exps$ os $\statetwo$ is $\downp$. Note that the
           corresponding tape test
           $\statethree_\exps$ of $\statethree$ reduces to
           $\statetwo_\exps$:
	   	\begin{center}$
           \statethree_\exps = \dstate{ \tmthree\tmfour }{ \ctx }{ 
           \stmetwo\cons\exps }{ \ste } \iamuapltwo \dstate{\tmthree}{ 
           \ctxp{\ctxhole\tmfour} }{ \resm\cons \stmetwo\cons\exps}{ \ste } = 
           \statetwo_\exps
           $\end{center}
           By \ih, $\statethree$ is exhaustible, then
           $\statethree_\exps$ evolves to an exhaustible state
           surrounding  $\exps$, call it $\statefour_\exps$.
           The IAM is deterministic,
           so $\statetwo_\exps$ itself
           reduces to $\statefour_\exps$.
	 \end{enumerate}
	 \item \emph{Abstraction 1}, \ie\ $\tmtwo=\la\var\tmthree$ and 
	 $\stme=\resm\cdot\stmetwo$. Identical to the previous one.
    \item \emph{Variable bound by an abstraction}, \ie\ $\tmtwo = \var$ and
      	\begin{center}$\begin{array}{lll}
      \statethree = & \dstate{ \var }{ \ctxp{\l\var.\ctxtwo_n} }{ \stme }{ \expsn\cdot\ste } 
      \\
      \iamdvar &
      \ustate{ \l\var.\ctxtwo_n\ctxholep\var}{ \ctx }{ (\var,\l\var.\ctxtwo_n,\expsn)\cdot\stme }{ \ste } = \statetwo
      	\end{array}$\end{center}
      The proof that $\statetwo$ is exhaustible is divided in two parts:
      \begin{enumerate}
      \item
        \emph{Log testing}. By \reflemmap{outer-inv}{four}, all log tests of $\statetwo$
        are also log tests of $\statethree$. Since the latter is
        exhausible by \ih, then all the log tests of
        $\statetwo$ are successful.
      \item \emph{Tape testing}.
        We need to consider various cases, corresponding to the
        various decompositions of the tape $\expstwo \cons \stme$
        where $\expstwo =
        (\var,\l\var.\ctxtwo_n,\expsn)$:
	\begin{enumerate}
	\item
          \emph{The \trpos to test is $\exps=\expstwo$}, \ie\ the
          first one. We are then considering a prefix of odd length of
          $\expstwo\cons\stme$, so the direction of the corresponding
          tape test $\statetwo_\exps$ is $\downp$. Observe, however, that
          by definition
      	\begin{center}$\begin{array}{lll}
          \statetwo_\exps = & \dstate{ \l\var.\ctxtwo_n\ctxholep\var}{ \ctx }{ (\var,\l\var.\ctxtwo_n,\expsn)}{ \ste }
          \\
          \iamdlamtwo & \ustate{ \var}{ \ctxp{\l\var.\ctxtwo_n} }{ \stempty}{ \expsn\cons\ste} = \statethree^\bot
      	\end{array}$\end{center}
          where $\statethree^\bot$ is trivially surrounding $\exps$.
          Moreover, by \ih, $\statethree$ is exhaustible, a property
          which is easily transferred to $\statethree^\bot$: the log tests are the same by \reflemmap{outer-inv}{one}, while $\statethree^\bot$ satisfies tape testing
          trivially, because the tape is empty.
	\item
          \emph{The prefix $\stmetwo\cons\exps$ of the tape has
            even length} and the direction of the corresponding tape test
          $\statetwo_\exps$ is $\upp$. Let $\stmetwo
          =(\var,\l\var.\ctxtwo_n,\expsn)\cons\stmethree$. Note
          that the corresponding tape test $\statethree_\exps$ of $\statethree$
          reduces to $\statetwo_\exps$:
      	\begin{center}$\begin{array}{lll}
          \statethree_\exps = & \dstate{ \var }{ \ctxp{\l\var.\ctxtwo_n}}
          {\stmethree\cons\exps}{ \expsn\cdot\ste } 
          \\
          \iamdvar &
          \ustate{ \l\var.\ctxtwo_n\ctxholep\var}{ \ctx }{
          (\var,\l\var.\ctxtwo_n,\expsn)\cons\stmethree\cons\exps}{
           \ste } = \statetwo_\exps
      	\end{array}$\end{center}
          By \ih, $\statethree$ is exhaustible, then
          $\statethree_\exps$ evolves to an exhaustible state
          surrounding $\exps$, call it $\statefour_\exps$.
          The IAM is deterministic,
          so $\statetwo_\exps$ itself
          reduces to $\statefour_\exps$, and the test is successful.
	\item
          \emph{The prefix $\stmetwo\cons\exps$ of the tape has
            odd strictly positive length} and the direction of the corresponding
          log test
          $\statetwo_\exps$ is $\downp$. Let $\stmetwo
          =(\var,\l\var.\ctxtwo_n,\expsn)\cons\stmethree$. Note
          that $\statetwo_\exps$ reduces to the corresponding log test
          $\statethree_\exps$ of $\statethree$:
	        	\begin{center}$\begin{array}{lll}
		\statetwo_\exps = &
		\dstate{ \l\var.\ctxtwo_n\ctxholep\var}{ \ctx }{ (\var,\l\var.\ctxtwo_n,\expsn)\cons\stmethree\cons\exps }{ \ste } 
	  \\
	  \iamdlamtwo &
	  \ustate{ \var }{ \ctxp{\l\var.\ctxtwo_n} }{\stmethree\cons\exps}{ \expsn\cdot\ste }  = \statethree_\exps
	        	\end{array}$\end{center}
          We can then proceed as usual using the \ih
	\end{enumerate}
	\end{enumerate}
	
    \item
      \emph{Abstraction 2}, \ie\ $\tmtwo = \la\var\tmthree$ and
      	\begin{center}$\begin{array}{lll}
      \statethree = &
      \dstate{ \la\var\tmthree }{ \ctx }{
        (\var,\la\var\ctxtwo_n,\stetwo)\cons\stme }{ \ste }
      \\
      \iamdlamtwo & \ustate{ \var}{ \ctxp{\la\var\ctxtwo_n} }{ \stme }{
        \stetwo \cons\ste } = \statetwo
            	\end{array}$\end{center}
      \begin{enumerate}
      \item
        \emph{Log testing}.  Let $\exps =
        (\var,\ctxp{\la\var\ctxtwo_n},\stetwo)$ and note that the tape tests of $\statethree$ of focus $\exps$ does the following transition:
      	\begin{center}$\begin{array}{lll}
        \statethree_\exps = &
        \dstate{ \la\var\tmthree }{ \ctx }{ (\var,\la\var\ctxtwo_n,\stetwo) }{ \ste } 
        \\
	\iamdlamtwo &
	\ustate{ \var}{ \ctxp{\la\var\ctxtwo_n} }{ \stempty}{ \stetwo \cons\ste }  = \statetwo_\stempty
              	\end{array}$\end{center}
        Now, $\statetwo_\stempty$ surrounds $\exps$ and thus, by \ih (tape testing of $\statethree$), $\statetwo_\stempty$ is exhaustible. By \reflemmap{outer-inv}{two}, $\statetwo_\stempty$ and $\statetwo$ have the same log tests, so log testing for $\statetwo$ holds because it does for $\statetwo_\stempty$.
      \item
        \emph{Tape testing}. As usual, we have to consider various cases,
        corresponding to the possible decompositions $\stme = \stmetwo\cons\exps\cons\stmethree$
        of the tape.
	\begin{enumerate}
	\item
          \emph{$\sizee{\stmetwo}$ is odd}, so that the prefix
          $\stmetwo\cons\exps$ of the tape has even length and the
          direction of the tape test  $\statetwo_\exps$ corresponding
          to $\exps$ is $\upp$. Note
          that the tape test $\statethree_\exps$ of $\statethree$
          reduces to the corresponding tape test $\statetwo_\exps$ of $\statetwo$:
      	\begin{center}$\begin{array}{lll}
          \statethree_\exps = & \dstate{ \l\var.\ctxtwo_n\ctxholep\var}{ \ctx }{ (\var,\l\var.\ctxtwo_n,\stetwo)\cons\stmetwo\cons\exps }{ \ste } 
          \\
	  \iamdlamtwo &
	  \ustate{ \var }{ \ctxp{\l\var.\ctxtwo_n} }{\stmetwo\cons\exps}{ \stetwo\cdot\ste } 
	  = \statetwo_\exps
              	\end{array}$\end{center}
          We can then proceed as usual, exploiting the determinism of
          the $\IAM$ and the \ih
        \item
          \emph{$\sizee{\stmetwo}\neq 0$ is even}, so that the prefix
          $\stmetwo\cons\exps$ of the tape has odd length and the
          direction of the tape test $\statetwo_\exps$ corresponding
          to $\exps$ is
          $\downp$. Note that $\statetwo_\exps$ reduces to the corresponding
          tape test $\statethree_\exps$ of $\statethree$:
      	\begin{center}$\begin{array}{lll}
          \statetwo_\exps = &\dstate{ \var }{
            \ctxp{\l\var.\ctxtwo_n} }{\stmetwo\cons\exps}{
            \stetwo\cdot\ste } 
            \\
            \iamdvar & \ustate{
            \l\var.\ctxtwo_n\ctxholep\var}{ \ctx }{
            (\var,\l\var.\ctxtwo_n,\stetwo)\cons\stmetwo\cons\exps}{
            \ste } = \statethree_\exps
              	\end{array}$\end{center}
          Again, we can then proceed as usual using the \ih
	\end{enumerate}
      \end{enumerate}
      
      \item \emph{Explicit Substitution}, \ie $\tmtwo = \tmthree\esub\var\tmfour$ and
      	\begin{center}$
	    \statethree = \dstate{ \tmthree\esub\var\tmfour }{ \ctx }{ \stme }{ \ste } 
    \iamdes 
    \dstate{ \tmthree }{ \ctxp{\ctxhole\esub\var\tmfour} }{ \stme }{ \ste } = \statetwo
    $\end{center}
 For log testing, it follows from
         \reflemmap{outer-inv}{three} and the \ih:
         $\statethree$ is a head translation of $\statetwo$, and the
         lemma states that they have the same log tests, which are
         exhaustible because $\statethree$ is exhaustible by \ih
		
	 For tape testing it goes exactly as the application case. We spell it out anyway. Consider a decomposition $\stme =
         \stmetwo\cons\exps\cons\stmethree$. Two cases, depending on
         the parity of $\sizee{\stmetwo}$:
	 \begin{enumerate}
	 \item
           \emph{$\sizee{\stmetwo}$ is odd}. Then the position
           length of the tape $\stmetwo\cons\exps$ is
           even and so the
           direction of the corresponding tape test $\statetwo_\exps$ is
           $\upp$. Note that $\statetwo_\exps$ reduces to a tape test
           $\statethree_\exps$ for $\statethree$:
	         	\begin{center}$\begin{array}{lll}
		\statetwo_\exps = & \ustate{\tmthree}{
             \ctxp{\ctxhole\esub\var\tmfour} }{  \stmetwo\cons\exps}{
             \ste } 
             \\
             \iamues & \ustate{\tmthree \esub\var\tmfour}{ \ctx }{
             \stmetwo\cons\exps}{ \ste } = \statethree_\exps
              	\end{array}$\end{center}
                       Again, we then proceed as usual using the \ih
                       
	 \item
           \emph{$\sizee{\stmetwo}$ is even}. Then $\sizee{
             \stmetwo\cons\exps}$ is odd, and the direction of the corresponding
           tape test $\statetwo_\exps$ os $\statetwo$ is $\downp$. Note that the
           corresponding tape test
           $\statethree_\exps$ of $\statethree$ reduces to
           $\statetwo_\exps$:
      	\begin{center}$\begin{array}{lll}
           \statethree_\exps = & \dstate{ \tmthree\esub\var\tmfour }{ \ctx }{ \stmetwo\cons\exps }{ \ste } 
           \\
           \iamdes & \dstate{\tmthree}{ \ctxp{\ctxhole\esub\var\tmfour} }{ \stmetwo\cons\exps}{ \ste } = \statetwo_\exps
              	\end{array}$\end{center}
           Again, we then proceed as usual, exploiting the determinism of
          the $\IAM$ and the \ih
	 \end{enumerate}

	\item \emph{Variable bound by an explicit substitution}, \ie $\tmtwo = \var$ and
	\begin{center}$\begin{array}{lll}
	\statethree = &\dstate{ \var }{ \ctxp{\ctxtwo_n\esub\var\tmthree} }{ \stme }{ \expsn\cdot\ste } 
	\\
	 \iamdvartwo&
	\dstate{ \tmthree }{ \ctxp{\ctxtwo_n\ctxholep\var\esub\var\ctxhole} }{ \stme }{ (\var,\ctxtwo_n\esub\var\tmthree,\expsn)\cdot\ste } = \statetwo
	\end{array}$\end{center}
		\begin{enumerate}
			\item \emph{Log testing}: let $\exps \defeq (\var, \ctxtwo_n\esub\var\tmthree,\expsn)$ and $m = \size{\exps\cdot\ste}$. The $m$-log test of $\statetwo$ is 
				\begin{center}
				$\outsi{m}{\state} = \ustate{ \tmthree }{ \ctxp{\ctxtwo_n\ctxholep\var\esub\var\ctxhole} }{ \epsilon }{ (\var,\ctxtwo_n\esub\var\tmthree,\expsn)\cdot\ste }$
				\end{center}
			which makes a transition
				\begin{center}
				$\iamuestwo \ustate{ \var }{ \ctxp{\ctxtwo_n\esub\var\tmthree} }{ \epsilon }{ \expsn\cdot\ste } = (\statethree_\epsilon)^\bot$
				\end{center}
			that is a state surrounding $\exps$, as required by log testing. We have to prove that $(\statethree_\epsilon)^\bot$ is exhaustible. Tape testing is trivial, because the tape is empty. Log testing follows from the \ih\ and the fact that $(\statethree_\epsilon)^\bot$ is $\statethree$ with reversed direction and without the tape, and so by \reflemmap{outer-inv}{one} and \reflemmap{outer-inv}{two} they have the same log tests.
			
			Note that the $i$-log tests of $\statetwo$ for $i<m$ are the $i$-log tests of $\statetwo$ (\reflemmap{outer-inv}{four}), and so they satisfy the log testing clause by the \ih
			
			\item \emph{Tape testing}:  it goes exactly as in the previous ordinary cases. We spell it out anyway. Consider a decomposition $\stme =
         \stmetwo\cons\exps\cons\stmethree$. Two cases, depending on
         the parity of $\sizee{\stmetwo}$:
	 \begin{enumerate}
	 \item
           \emph{$\sizee{\stmetwo}$ is odd}. Then the position
           length of the tape $\stmetwo\cons\exps$ is
           even and so the
           direction of the corresponding tape test $\statetwo_\exps$ is
           $\upp$. Note that $\statetwo_\exps$ reduces to a tape test
           $\statethree_\exps$ for $\statethree$:
	\begin{center}$\begin{array}{lll}
	\statetwo_\exps = &
	   \ustate{ \tmthree }{ \ctxp{\ctxtwo_n\ctxholep\var\esub\var\ctxhole} }{ \stmetwo\cons\exps }{ (\var,\ctxtwo_n\esub\var\tmthree,\expsn)\cdot\ste }
\\	    \iamuestwo  &
	    \ustate{ \var }{ \ctxp{\ctxtwo_n\esub\var\tmthree} }{ \stmetwo\cons\exps }{ \expsn\cdot\ste } = \statethree_\exps\end{array}$\end{center}
                       Again, we then proceed as usual using the \ih
                       
	 \item
           \emph{$\sizee{\stmetwo}$ is even}. Then $\sizee{
             \stmetwo\cons\exps}$ is odd, and the direction of the corresponding
           tape test $\statetwo_\exps$ os $\statetwo$ is $\downp$. Note that the
           corresponding tape test
           $\statethree_\exps$ of $\statethree$ reduces to
           $\statetwo_\exps$:
	   	\begin{center}$\begin{array}{lll}
		\statethree_\exps = &
           \dstate{ \var }{ \ctxp{\ctxtwo_n\esub\var\tmthree} }{ \stmetwo\cons\exps }{ \expsn\cdot\ste } 
           \\
	\iamdvartwo &
	\dstate{ \tmthree }{ \ctxp{\ctxtwo_n\ctxholep\var\esub\var\ctxhole} }{ \stmetwo\cons\exps }{ (\var,\ctxtwo_n\esub\var\tmthree,\expsn)\cdot\ste }   =\statetwo_\exps
	              	\end{array}$\end{center}
           Again, we then proceed as usual, exploiting the determinism of
          the $\IAM$ and the \ih
	 \end{enumerate} 

      	\end{enumerate}
      \end{enumerate}
        Now, suppose that $\pol = \upp$. Cases of
        $\statethree\toiam\statetwo$:
	\begin{enumerate}
		\item \emph{Coming from the left of an application}, \ie\ $\ctx = \ctxtwop{\ctxhole \tmthree}$ and 
		  	\begin{center}$
                  \statethree = \ustate{ \tmtwo }{ \ctxtwop{\ctxhole\tmthree} }{ \exps\cdot\stme }{ \ste } 
                  \iamuaplone 
                  \dstate{ \tmthree }{ \ctxtwop{\tmtwo\ctxhole} }{ \stme }{ \exps\cdot\ste } = \statetwo.
                  $\end{center}
	          The proof that $\statetwo$ is exhaustible is divided in two parts:
	          \begin{enumerate}
		  \item
                    \emph{Log testing}. The log tests of
                    $\statetwo$ are those of $\statethree$ plus\\
                    $\ustate{ \tmthree }{ \ctxtwop{\tmtwo\ctxhole} }{ \epsilon}{ \exps\cdot\ste }$.
                    The former are fine because of the \ih, while
                    about the latter, observe that
                    $\ustate{ \tmthree }{ \ctxtwop{\tmtwo\ctxhole} }{ \epsilon}{ \exps\cdot\ste }$
                    evolves to $\dstate{ \tmtwo }{ \ctxtwop{\ctxhole\tmthree} }{ \exps }{ \ste }$
                    which is a tape test of $\statethree$. The thesis
                    easily follows by \ih
		
		\item \emph{Tape testing}.
                    Let $\stmetwo$ be a prefix of $\stme$ such that $\stmetwo = \stmethree\cons\expstwo$. Two cases:
		    \begin{enumerate}
		    \item
                      \emph{$\sizee\stmetwo$ is odd}, and the direction
                      is $\downp$. Note that the tape test
                      $\statethree_\exps$ of $\statethree$ corresponding
                      to $\exps$ reduces to
                      a tape test $\statetwo_\exps$ of $\statetwo$:
				   	\begin{center}$\begin{array}{lll}
                      \statethree_\exps = &\ustate{ \tmthree }{ \ctxtwop{\tmtwo\ctxhole} }{ \exps\cons\stmetwo }{ \ste }
                      \\
                      \iamuaplone &
		      \dstate{ \tmthree }{ \ctxtwop{\tmtwo\ctxhole} }{ \stmetwo }{ \exps\cdot\ste } = \statetwo_\exps
	              	\end{array}$\end{center}
                      We can then proceed as usual, using the \ih~ and determinism of the IAM.
		    \item
                      \emph{$\sizee\stmetwo$ is even}, and the direction is $\upp$. Note that $\statetwo_\exps$ reduces to the
                      corresponding tape test $\statethree_\exps$ of $\statethree$: 
				   	\begin{center}$\begin{array}{lll}
                      \statetwo_\exps = & \ustate{ \tmthree }{ \ctxtwop{\tmtwo\ctxhole} }{ \stmetwo }{ \exps\cdot\ste }
                      \\
		      \iamuapr &
		      \dstate{ \tmthree }{ \ctxtwop{\tmtwo\ctxhole} }{ \exps\cons\stmetwo }{ \ste }
		      =\statethree_{\exps}
	              	\end{array}$\end{center}
                      Again, we can proceed as usual, using the \ih
		    \end{enumerate}
	          \end{enumerate}
	          
		\item \emph{Coming from the right of an application}, \ie\ $\ctx = \ctxtwop{\tmthree \ctxhole}$ and 
						   	\begin{center}$    \statethree = \ustate{ \tmtwo }{ \ctxtwop{\tmthree\ctxhole} }{ \stme }{ \exps\cdot\ste }
                  \iamuapr 
                  \dstate{ \tmthree }{ \ctxtwop{\ctxhole\tmtwo} }{ \exps\cdot\stme }{ \ste } = \statetwo.
                  $\end{center}
	          The proof that $\statetwo$ is exhaustible is divided in two parts:
	          \begin{enumerate}
		  \item
                    \emph{Log testing}:
                    the log tests of $\statetwo$
                    are among the log tests of $\statethree$,
                    so log testing follows from \ih
		  \item \emph{Tape testing}. Let $\stmetwo$ be a
                    prefix of $\stme$. Two cases:
		    \begin{enumerate}
		    \item \emph{$\stmetwo = \stme$ is empty}.
                      So that the tape contains only $\exps$, its
                      length is odd, and the direction is $\downp$. The state to be proven exhaustible is                      
                      $$\statetwo_\exps = \dstate{ \tmthree }{
                        \ctxtwop{\ctxhole\tmtwo} }{ \exps }{ \ste }$$
                      Now, note that the log test $\outsi{\size{\exps\cdot\ste}}\statethree$ of $\statethree$ reduces in one step to $\statetwo_\exps$:
				   	\begin{center}$\begin{array}{rll}
                      \outsi{\size{\exps\cdot\ste}}\statethree = & \ustate{ \tmtwo }{ \ctxtwop{\tmthree\ctxhole} }{ \epsilon}{ \exps\cdot\ste } 
                      \\
                      \iamuapr &
                      \dstate{ \tmthree }{
                        \ctxtwop{\ctxhole\tmtwo} }{ \exps }{ \ste }
	              	\end{array}$\end{center}
                      By log testing for $\statethree$, there is a state $q_\exps$ surrounding $\exps$ such that $\outsi{\size{\exps\cdot\ste}}\statethree \toiam^* q_\exps$. By determinism of the IAM, $\statetwo_\exps \toiam^* q_\exps$.
		    \item
                      \emph{$\stmetwo \neq \stme$ is non-empty}. Then
                      $\stmetwo = \stmethree\cons\expstwo$ Two cases:
		      \begin{enumerate}
		      \item
                        \emph{$\sizee{\stmethree\cons\expstwo}$ is
                          even}, so that the tape
                        $\exps\cons\stmethree\cons\expstwo$ has odd
                        length and the direction is $\downp$.  Note
                        that the tape test $\statethree_{\expstwo}$
                        corresponding to $\expstwo$
                        of $\statethree$ reduces to the tape test
                        $\statetwo_{\expstwo}$ corresponding to
                        $\expstwo$ of $\statetwo$:
				   	\begin{center}$\begin{array}{rll}
                        \statethree_{\expstwo} = &
			\ustate{ \tmthree }{ \ctxtwop{\ctxhole \tmtwo} }{ \stmethree\cons\expstwo }{ \exps\cons\ste }
			\\
			\iamuapr &
			 \dstate{ \tmthree }{ \ctxtwop{\ctxhole\tmtwo} }{ \exps\cons\stmethree\cons\expstwo }{ \ste }
			 = \statetwo_{\expstwo}
	              	\end{array}$\end{center}
                         In this case, as usual, we can conclude
                         by determinism of the $\IAM$.
		       \item
                         \emph{$\sizee\stmetwo$ is odd}, so that the
                         tape $\exps\cons\stmethree\cons\expstwo$
                         has even length and the direction is
                         $\upp$. Note that $\statetwo_{\expstwo}$
                         reduces to the corresponding tape test
                         $\statethree_{\expstwo}$ of $\statethree$:
				   	\begin{center}$\begin{array}{rll}
                         \statetwo_{\expstwo} = &
			\ustate{ \tmthree }{ \ctxtwop{\ctxhole\tmtwo} }{ \exps\cons\stmethree\cons\expstwo }{ \ste }
			\\
			\iamuaplone &
			\dstate{ \tmthree }{ \ctxtwop{\ctxhole \tmtwo} }{ \stmethree\cons\expstwo }{ \exps\cons\ste }
			=\statethree_{\expstwo}
	              	\end{array}$\end{center}
                        Again, the usual scheme allows us to conclude
                        that tape testing holds.                        
		\end{enumerate}
		\end{enumerate}
		
	\end{enumerate}

\item \emph{Explicit Substitution}
\begin{center}$\statethree = \ustate{ \tmtwo }{ \ctxp{\ctxhole\esub\var\tmthree} }{ \stme }{ \ste } 
	\iamues 
	\ustate{ \tmtwo\esub\var\tmthree }{ \ctx }{ \stme }{ \ste } = \statetwo
$\end{center}
\begin{enumerate}
		\item \emph{Log testing}: by \reflemmap{outer-inv}{three} (head translation), the log tests of $\statetwo$ are log tests of $\statethree$, which satisfy log testing by the \ih
	
	\item \emph{Tape testing}: it goes exactly as for the other ordinary cases (\ih, plus determinism in one of the two sub-cases).
	\end{enumerate}
	
\item \emph{Coming from inside an explicit substitution}:
\begin{center}$\begin{array}{rll}
\statethree = & \ustate{ \tmtwo }{ \ctxp{\ctxtwop\var\esub\var\ctxhole} }{ \stme }{ (\var, \ctxtwo[\var\leftarrow\tmtwo],\stetwo)\cdot\ste } 
\\
	\iamuestwo &
	\ustate{ \var }{ \ctxp{\ctxtwo[\var\leftarrow\tmtwo]} }{ \stme }{ \stetwo\cdot\ste } = \statetwo
\end{array}$\end{center}
	\begin{enumerate}
		\item \emph{Log testing}: by \ih, $\statethree$ is exhaustible, and its $\size\ste +1$-log test evolves to 
		\begin{center}$\begin{array}{rll}
		\outsi{\size\ste +1}\statethree = & \ustate{ \tmtwo }{ \ctxp{\ctxtwop\var\esub\var\ctxhole} }{ \epsilon }{ (\var, \ctxtwo[\var\leftarrow\tmtwo],\stetwo)\cdot\ste } 
		\\
	\iamuestwo &
	\ustate{ \var }{ \ctxp{\ctxtwo[\var\leftarrow\tmtwo]} }{ \epsilon }{ \stetwo\cdot\ste } = \statetwo_\epsilon
	   	\end{array}$\end{center}
		which is exhaustible. By \reflemmap{outer-inv}{two}, $\statetwo_\epsilon$ and $\statetwo$ have the same log tests, which are then successful.
	
	\item \emph{Tape testing}: since the tape is unaffected by the transition, this case goes exactly as the other ordinary ones.
	\end{enumerate}		
	\end{enumerate}
\end{proof}

%% file: 08-Improvements_Abstractly.tex
\section{Improvements, Abstractly}\label{sec:proof-sound}
We now introduce \emph{improvements}, a refinement of the classical notion of bisimulation inspired by Sands~\cite{DBLP:journals/toplas/Sands96}. They are the main tool for the proofs of soundness and adequacy of the $\IAM$. 

An improvement is a
weak bisimulation between two transition systems preserving termination and guaranteeing 
that, whenever $s$ and $q$ are related and terminating, then $q$ terminates in no 
more steps than $s$---the \emph{no-more-steps} part implies that the definition 
is asymmetric in the way it treats the two transition systems.

\paragraph{Preliminaries for Bisimulations.} A deterministic transition system (DTS) 
is a pair $\mathcal{S} =(S,\mathcal{T})$, where $S$ is a set of \emph{states} and 
$\mathcal{T}:S\rightharpoonup S$ a partial function.
If $\mathcal{T}(s)=s'$, then we write $s\rightarrow s'$, and if $s$ rewrites in 
$s'$ in $n$ steps then we write $s\rightarrow^n s'$. We note with $\mathcal{F}_S$ the set of final states, \ie 
the subset of $\mathcal{S}$ containing all $s\in\mathcal{S}$ such that 
$\mathcal{T}(s)$ is undefined. A state $s$ is terminating if there exists $n\geq 0$ and 
$s'\in\mathcal{F}_S$ such that $s\rightarrow^n s'$. We call $S_\downarrow$ the 
set of terminating states of $S$ and by $S_\uparrow$ we denote $S\setminus 
S_\downarrow$. The \emph{evaluation length map} $|\cdot|:S\rightarrow\mathbb{N}\cup\{\infty\}$ is defined as $\size s\defeq n$ if $s\rightarrow^n s'$ and $s'\in\mathcal{F}_\mathcal{S}$, 
	and $\size s \defeq \infty$ if $s\in\mathcal{S}_\uparrow$.

\begin{definition}[Improvement]
	Given two DTS $\mathcal{S}$ and $\mathcal{Q}$, a relation $\mathcal{R}\subseteq S\times Q$ is 
	an \emph{improving bisimulation}, or simply an \emph{improvement}, if 
	$(s,q)\in\mathcal{R}$ implies the followings, schematized in 
	Fig.~\ref{figure:TP}.
	\begin{itemize}
		\item \emph{Final state left}: if $s\in\mathcal{F}_\mathcal{S}$, then $q\in\mathcal{F}_\mathcal{Q}$.
		\item \emph{Final state right}: if $q\in\mathcal{F}_\mathcal{Q}$, then $s\rightarrow^n s'$, for some 
		$s'\in\mathcal{F}_\mathcal{S}$ and $n\geq 0$.
		\item \emph{Transition left}: if $s\rightarrow s'$, then there exists $s'',q',n,m$ such 
		that $s'\rightarrow^m s''$, $q\rightarrow^n q'$, 
		$s''\mathcal{R}q'$ and $n\leq m+1$.
		\item \emph{Transition right}: if $q\rightarrow q'$, then there exists $s',q'',n,m$ such that 
		$s\rightarrow^m s'$, $q'\rightarrow^n q''$, $s'\mathcal{R}q''$ and 
		$m\geq n+1$.
	\end{itemize}
\end{definition}

\begin{figure}
\begin{center}
				\begin{tabular}{cc|ccc}
				\begin{tikzpicture}[node distance=30mm, auto, transform 
				shape,scale=1]
				\node (p) at (0,0) {$s$};
				\node (q) at (3,0) {$q$};
				\node (w) at (0,-.8) {$s'$};
				\node (t) at (0,-1.6) {$s''$};
				\node (r) at (3,-1.6) {$q'$};
				\node at (1.5,0) {$\mathcal{R}$};
				\node at (1.5,-1.6) {$\mathcal{R}$};
				\draw (p) edge[->] node {} (w);
				\draw (w) edge[->,dashed] node {$^m$} (t);
				\draw (q) edge[->,left, dashed] node {$^{n\leq m+1}$} (r);
				\end{tikzpicture}
				&&&
				\begin{tikzpicture}[node distance=30mm, auto, transform 
				shape,scale=1]
				\node (p) at (0,0) {$s$};
				\node (q) at (3,0) {$q$};
				\node (w) at (0,-1.6) {$s'$};
				\node (t) at (3,-.8) {$q'$};
				\node (r) at (3,-1.6) {$q''$};
				\node at (1.5,0) {$\mathcal{R}$};
				\node at (1.5,-1.6) {$\mathcal{R}$};
				\draw (p) edge[->,dashed] node {$^{m\geq n+1}$} (w);
				\draw (q) edge[->] node {} (t);
				\draw (t) edge[->,left,dashed] node {$^n$} (r);
				\end{tikzpicture}
				\end{tabular}
		\end{center}
		\vspace{-8pt}
		\caption{Diagrammatic definition of improvements.}
		\label{figure:TP}
	\end{figure}
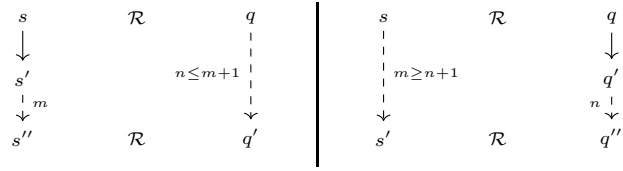
	
What improves along an improvement is the number of transitions required to reach a final state, if any.

\begin{proposition}\label{prop:imp}
	Let $\mathcal{R}$ be an improvement on two DTS $\mathcal{S}$ and $\mathcal{Q}$, and $s\mathcal{R}q$. 
	\begin{enumerate}
	 \item \label{prop:termination}
	 \emph{Termination equivalence}: $s\in\mathcal{S}_\downarrow$ if and only if $q\in\mathcal{Q}_\downarrow$.
	 
	 \item \label{lemma:improv}
	 \emph{Improvement}: $|s|\geq|q|$.
	\end{enumerate}

\end{proposition}
\begin{LONG}\begin{proof}
		\hfill
		\begin{enumerate}
			\item 	$\Rightarrow$. Let us suppose $s\in\mathcal{S}_\downarrow$ and let $n$ be the number of steps that $s$ needs to terminate. We proceed by induction on $n$. If $n=0$, $s\in\mathcal{F}_\mathcal{S}$ and since $s\mathcal{R}q$, $q\in\mathcal{F}_\mathcal{Q}$ and thus $q\in\mathcal{Q}_\downarrow$. If $n=h>0$, then $s\rightarrow s'$, and thus there exists $s'',q',k,j$ such that $q\rightarrow^k q'$, $s'\rightarrow^j s''$, $s''\mathcal{R}q'$ and $k\leq j+1$. Since $s''$ terminates in less than $h-1$ steps, by induction hypothesis $q'\in\mathcal{Q}_\downarrow$ and thus also $q\in\mathcal{Q}_\downarrow$.
			
			$\Leftarrow$. Let us suppose $q\in\mathcal{Q}_\downarrow$ and let $n$ be the number of steps that $q$ needs to terminate. We proceed by induction on $n$. If $n=0$, $q\in\mathcal{F}_\mathcal{Q}$ and since $s\mathcal{R}q$,  $s\in\mathcal{S}_\downarrow$. If $n=h>0$, then $q\rightarrow q'$, and thus there exists $s',q'',k,j$ such that $s\rightarrow^k s'$, $q'\rightarrow^j q''$, $s'\mathcal{R}q''$ and $k\geq j+1$. Since $q''$ terminates in less than $h$ steps, by induction hypothesis $s'\in\mathcal{S}_\downarrow$ and thus also $s\in\mathcal{S}_\downarrow$.
			\end{enumerate}

			\begin{enumerate}
			\item[2.] If $s\in\mathcal{S}_\uparrow$ and $q\in\mathcal{Q}_\uparrow$, then $|s|=|q|=\infty$. Let us consider the other case, \ie when $s\in\mathcal{S}_\downarrow$ and $q\in\mathcal{Q}_\downarrow$. We proceed by induction on $|s|$. If $|s|=0$, then $q\in\mathcal{F}_\mathcal{Q}$ and thus also $|q|=0$. If $|s|=n>0$, then $s\rightarrow s'$ and there exists $s'',q',m,l$ such that $q\rightarrow^m q'$, $s'\rightarrow^l s''$, $s''\mathcal{R}q'$ and $m\leq l+1$. By \ih, $|s''|\geq|q'|$. Thus, since $m\leq l+1$, then $|s|=|s''|+l+1\geq|q'|+m=|q|$.
		\end{enumerate}\end{proof}\end{LONG}

%% file: 09-Improvements_Concretely.tex
\section{Improvements, Concretely}
\label{sect:improv-concr}
In this section we define an improvement $\relf$ relation for the $\IAM$, to be used in the sequel to prove soundness and adequacy.

Given a $\tohl$-step $\tm \tohl \tmtwo$, the improvement $\relf$ has to relate states of code $\tm$ with states of code $\tmtwo$. Since $\tohl$ is the union of the three rewriting rules $\tohldb$, $\tohlls$ and $\tohlgc$, we are going to define $\relf$ as the union of three improvements $\relm$, $\rele$, and $\relgc$. The most interesting and subtle case is $\rele$. To explain it, we start by discussing some of the aspects of $\relm$, which is simpler.

Lifting a step $\tm \tohl 
\tmtwo$ to a relation between a $\IAM$ state $s$ of code $\tm$ and a state $q$ 
of code $\tmtwo$ requires changing all positions relative to $\tm$ in $s$ to 
positions relative to $\tmtwo$ in $q$. A first point to note is that we also 
have to change all the positions in the token, so that $\rel$ has 
to relate positions, \trposs, tape, log, and states.


\paragraph{Explaining the Need of Context Rewriting Using $\tohldb$}
A second more technical aspect is that one needs to extend linear head 
evaluation to contexts. Consider a step $\tm \tohldb \tmtwo$ where---for 
simplicity---the redex is at top level and the associated state 
$\dstate{\sctxp{\la x\tmthree}\tmfour}{\ctxhole}{\epsilon}{\epsilon}$ has an 
empty token. This should be $\relm$-related to a state 
$\dstate{\sctxp{\tmthree\esub\var\tmfour}}{\ctxhole}{\epsilon}{\epsilon}$. 
Let's have a look at how the two states evolve:
\[\begin{tikzpicture}[node distance=30mm, auto, transform 
shape,scale=1.2]
\node (p) at (0,0) {$\dstate{\sctxp{\la x\tmthree}\tmfour}{\ctxhole}{\epsilon}{\epsilon}$};
\node (q) at (4,0) {$\dstate{\sctxp{\tmthree\esub\var\tmfour}}{\ctxhole}{\epsilon}{\epsilon}$};
\node (y) at (4,-1) {$\dstate{\tmthree\esub\var\tmfour}{\sctx}{\epsilon}{\epsilon}$};
\node (w) at (0,-.66) {$\dstate{\sctxp{\la\var\tmthree}}{\ctxhole\tmfour}{\resm}{\epsilon}$};
\node (t) at (0,-1.33) {$\dstate{\la\var\tmthree}{\sctx\tmfour}{\resm}{\epsilon}$};
\node (e) at (0,-2) {$\dstate{\tmthree}{\sctxp{\la\var\ctxhole}\tmfour}{\epsilon}{\epsilon}$};
\node (r) at (4,-2) {$\dstate{\tmthree}{\sctxp{\ctxhole\esub\var\tmfour}}{\epsilon}{\epsilon}$};
\node at (2,0) {$\relm$};
\draw (p) edge[->] node {} (w);
\draw (w) edge[->] node {$^{|{\sctx}|}$} (t);
\draw (t) edge[->] node {} (e);
\draw (q) edge[->] node {$^{|{\sctx}|}$} (y);
\draw (y) edge[->] node {} (r);
\end{tikzpicture}
\]
To close the diagram, we need $\relm$ to relate the two bottom states. Note that their relation can be seen as a $\tohldb$ step involving the contexts of the two positions. Therefore we extend the definition of $\tohldb$ to contexts adding the following top level clause (then included in $\tohldb$ via a closure by head contexts): $\sctxp{\la\var\ctx}\tm \mapsto_\db \sctxp{\ctx\esub\var\tm}$. The new clause, in turn, requires a further extension of $\tohldb$ (again closed by head contexts): 
$\sctxp{\la\var\tm}\ctx \mapsto_\db \sctxp{\tm\esub\var\ctx}$.

Note that in the shown local bisimulation diagram the right side is shorter. 
This is typical of when the machine travels through the redex. Outside of the 
redex, however, the two sides have the same length, as the next example 
shows---example that also motivates a further extension of $\tohldb$ to 
contexts. Consider the case where $\tm\tohldb\tmtwo$ and the diagram is:
\begin{center}
	\begin{tikzpicture}[node distance=30mm, auto, transform 
	shape,scale=1.2]
	\node (p) at (0,0) 
	{$\ustate{\tm}{\blue{\ctxhole\tmthree}}{\exps}{\epsilon}$};
	\node (q) at (3,0) 
	{$\ustate{\tmtwo}{\blue{\ctxhole\tmthree}}{\exps}{\epsilon}$};
	\node (w) at (0,-.6) 
	{$\dstate\tmthree{\tm\ctxhole}\epsilon\exps$};
	\node (r) at (3,-.6) 
	{$\dstate\tmthree{\tmtwo\ctxhole}\epsilon\exps$};
	\node at (1.5,0) {$\relm$};
	\draw (p) edge[->] node {} (w);
	\draw (q) edge[->] node {} (r);
	\end{tikzpicture}
\end{center}
We then need to extend $\tohldb$ so that $\tm\ctxhole \tohldb \tmtwo \ctxhole$. A similar situation happens also when entering an ES with transition $\iamdvartwo$. To close these diagrams, we add two further cases of reduction on contexts. 
Of course, the same situation arises with $\ls$ and $\gcsym$ steps.
\begin{center}
$\begin{array}{ccc}
\infer{\tm\ctx \lolli_{\tt{a}} \tmtwo \ctx}{\tm \lolli_{\tt{a}} \tmtwo}
&
\infer{\tm\esub\var\ctx \lolli_{\tt{a}} \tmtwo\esub\var \ctx}{\tm \lolli_{\tt{a}} \tmtwo}
&
\tt{a} \in \set{\db,\ls,\gcsym}.
\end{array}$
\end{center}
\begin{SHORT}The full definition of $\relm$  is in the Appendix 
(\refsect{app-improvements}). Essentially, it is given by 
the following two base cases plus the expected extension to states (given by 
the analogous for $\relm$ of rules $\mathsf{tok1}_\ls$, $\mathsf{tok2}_\ls$, 
$\mathsf{tok3}_\ls$, $\mathsf{pos}_\ls$, $\mathsf{state}_\ls$ in 
\refdef{ls-improvement} below).
	\[\begin{array}{c@{\hspace{.8cm}}c}
	\infer[\mathsf{rdx}_\db]{(\tm,\hctx)\relm(\tmtwo,\hctx)}{\tm\tohldb\tmtwo}
	&
	\infer[\mathsf{ctx}_\db]{(\tm,\ctx)\relm(\tm,\ctxtwo)}{\ctx\tohldb\ctxtwo}
	\\[.17cm]
	\end{array}\]\end{SHORT}
\begin{LONG}
	\begin{definition}
		The (overloaded) binary relation $\relm$ between positions, stacks, and 
		states is defined by the following 
		rules\footnote{$\Gamma$ is a meta-variable that stands either for a log 
		$\ste$ or for a tape $T$.}.
		\[\begin{array}{c@{\hspace{.8cm}}c}
		\infer[\mathsf{rdx}_\db]{(\tm,\hctx)\relm(\tmtwo,\hctx)}{\tm\tohldb\tmtwo}
		&
		\infer[\mathsf{ctx}_\db]{(\tm,\ctx)\relm(\tm,\ctxtwo)}{\ctx\tohldb\ctxtwo}
		\\[.17cm]
		%
		\infer[\mathsf{tok1}_\db]{\epsilon\relm \epsilon}{}
		&
		\infer[\mathsf{tok2}_\db]{\resm\cdot\stme\relm 
			\resm\cdot\stme'}{\stme\relm\stme'}	
		\\[.17cm]
		\infer[\mathsf{pos}_\db]{(\var,\ctx,\ste)\relm
			(\var,\ctxtwo,\ste')}{(\var,\ctx)\relm(\var,\ctxtwo)\qquad\ste\relm\ste'}
		&
		\infer[\mathsf{tok3}_\db]{\exps\cdot\Gamma\relm
			\exps'\cdot\Gamma'}{\exps\relm\exps'\qquad\Gamma\relm\Gamma'}
		\\[.17cm]
		\multicolumn{2}{c}{
			\infer[\mathsf{state}_\db]{\nopolstate{\tm}{\ctx}{\stme}{\ste}{\pol}\relm
				\nopolstate{\tmtwo}{\ctxtwo}{\stme'}{\ste'}{\pol'}}{(\tm,\ctx)\relm
				(\tmtwo,\ctxtwo)\qquad\stme\relm\stme'\qquad\ste\relm\ste'\qquad
				\pol=\pol'}
		}
		\end{array}\]
		
	\end{definition}
	Note that	$\relm$ contains all pairs 
	$(\dstate{\tm}{\ctxhole}{\resm^k}{\epsilon},\dstate{\tmtwo}{\ctxhole}{\resm^k}{\epsilon})$,
	 where $\tm\tohldb\tmtwo$, \ie all the initial states containing a 
	$\db$-redex and its reduct.
\end{LONG}

\paragraph{Improvement for $\tohlls$.} As for $\tohldb$, the improvement for 
$\tohlls$ requires 
extending the rewriting relation to contexts. 
There are however some new subtleties. Given $\tm\tohldb\tmtwo$ and a position $(\tmthree,\ctx)$ for $\tm$, for $\relm$ the redex in $\tm$ falls always entirely either in $\tm$ or $\ctx$. If $\tm\tohlls\tmtwo$, instead, the redex can be split between the two. Consider the following diagram (where to simplify we assume the step to be at top level and the token to be empty).
\begin{center}
	\begin{tikzpicture}[node distance=30mm, auto, transform 
	shape,scale=1.2]
	\node (p) at (0,0) 
	{$\dstate{\hctxp{x}\esub\var\tmthree}\ctxhole\epsilon\epsilon$};
	\node (q) at (3,0) 
	{$\dstate{\hctxp{\tmthree}\esub\var\tmthree}\ctxhole\epsilon\epsilon$};
	\node (w) at (0,-.6) 
	{$\dstate{\hctxp{x}}{\ctxhole\esub\var\tmthree}\epsilon\epsilon$};
	\node (r) at (3,-.6) 
	{$\dstate{\hctxp{\tmthree}}{\ctxhole\esub\var\tmthree}\epsilon\epsilon$};
	\node at (1.5,0) {$\rele$};
	\draw (p) edge[->] node {} (w);
	\draw (q) edge[->] node {} (r);
	\end{tikzpicture}
\end{center}
To close it, we have to $\rele$-relate the two bottom states, where the pattern 
of the redex/reduct is split between the two parts of the position. This 
motivates clause $\mathsf{rdx2}$ in the definition of $\rele$ below.

The new rule comes with consequences. Consider the following diagram involving the new clause for $\rele$:
\begin{equation}
\label{eq:rele-diagram}
\begin{tikzpicture}[node distance=30mm, auto, transform 
shape,scale=1.2]
\node (p) at (0,0) {$\dstate x{\hctx\esub\var\tm}\epsilon\epsilon$};
\node (q) at (3,0) {$\dstate\tm{\hctx\esub\var\tm}\epsilon\epsilon$};
\node (w) at (0,-.6) {$\dstate\tm{\hctxp{x}\esub 
\var\ctxhole}\epsilon{(x,\hctx\esub\var\tm,\epsilon)}$};
\node at (1.5,0) {$\rele$};
\draw (p) edge[->] node {} (w);
\end{tikzpicture}
\end{equation}
To close the diagram, as usual, we have to $\rele$-relate them. There are, however, two delicate points. First, we cannot see the context $\hctxp{x}\esub \var\ctxhole$ as making a $\tohlls$ step towards $\hctx\esub\var\tm$, because $\tm$ does not occur in $\hctxp{x}\esub \var\ctxhole$. For that, we have to introduce a variant of $\tohlls$ on contexts that is parametric in $\tm$ (and more general than the one to deal with the showed simplified diagram):
\[
\begin{array}{rll}
\hctxp \var\esub\var\ctx\quad &\mapsto_{\ls,\tm}& \hctxp\ctx\esub\var{\ctxp\tm}.
\end{array}
\]
The second delicate point of diagram \refeq{rele-diagram} is that the extension 
of $\rele$ has to also $\rele$-relate logs of different length, namely 
$\epsilon$ and $(x,\hctx\esub\var\tm,\epsilon)$. This happens because positions 
of the two states do isolate the same term, but at different depths, as one is 
in the ES. Then the definition of $\rele$ has two clauses, one for logs 
($\mathsf{pos2}_\ls$) and one for states ($\mathsf{state2}_\ls$), to handle 
such a case. The mismatch in logs lengths is at most 1.

\begin{definition}
\label{def:ls-improvement}
	Binary relation $\rele$ is defined by\footnote{$\Gamma$ is a meta-variable that stands either for a log $\ste$ or for a tape $T$.}:
	\[\small\begin{array}{c@{\hspace{.4cm}}c @{\hspace{.4cm}}c}
	\infer[\mathsf{rdx}_\ls]{(\tm,\hctx)\rele(\tmtwo,\hctx)}{\tm\tohlls\tmtwo}
	&
	\infer[\mathsf{ctx}_\ls]{(\tm,\ctx)\rele(\tm,\ctxtwo)}{\ctx\tohlls\ctxtwo}
	\\[.17cm]
	
	\infer[\mathsf{rdx2}]{(\hctxp{\var},\hctxtwo)\rele(\hctxp{\tm},\hctxtwo)}
	{\hctxtwo=\hctxtwo'\ctxholep{\hctxthree[\var\leftarrow\tm]}}
	& 
		\infer[\mathsf{tok1}_\ls]{\epsilon\rele \epsilon}{}
	\\[.17cm]
	
	\infer[\mathsf{tok2}_\ls]{\resm\cdot\stme\rele 
		\resm\cdot\stme'}{\stme\rele\stme'}	
	&
	\infer[\mathsf{tok3}_\ls]{\exps\cdot\Gamma\rele
		\exps'\cdot\Gamma'}{\exps\rele\exps'\qquad\Gamma\rele\Gamma'}
	\\[.17cm]
	
		\infer[\mathsf{pos}_\ls]{(\var,\ctx,\ste\rele
			(\var,\ctxtwo,\ste')}{(\var,\ctx)\rele(\var,\ctxtwo)\qquad\ste\rele\ste'}
		&
		\infer[\mathsf{pos2}_\ls]{(\var,\ctx,\ste\cdot \exps)\rele 
			(\var,\ctxtwo,\ste')}{\ctx\tohllst{\var}\ctxtwo
			\qquad\ste\rele\ste'}		
	
	\\\\
	
	\multicolumn{3}{c}{
		\infer[\mathsf{state}_\ls]{\nopolstate{\tm}{\ctx}{\stme}{\ste}{\pol}\rele
			\nopolstate{\tmtwo}{\ctxtwo}{\stme'}{\ste'}{\pol'}}{(\tm,\ctx)\rele
			(\tmtwo,\ctxtwo)\qquad\stme\rele\stme'\qquad\ste\rele\ste'\qquad
			\pol=\pol'}
	}	 
	\\[.17cm]
	\multicolumn{3}{c}{
		\infer[\mathsf{state2}_\ls]{\nopolstate{\tm}{\ctx}{\stme}{\ste\cdot 
				\exps}{\pol}\rele\nopolstate{\tm}{\ctxtwo}{\stme'}{\ste'}{\pol'}}
		{\ctx\tohllst{\tm}\ctxtwo\qquad\stme\rele\stme'\qquad\ste\rele\ste'\qquad \pol=\pol'}
	}
	\end{array}\]
\end{definition}

Note that	$\rele$ contains all pairs 
$(\dstate{\tm}{\ctxhole}{\resm^k}{\epsilon},\dstate{\tmtwo}{\ctxhole}{\resm^k}{\epsilon})$,
 where $\tm\tohlls\tmtwo$, \ie all the initial states containing a $\ls$-redex 
and its reduct.\begin{SHORT} There also is a simpler relation $\relgc$ defined 
in the Appendix (\refsect{app-improvements}) for lack of space.\end{SHORT}

\begin{LONG}
	\paragraph{Improvement for $\tohlgc$} The candidate improvement 
	$\relgc$ induced by $\tohlgc$ requires an extension of $\tohlgc$ with a 
	rule on 
	contexts which is similar to the parametric one for $\tohlls$. Let 
	$\tm\esub\var\tmtwo \tohlgc \tm$ and consider:
	
	\begin{center}
		\begin{tikzpicture}[node distance=30mm, auto, transform 
		shape,scale=1.2]
		\node (p) at (0,0) 
		{$\dstate{\tm\esub\var\tmtwo}\ctxhole\epsilon\epsilon$};
		\node (q) at (3,0) {$\dstate{\tm}\ctxhole\epsilon\epsilon$};
		\node (w) at (0,-.6) 
		{$(\dstate{\tm}{\ctxhole\esub\var\tmtwo}\epsilon\epsilon$};
		\node at (1.5,0) {$\relgc$};
		\draw (p) edge[->] node {} (w);
		\end{tikzpicture}
	\end{center}
	To close the diagram, we extend the definition of $\togc$ to context with 
	the following parametric rule (closed by head contexts): 
	\[
	\begin{array}{r@{\hspace{.2cm}}l@{\hspace{.2cm}}l@{\hspace{.4cm}}l}
	\ctx\esub\var\tmtwo&\mapsto_{\gcsym,\tm}& \ctx & \mbox{if 
	$\var\notin\fv\tm$}.\\
	\end{array}
	\]
	We also need, as for $\relm$ and $\rele$, the rules (closed by head 
	contexts):
	\[
	\infer{\tm\ctx \tohlgc \tmtwo \ctx}{\tm\tohlgc\tmtwo}\qquad\qquad
	\infer{\tm\esub\var\ctx \tohlgc \tmtwo\esub\var \ctx}{\tm\tohlgc\tmtwo}
	\]

	\begin{definition}
		Binary relation $\relgc$ is defined by the following 
		rules.
		\[\begin{array}{c@{\hspace{.4cm}}c @{\hspace{.4cm}}c}
\infer[\mathsf{rdx}_\gcsym]{(\tm,\hctx)\relgc(\tmtwo,\hctx)}{\tm\tohlgc\tmtwo}
		&
\infer[\mathsf{ctx}_\gcsym]{(\tm,\ctx)\relgc(\tm,\ctxtwo)}{\ctx\tohlgc\ctxtwo}
		\\[.17cm]
		\infer[\mathsf{tok1}_\gcsym]{\epsilon\relgc \epsilon}{}
		&
\infer[\mathsf{ctx2}_\gcsym]{(\tm,\ctx)\relgc(\tm,\ctxtwo)}{\ctx\tohlgcv{\tm}\ctxtwo}
		\\[.17cm]
		%
		\infer[\mathsf{tok2}_\gcsym]{\resm\cdot\stme\relgc 
			\resm\cdot\stme'}{\stme\relgc\stme'}	
		&
		\infer[\mathsf{tok3}_\gcsym]{\exps\cdot\Gamma\relgc
			\exps'\cdot\Gamma'}{\exps\relgc\exps'\qquad\Gamma\relgc\Gamma'}
		\\[.17cm]
		\multicolumn{2}{c}{
			\begin{array}{c@{\hspace{.4cm}}c}	 
			\infer[\mathsf{pos}_\gcsym]{(\var,\ctx,\ste)\relgc		
(\var,\ctxtwo,\ste')}{(\var,\ctx)\relgc(\var,\ctxtwo)\qquad\ste\relgc\ste'}	
			\end{array}
		}
		\\\\
		\multicolumn{2}{c}{
\infer[\mathsf{state}_\gcsym]{\nopolstate{\tm}{\ctx}{\stme}{\ste}{\pol}\relgc	
\nopolstate{\tmtwo}{\ctxtwo}{\stme}{\ste}{\pol'}}{(\tm,\ctx)\relgc	
(\tmtwo,\ctxtwo)\qquad\stme\relgc\stme'\qquad\ste\relgc\ste'\qquad
				\pol=\pol'}
		}	 
		\end{array}\]
	\end{definition}
\end{LONG}

The proof of the next theorem is a tedious easy check of diagrams.
\begin{theorem}
	\label{thm:rele-improv}
	$\rele$, $\relm$ and $\relgc$ are improvements.
\end{theorem}

%% file: 10-Concrete_Correctness.tex
\section{Soundness and Adequacy, Proved}\label{sec:imp_thm}
Here we use the improvements of the previous sections to prove soundness and adequacy. Consider $\relf=\relm\cup\rele\cup\,\relgc
$, 
that is an improvement because its components are. Consequently, if 
$\tm\tohl\tmtwo$, then the $\IAM$ run on $\tmtwo$ improves the one on $\tm$, that is, 
$s_{\tm,k}=\dstate{\tm}{\ctxhole}{\resm^k}{\epsilon}
\ \ \relf\ \ 
\dstate{\tmtwo}{\ctxhole}{\resm^k}{\epsilon}=\state_{\tmtwo,k}$.

Improvements transfer more than termination/divergence along $\tohl$. They also give bisimilar, structurally equivalent tapes, proving the invariance of the semantics, that is, soundness.

\begin{theorem}[Soundness]
	\label{l:invariant-divergence}
	If $\tm\tohl\tmtwo$, then $\sem{\tm}{k}=\sem{\tmtwo}{k}$ for each $k\geq 0$.
\end{theorem}
\proof
	Since $\tm\tohl\tmtwo$, then 
	$s=\dstate{\tm}{\ctxhole}{\resm^k}{\epsilon}\relf
	\dstate{\tmtwo}{\ctxhole}{\resm^k}{\epsilon}=q$ by the results about 
	improvements (Theorem \ref{thm:rele-improv}). 
	Since improvements transfer termination/divergence (\refprop{termination}), we have $\sem{\tm}{k}=\bot$ iff $\sem{\tmtwo}{k}=\bot$. If $\sem{\tm}{k}\neq \bot$ let 
	$s'$ be the 
	terminal state of $s$. Since $\relf$ is an improvement, there is a 
	state $q'=\nopolstate{\tmthree'}{\ctx'}{\stme'}{\ste'}{\pol}$ such that 
	$q\toiam^* q'$ and $s'\relf q'$. Cases of $s'$:
	\begin{varitemize}
		\item $s'=\dstate{\la\var\tmfour}{\ctx}{\epsilon}{\ste}$. 
		Then, 
		since $s'\relf q'$, $\stme'=\epsilon$. 
		Moreover, either $\tm\tohl\tmtwo$, and thus 
		$\tmthree'=\la\var\tmfour'$ or $\ctx\tohl\ctxtwo$ and thus 
		$\tmthree=\tmthree'$. Then, 
		$\sem\tm{k}=\sem\tmtwo{k}=\,\Downarrow$.
		\item $s'=\ustate{\tm}{\ctxhole}{\resm^m\cdot 
			\exps\cdot\resm^l}{\epsilon}$. Then, since $s'\relf q'$, 
		$\ctx'=\ctxhole$, because the hole cannot $\tohl$-reduce. Moreover, 
		the structure of the tape is preserved by $\relf$ and thus 
		$\sem\tm{k}=\sem\tmtwo{k}=\langle m,l\rangle$.
		\item $s'=\dstate{\var}{\ctx}{\resm^m}{\ste}$. Since a variable 
		cannot $\tohl$-reduce, also $\tmthree'=\var$. Then 
		$\sem\tm{k}=\sem\tmtwo{k}=x$.\qed
	\end{varitemize}

\paragraph{Adequacy}
Adequacy is the fact that $\sem{\tm}{}$ is successful if and only if $\tohl$ terminates. We prove the two directions separately.

\paragraph{Direction $\IAM$ to $\tohl$.} The \emph{only if} direction of the 
statement is easy to prove. Since $\sem{\tm}{k}$ is invariant by $\tohl$ 
(soundness) and $\tohl$ terminates on $\tm$ we can as well assume that $\tm$ is 
normal. The rest is given by the following proposition.
\begin{proposition}[Reading the head variable on $\tohl$-normal 
forms]
	Let $\tm=\lambda\var_0\,...\,\lambda\var_n.\vartwo\tmtwo_1...\tmtwo_l$ be a 
	head linear normal form up to substitution. If $y=\var_m$ where 
	$0\leq m\leq 
	n$, then $\sem{\tm}{n+1}=\langle m,l\rangle$, otherwise, if $\vartwo$ is 
	free, then $\sem{\tm}{n+1}=\langle\vartwo,l\rangle$. Moreover, if $k\leq 
	n$ and $\tm$ is closed, then $\sem{\tm}{k}=\,\Downarrow$.
\end{proposition}
\begin{proof}
	We proceed computing $\sem{\tm}{n+1}$ explicitly. We have:
	\[
	\begin{array}{rcl}
	\dstate{\tm}\ctxhole{\resm^{n+1}}\epsilon&\toiam^{n+1}&\dstate{\vartwo\tmtwo_1...\tmtwo_l}{\lambda
	x_0\,...\,\lambda x_n.\ctxhole}\epsilon\epsilon)\\
	&\toiam^{l}&\dstate{\vartwo}{\lambda x_0\,...\,\lambda 
	x_n.\ctxhole\tmtwo_1...\tmtwo_l}{\resm^l}\epsilon
	\end{array}
	\]
	If $\vartwo$ is free, the $\IAM$ stops and 
	$\sem{\tm}{n+1}=\langle\vartwo,l\rangle$. 
	Otherwise, if $\vartwo$ is bound by a $\lambda$-abstraction, \ie 
	$\vartwo=\var_m$ for $0\leq m\leq n$, the computation continues.
	\[
	\begin{array}{lll}
	\dstate{\tm}\ctxhole{\resm^{n+1}}\epsilon&\toiam^{n+1+l}\\
	\dstate{\var_m}{\lambda
	x_0\,...\,\lambda 
	x_n.\ctxhole\tmtwo_1...\tmtwo_l}{\resm^l}\epsilon&\toiam\\
	\ustate{\lambda x_m\,...\,\lambda 
	x_n.x_m\tmtwo_1...\tmtwo_l}{\lambda x_0\,...\,\lambda 
		x_{m-1}.\ctxhole}\epsilon{\exps\cdot\resm^l}
	&\toiam^m
	\\
	\ustate\tmtwo{\ctxhole}{\resm^m\cdot \exps\cdot\resm^l}\epsilon.
	\end{array}
	\]
	where $\exps=(\var_m,\lambda x_0\,...\,\lambda 
	x_n.\ctxhole\tmtwo_1...\tmtwo_l,\epsilon)$. Then 
	$\sem\tm{k}=\langle m,l\rangle$.
	
	Moreover please note that if $k\leq n$ and $\tm$ is closed, we have:
	\[	
	\dstate{\tm}\ctxhole{\resm^{k}}\epsilon\toiam^{k}\dstate{\lambda\var_k\,...\,\lambda
		x_n.\vartwo\tmtwo_1...\tmtwo_l}{\lambda
	x_0\,...\,\lambda x_{k-1}.\ctxhole}\epsilon\epsilon
	\]
	and thus $\sem{\tm}{k}=\,\Downarrow$.
\end{proof}

\paragraph{Direction $\tohl$ to $\IAM$.}
The proof of the \emph{if} direction of the adequacy theorem is by contra-position: if the $\tohl$ diverges on $\tm$ then no run of the $\IAM$ on $\tm$ ends in a successful state. 

The proof is obtained via a quantitative analysis of the improvements, showing that the length of runs strictly decreases along $\tohl$. Note that improvements guarantee that the length of runs does not increase. To prove that it actually decreases one needs an additional \emph{global} analysis of runs---improvements only deal with \emph{local} bisimulation diagrams.
On proof nets, this decreasing property correspond to the standard fact that IAM paths passing through a cut have shorter residuals after that cut.

We recall that we write $|\tm|_k$ for the length of the run $\dstate{\tm}{\ctxhole}{\resm^k}{\epsilon}$, with the convention that $|\tm|_k= \infty$ if the machine diverges.

\begin{lemma}[The length of terminating runs strictly decreases along $\tohl$]
	\label{lemma:eff}
	Let $\tm\tohl\tmtwo$ and $|\tm|_k \neq \infty$. There exists $k\geq 0$ such that $|\tm|_h>|\tmtwo|_h$  for each $h\geq k$.
\end{lemma}

\begin{LONG}
\begin{proof}
	We treat the case of $\tm\tohldb\tmtwo$, the others are obtained via 
	similar diagrams. If $\tm$ has a $\tohldb$-redex then it has the shape 
	$\tm=\hctxp{\sctxp{\la\var\tmthree}\tmfour}$ and $\tmtwo$ is in the form 
	$\tmtwo=\hctxp{\sctxp{\tmthree\esub\var\tmfour}}$. By induction on the 
	structure of $\hctx$ one can prove that there exist $k,n\geq 0$ such that 
	$\dstate{\tm}\ctxhole{\resm^k}\epsilon\toiam^n 
	\dstate{\sctxp{\la\var\tmthree}\tmfour}\hctx\epsilon\epsilon)$
	and 
	$\dstate{\tmtwo}\ctxhole{\resm^k}\epsilon\toiam^n 
	(\dstate{\sctxp{\tmthree\esub\var\tmfour}}\hctx\epsilon\epsilon$.
	Given such $n$ and $k$ by the lifting lemma (\reflemma{pumping}) also the 
	following holds: for any $j\geq 0$, 
	$\dstate{\tm}\ctxhole{\resm^j\cdot\resm^k}\epsilon\toiam^n 
	\dstate{\sctxp{\la\var\tmthree}\tmfour}\hctx{\resm^j}\epsilon$
	and 
	$\dstate{\tmtwo}\ctxhole{\resm^j\cdot\resm^k}\epsilon\toiam^n 
	\dstate{\sctxp{\tmthree\esub\var\tmfour}}\hctx{\resm^j}\epsilon$.
	Moreover, by definition of the improvement $\relm$ we have the following 
	diagram.
	\begin{center}
		\begin{tikzpicture}[node distance=30mm, auto, transform 
		shape,scale=1.2]
		\node (u) at (0,.66) 
		{$\dstate{\hctxp{\sctxp{\la\var\tmthree}\tmfour}}\ctxhole{\resm^j\cdot\resm^k}\epsilon$};
		\node (i) at (3.8,.66) 
		{$\dstate{\hctxp{\sctxp{\tmthree\esub\var\tmfour}}}\ctxhole{\resm^j\cdot\resm^k}\epsilon$};
		\node (p) at (0,0) 
		{$\dstate{\sctxp{\la\var\tmthree}\tmfour}\hctx{\resm^j}\epsilon$};
		\node (q) at (3.8,0) 
		{$\dstate{\sctxp{\tmthree\esub\var\tmfour}}\hctx{\resm^j}\epsilon$};
		\node (y) at (3.8,-1) 
		{$\dstate{\tmthree\esub\var\tmfour}{\hctxp{\ctxhole{L}}}{\resm^j}\epsilon$};
		\node (w) at (0,-.66) 
		{$\dstate{\sctxp{\la\var\tmthree}}{\hctxp{\ctxhole\tmfour}}{\resm\cdot\resm^j}\epsilon$};
		\node (t) at (0,-1.33) 
		{$\dstate{\la\var\tmthree}{\hctxp{\sctx\tmfour}}{\resm\cdot\resm^j}\epsilon$};
		\node (e) at (0,-2) 
		{$s_1=\dstate{\tmthree}{\hctxp{\sctxp{\la\var\ctxhole}\tmfour}}{\resm^j}\epsilon$};
		\node (r) at (3.8,-2) 
		{$\dstate{\tmthree}{\hctxp{\sctxp{\ctxhole\esub\var\tmfour}}}{\resm^j}\epsilon=s_2$};
		\node at (1.9,0) {$\relm$};
		\node at (1.9,-2) {$\relm$};
		\node at (1.9,.66) {$\relm$};
		\draw (p) edge[->] node {} (w);
		\draw (w) edge[->] node {$^{|{\sctx}|}$} (t);
		\draw (u) edge[->] node {$^n$} (p);
		\draw (i) edge[->] node {$^n$} (q);
		\draw (t) edge[->] node {} (e);
		\draw (q) edge[->] node {$^{|{\sctx}|}$} (y);
		\draw (y) edge[->] node {} (r);
		\end{tikzpicture}
	\end{center}
	From $s_1 \relm s_2$, the hypothesis $|\tm|_k \neq \infty$, and the 
	properties of improvements (Lemma~\ref{lemma:improv}), we obtain 
	$|s_1|\geq|s_2|$. Then, by setting $h\defeq k+j$, we have	
	$|\tm|_h=n+1+|\sctx|+1+|s_1|>n+|\sctx|+1+|s_2|=|\tmtwo|_h$.
\end{proof}
\end{LONG}

\begin{SHORT}The proof of the lemma is in the Appendix.\end{SHORT} Using the 
lemma, we prove the \emph{if} direction of  adequacy, that then follows.

\begin{proposition}[$\tohl$-divergence implies that the $\IAM$ never succeeds]
	Let $\tm$ be a $\tohl$-divergent LSC term. There is no $k\geq 0$ such that $\sem\tm{k}$ is successful.
\end{proposition}
\begin{proof}
	By contradiction, suppose that there exists $k$ such that $\sem\tm{k}$ is successful. Then by soundness $|\tm|_k =n \in\nat$ and it ends on a successful state. By monotonicity of runs (\reflemma{runs-monotonicity}), $|\tm|_k = |\tm|_h = n$ for every $h>k$.
	Since $\tm$	$\tohl$-divergent, then there exists an infinite reduction sequence $\rho:\tm=\tm_0\tohl\tm_1\tohl\tm_2\tohl\cdots\tm_k\tohl\cdots$. Since the length of terminating runs strictly decreases along $\tohl$ for sufficiently long inputs (Lemma~\ref{lemma:eff}), for each $i\in\nat$ if $\tm_i \tohl \tm_{i+1}$ then there exists $k_i$ such that $|\tm_i|_{k_i} > |\tm_{i+1}|_{k_i}$. Now, consider $h = \max\set{k_0,\ldots, k_n, k_{n+1}}$. We have that $|\tm_j|_h > |\tm_{j+1}|_h$ for every $j \in \set{0,1,\ldots, n+1}$. Then $|\tm_0|_h \geq |\tm_{n+1}|_h + n+1$. Since the length of runs is non-negative, we obtain that $|\tm_0|_h \geq n+1$, which is absurd because $h\geq k_0$ and so $|\tm_0|_h = n$.
\end{proof}

\begin{theorem}[Adequacy]
	Let $\tm$ be a LSC term. Then $\tm$ has $\tohl$-normal form if and only if 
	there exists $k>0$ such that either $\sem{\tm}{k}=\langle m,l\rangle$ for 
	some $m,l\geq 0$ or $\sem{\tm}{k}=\langle\var,n\rangle$ for some 
	$\var\in\mathcal{V},n\geq 0$. Moreover, when $\tm$ is closed, $\tm$ has 
	weak head normal form if 
	and only if $\,\sem{\tm}{0}=\,\Downarrow$.
\end{theorem}

\paragraph{Terminating without Ever Succeeding.} It is possible that $\tohl$ diverges on $\tm$ and all the runs of the $\IAM$ terminate on $\tm$ without ever succeeding. The idea is that the  
$\IAM$ performs a fine analysis of the $\tohl$ evaluation of $\tm$, approximating $\tohl$ while incrementally building the \levy-Longo tree of $\tm$. On a looping term such as 
$\mathbf{\Omega}$ the $\IAM$ does diverge. On a non-terminating term such as
$\mathbf\Lambda=(\la\var\la\vartwo\var\var)(\la\var\la\vartwo\var\var)$, instead, 
the $\IAM$ does \emph{not} diverge, it gives
$\sem{\mathbf\Lambda}{k}=\,\Downarrow$ for each $k\geq 0$. Note in fact that 
$\mathbf{\Lambda}\toh\la\var\mathbf{\Lambda}$, \ie $\mathbf\Lambda$ has 
an infinite number of abstractions in its limit normal form and thus an 
infinite number of $\resm$ on the tape would be needed to inspect them all. 
On the contrary, $\mathbf{\Omega}$ has no head lambdas in its limit normal 
form and thus $\sem{\mathbf{\Omega}}{k}=\bot$ for each $k\geq 0$.

%% file: 11-Proof-nets_Comparison.tex
\section{Comparison with the Original Proof Net Presentation}\label{sect:pn}
\begin{SHORT}\begin{figure}[t]
	\begin{center}
		\begin{tabular}{c||c}
			\input{explicit-translation}
			&\input{mell-links}
		\end{tabular}
		\vspace{-8pt}
		\caption{On the left, the call-by-name translation $(\cdot)^\dagger$ of 
		the $\lambda$-calculus into linear 
			logic proof nets.\\ On the right, transition rules of the proof nets presentation of the IAM related to exponential 
			signatures.}
		\label{fig:proofnets}
	\end{center}
	
\end{figure}\end{SHORT}
\begin{LONG}\begin{figure}[t]
		\begin{center}
			\begin{tabular}{c}
				\input{explicit-translation}\\[15pt]
				\hhline{=}\\
				\input{mell-links}
			\end{tabular}
			\vspace{-8pt}
			\caption{Above, the call-by-name translation $(\cdot)^\dagger$ of 
				the $\lambda$-calculus into linear 
				logic proof nets.\\ Below, transition rules of the proof nets 
				presentation of the IAM related to exponential 
				signatures.}
			\label{fig:proofnets}
		\end{center}
		
	\end{figure}
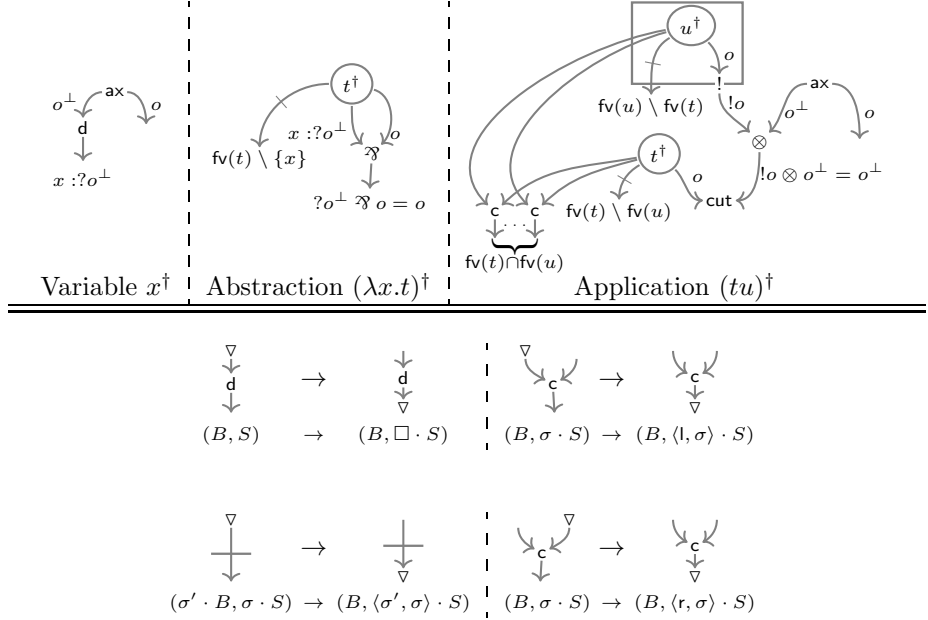\end{LONG}
Here we sketch how the $\IAM$ relates to the original presentation 
based on linear 
logic proof nets, due to Mackie and Danos \& Regnier~\cite{mackie_geometry_1995,DBLP:conf/lics/DanosHR96,danos_reversible_1999}, the IAM. For lack of space, we avoid 
defining proof nets and related concepts, and focus only on the key points.

Essentially, the $\IAM$ corresponds to the IAM on proof nets representing $\l$-terms according 
to the call-by-name translation $\tm^\dagger$ in 
Fig.~\ref{fig:proofnets}\footnote{The translation uses a recursive type $\mtype 
= ?\mtype^\bot \parr \mtype$ in order to be able to represent untyped terms of 
the $\l$-calculus---this is standard. Every net has a unique conclusion labeled 
with $\mtype$, which is the \emph{output}, and all the other conclusions have 
type $?\mtype^\bot$ and are labelled with a free variable of the term. In the 
abstraction case $\la\var\tm$, if $\var\notin\fv\tm$ then a weakening is added 
to represent that variable.}, and 
considering only paths from the distinguished conclusion of the obtained net, as in 
\cite{DBLP:conf/lics/DanosHR96} (while 
\cite{mackie_geometry_1995,danos_reversible_1999} use the call-by-value translation, and \cite{danos_reversible_1999} considers paths starting on whatever conclusions).

There is a bisimulation between the $\IAM$ and such a restricted IAM, which is not strong 
because two $\IAM$ transitions rather are \emph{macros}, packing together whole 
sequences of transitions in their presentation. Namely, transition $\iamdvar$ 
short-circuits the path between a variable $\var$ and its abstraction 
$\la\var\ctx_n\ctxholep{\var}$. In proof nets, this path 
traverses a dereliction, exactly $n$ auxiliary doors, possibly a contraction 
tree, and ends on the $\parr$ representing the abstraction. The dual 
transition $\iamdlamtwo$ does the reverse job, corresponding to the reversed 
path. Aside the different notations and the \emph{macrification}, our 
transitions
correspond exactly to the actions attached to proof nets edges presented 
in~\cite{danos_reversible_1999}\footnote{We refer to 
\cite{danos_reversible_1999} rather than \cite{DBLP:conf/lics/DanosHR96} 
because in \cite{DBLP:conf/lics/DanosHR96} the definition is only sketched, 
while \cite{danos_reversible_1999} is more accurate.}, as we explain next.

In the proof nets presentation the token is given by two stacks, called \emph{boxes stack 
B} and \emph{balancing stack S}, corresponding exactly to our log $\ste$ 
and tape $\stme$, respectively. They are formed by sequences of multiplicative constants $\psym$ (corresponding to our $\resm$) and by \emph{exponential signatures} $\sigma$. They are defined by the following 
grammar\footnote{With respect to \cite{danos_reversible_1999}: for clarity, we 
use symbols $\lsym$ and $\rsym$ instead of $\psym'$ and $\qsym'$, and we omit 
$\qsym$, dual of $\psym$, as the use of the 
call-by-name translation allows to get rid of it.}.
\[
\begin{array}{rrcl}
\textsc{Balancing stacks} & S&\grameq&\epsilon\midd\psym\cdot S\midd\sigma\cdot S\\
\textsc{Boxes stacks} & B&\grameq&\epsilon\midd\sigma\cdot B\\
\textsc{Exp. signatures} & 
\sigma,\sigma'&\grameq&\Box\midd\langle\sigma,\sigma'\rangle\midd\langle\lsym,\sigma\rangle\midd
\langle\rsym,\sigma\rangle
\end{array}
\]
Intuitively, exponential signatures are binary trees with $\Box$, $\lsym$ or $\rsym$ as leaves, 
where $\lsym$ and $\rsym$ denote the left/right premise of a contraction. 
Fig.~\ref{fig:proofnets} shows the IAM
transitions concerning exponential signatures that are the relevant difference with respect to 
the $\IAM$.

To explain how $\iamdvar$ is simulated by the IAM, let's recall it:
\[
\dstate{ \var }{ \ctxp{\la\var\ctxtwo_n} }{ \stme }{ \expsn\cdot\ste } 
\iamdvar 
\ustate{ \la\var\ctxtwo_n\ctxholep\var}{ \ctx }{ 
	(\var,\la\var\ctxtwo_n,\expsn)\cdot\stme }{ \ste }.
\]
The IAM does the same, just in more steps and with another syntax. Consider a 
token $(B_n\cdot B,S)$ 
approaching a variable $\var$ that is $n$ 
boxes deeper than its binder $\la\var\ctxtwo_n\ctxholep\var$. Variables are translated as 
dereliction links and thus we have: $(B_n\cdot B,S)\to(B_n\cdot B,\Box\cdot S)$.

Then, the token travels until the binder of $\var$ is found (a $\parr$ in the proof net translation 
of the term), \ie it traverses exactly $n$ 
boxes always exiting from the auxiliary doors. Moreover, for every such box a 
contraction could be encountered. Let us suppose for the moment that $\var$ is 
used linearly, so that no contractions are encountered. Then the token rewrites 
in the following way, traversing $n$ auxiliary doors.
\[
\begin{array}{rcl}
(\sigma_1\cdot B_{n-1}\cdot B, \Box\cdot S)&\to&
(\sigma_2\cdot B_{n-2}\cdot B,\langle \sigma_1,\Box\rangle\cdot S)\\
&\to&
(\sigma_3\cdot B_{n-3}\cdot B, \langle \sigma_2,\langle 
\sigma_1,\Box\rangle\rangle\cdot S)\\
&\to&\cdots\to(B, \langle \sigma_n,\langle\cdots\langle \sigma_1,\Box\rangle
\cdots\rangle\rangle\cdot S).
\end{array}
\]
Note the perfect matching between the two formulations: in both cases 
the first $n$ \trposs/signatures in the log/boxes stack are removed from it 
and, once wrapped in a 
single \trpos/signature, then put on the tape/balancing stack. In presence of 
contractions the exponential signature $\langle \sigma_n,\langle\cdots\langle \sigma_1,\Box\rangle
\cdots\rangle\rangle$ is interleaved by $\lsym$ and $\rsym$ leaves. 
These symbols represent nothing more than a binary code used to traverse the 
contraction tree of $\var$. In the $\IAM$, we use a more human readable way of 
representing the same information: we explicitly save the variable occurrence 
through its position inside its binder. 

%% file: explicit-translation.tex
\begin{tabular}{c:c:c}

\begin{tikzpicture}[ocenter]
\node at (0,0) [etic](axRightConclusion){};
\node at (axRightConclusion.center) [etic, above right=6pt and 1pt](dummy){\scriptsize $\mtype$};
\node at (axRightConclusion.center) [etic, left = 1.4*\stlar](der){\scriptsize $\der$};
\node at (der.center) [etic, above left=6pt and 1pt](dummy){\scriptsize $\mtype^\bot$};
\node at \med{axRightConclusion}{der} [etic, above = \hstalt ](axSym){\scriptsize $\ax$};
\draw[nopol, out=0, in=90](axSym)to(axRightConclusion);
\draw[nopol, out=180, in=90](axSym)to(der);
\node at (der.center) [etic, below= 1.2*\hstalt] (derGhost){};
\draw[nopol](der)to(derGhost);
\node at (derGhost.center) [etic, below = 1pt] (derGhostType){\scriptsize $\var:?\mtype^\bot$};
\end{tikzpicture}
&

\begin{tikzpicture}[ocenter]  
\node [net] (proofel) {$\pntransl\tm$};

\node at (proofel.center) [below right = 1.2*\stalt and \hstlar, etic](par){$\parr$};
\node at (par.center) [etic, above left=5pt and 8pt](dummy){\scriptsize $\var:?\mtype^\bot$};
\node at (par.center) [etic, above right=5pt and 6pt](dummy2){\scriptsize $\mtype$};

\draw[nopol, out=-90, in=135](proofel)to(par);
\draw[nopol, out=-15, in=45](proofel)to(par);
\node at (par.center) [below left = .8*\stalt and 1pt, etic](parghost){\scriptsize $?\mtype^\bot \parr \mtype = \mtype$};
\draw[nopol](par)to(parghost);

\node at (proofel.center) [below left = 1.3*\stalt and 2.3*\stlar, etic](Gammanode){\scriptsize $\fv\tm\setminus\set\var$}; 
\draw[nopolgen, in=90, out=180] (proofel)to(Gammanode);
\end{tikzpicture}

&

\begin{tikzpicture}[ocenter]  
\node [net] (leftterm) {$\pntransl\tm$};
\node at (leftterm.center) [below left = \stalt and \stlar, etic](Gammanode){\scriptsize $\fv\tm \setminus\fv\tmtwo$}; 
\draw[nopolgen, in=90, out=210] (leftterm)to(Gammanode);

\node at (leftterm.center)[below right = .8*\stalt and 1.5*\stlar,etic](cut){\scriptsize $\cut$};
\draw[nopol, out=-45, in=180] (leftterm)to(cut);
\node at (cut.center) [above left = .3*\stalt and .4*\stlar, nospace] (dummy3) {\scriptsize$\mtype$};
\node at (cut.center) [above right = .3*\stalt and \stlar, nospace] (dummy3) {\scriptsize$!\mtype\tens\mtype^\bot = \mtype^\bot$};

\node at (cut.center)[above right = \stalt and \stlar,etic](tens){\scriptsize $\tens$};
\node at (tens.center) [above right = .4*\stalt and .6*\stlar, nospace] (dummy) {\scriptsize$\mtype^\bot$};
\draw[nopol, out=-90, in=0] (tens)to(cut);

\node at (tens.center)[above right = \stalt and 1.5*\stlar,etic](ax){\scriptsize $\ax$};
\node at (ax.center)[below right = \stalt and \stlar,etic](axghostconclusion){};
\node at (axghostconclusion.center) [above right = .4*\stalt and .1*\stlar, nospace] (dummy) {\scriptsize$\mtype$};
\draw[nopol, out=180, in=45] (ax)to(tens);
\draw[nopol, out=0, in=90] (ax)to(axghostconclusion);

\node at (tens.center)[above left = \stalt and \stlar,etic](bang){\scriptsize $!$};
\draw[nopol, out=-90, in=135] (bang)to(tens);

\node at (bang.center) [above left = \stalt and .7*\stlar, net] (rightterm) {$\pntransl\tmtwo$};
\draw[nopol, out=-45, in=90] (rightterm)to(bang);
\node at (bang.center) [above right = .4*\stalt and .1*\stlar, nospace] (dummy2) {\scriptsize$\mtype$};
\node at (bang.center) [below right = .2*\stalt and .2*\stlar, nospace] (dummy2) {\scriptsize$!\mtype$};
\abox{bang}{exbox}{32pt}{8pt}{30pt}
\node at (bang.center)[etic] (bangsym){$!$};

\node at (rightterm.center) [below left = 1.4*\stalt and \stlar, etic](Deltanode){\scriptsize $\fv\tmtwo \setminus\fv\tm$}; 
\draw[nopolgen, in=90, out=210] (rightterm)to(Deltanode);

\node at (Gammanode.center) [left = 2.1*\stlar, etic](contr1){\scriptsize $\contr$}; 
\node at (contr1.center) [below = \hstalt, etic](contr1ghost){}; 
\draw[nopol] (contr1)to(contr1ghost);
\draw[nopol, out=195, in = 135] (rightterm)to(contr1);
\draw[nopol, out=195, in = 45] (leftterm)to(contr1);

\node at (contr1.center) [left = \stlar, etic](contr2){\scriptsize $\contr$}; 
\node at (contr2.center) [below = \hstalt, etic](contr2ghost){}; 
\draw[nopol] (contr2)to(contr2ghost);
\draw[nopol, out=190, in = 135, overlay] (rightterm)to(contr2);
\draw[nopol, out=190, in = 45] (leftterm)to(contr2);

\gdotsname{contr1}{contr2}{below = 1pt}{puntini}
\node at \med{contr1ghost}{contr2ghost}[below=\sepbox, rotate=90, nospace, anchor = center](bracket){\LARGE \{};
 \node at (bracket.center)[ below=1.3*\sepbox, nospace](bracket2){\scriptsize $\fv\tm{\cap}\fv\tmtwo$};
\end{tikzpicture}
\\
Variable $\pntransl{\var}$
&
Abstraction $\pntransl{(\la\var\tm)}$
&
Application $\pntransl{(\tm\tmtwo)}$

\end{tabular}

%% file: mell-links.tex
\begin{tabular}{c@{\hspace{.1cm}}c@{\hspace{.1cm}}c:c@{\hspace{.1cm}}c@{\hspace{.1cm}}c}
\begin{tikzpicture}[ocenter]
\node at (0,0) [etic] (derPal){};
\node at (derPal.center) [etic, above= \stalt] (derPax){$\triangledown$};

\lder{derPal}{derPax}{der};
\end{tikzpicture}
 &
 $\to$
 &
 \begin{tikzpicture}[ocenter]
\node at (0,0) [etic] (derPal){$\triangledown$};
\node at (derPal.center) [etic, above= \stalt] (derPax){};

\lder{derPal}{derPax}{der};
\end{tikzpicture}
&
\begin{tikzpicture}[ocenter]
\node at (0,0)[etic] (contrPal){};
\node at (contrPal.center) [etic,above left = \stalt and \hstlar](contrLPax){$\triangledown$};
\node at (contrPal.center) [etic,above right = \stalt and \hstlar](contrRPax){};

\lcontr{contrPal}{contrLPax}{contrRPax}{contr};
\end{tikzpicture}
&
$\to$
&
\begin{tikzpicture}[ocenter]
\node at (0,0)[etic] (contrPal){$\triangledown$};
\node at (contrPal.center) [etic,above left = \stalt and \hstlar](contrLPax){};
\node at (contrPal.center) [etic,above right = \stalt and \hstlar](contrRPax){};

\lcontr{contrPal}{contrLPax}{contrRPax}{contr};
\end{tikzpicture}

\\

	\scriptsize $(B, S)$
&
\scriptsize $\to$
&
\scriptsize $(B,\Box\cdot S)$
&
\scriptsize
 $(B,\sigma\cdot S)$
 &
\scriptsize $\to$
 &
\scriptsize $(B,\langle\mathsf{l},\sigma\rangle\cdot	S)$
\\\\\\
\begin{tikzpicture}[ocenter]
\node at (0,0) [etic] (derPal){};
\node at (derPal.center) [etic, above= \stalt] (derPax){$\triangledown$};
\node at (derPal.center) [etic, above right= \hstalt and \hstlar] (ghostR){};
\node at (derPal.center) [etic, above left= \hstalt and \hstlar] (ghostL){};
\draw[nopol] (derPax) to (derPal);
\draw[exboxline](ghostL) to (ghostR);
\end{tikzpicture}
 &
 $\to$
 &
\begin{tikzpicture}[ocenter]
\node at (0,0) [etic] (derPal){$\triangledown$};
\node at (derPal.center) [etic, above= \stalt] (derPax){};
\node at (derPal.center) [etic, above right= \hstalt and \hstlar] (ghostR){};
\node at (derPal.center) [etic, above left= \hstalt and \hstlar] (ghostL){};
\draw[nopol] (derPax) to (derPal);
\draw[exboxline](ghostL) to (ghostR);
\end{tikzpicture}
&
\begin{tikzpicture}[ocenter]
\node at (0,0)[etic] (contrPal){};
\node at (contrPal.center) [etic,above left = \stalt and \hstlar](contrLPax){};
\node at (contrPal.center) [etic,above right = \stalt and \hstlar](contrRPax){$\triangledown$};

\lcontr{contrPal}{contrLPax}{contrRPax}{contr};
\end{tikzpicture}
&
$\to$
&
\begin{tikzpicture}[ocenter]
\node at (0,0)[etic] (contrPal){$\triangledown$};
\node at (contrPal.center) [etic,above left = \stalt and \hstlar](contrLPax){};
\node at (contrPal.center) [etic,above right = \stalt and \hstlar](contrRPax){};

\lcontr{contrPal}{contrLPax}{contrRPax}{contr};
\end{tikzpicture}
\\
\scriptsize $(\sigma'\cdot B,\sigma\cdot S)$
&
\scriptsize $\to$
&
\scriptsize $(B,\langle\sigma',\sigma\rangle\cdot S)$
&
\scriptsize $(B,\sigma\cdot S)$
 &
\scriptsize $\to$
 &
\scriptsize $(B,\langle\mathsf{r},\sigma\rangle\cdot	S)$

\end{tabular}

%% file: 12-Conclusion.tex
\section{Conclusions}
This paper presents a direct proof of the implementation theorem for Mackie and Danos \& Regnier's Interaction Abstract Machine, building over a natural notion of bisimulation and avoiding detours via game semantics. Additionally, it (re)formulates the machine directly on $\lambda$-terms, 
making it conceptually closer to traditional abstract machines, and more apt to formalizations in proof assistants.

Our work opens the way to a fine analysis of the complexity of the implementation of the $\lambda$-calculus, in particular regarding the space-time trade-off, by comparing the $\IAM$, that the literature suggests being tuned for space-efficiency, to traditional environment machines that are instead tuned for time-efficiency.